\documentclass[10pt]{article}
\usepackage[a4paper]{geometry}
\usepackage{amsmath}
\usepackage{comment}
\excludecomment{remove}

\usepackage{xcolor}
\usepackage{listings}
\usepackage{courier}
\usepackage{hyperref} 
\usepackage{cleveref}
\usepackage{amsthm}

\makeatletter
\lst@AddToHook{Init}{\def\@currentcounter{lstnumber}}%
\makeatother

\newtheorem{theorem}{Theorem}[section]
\newtheorem{lemma}[theorem]{Lemma}
\newtheorem{corollary}[theorem]{Corollary}

\AddToHook{env/lemma/begin}{\crefalias{theorem}{lemma}}
\AddToHook{env/corollary/begin}{\crefalias{theorem}{corollary}}
\AddToHook{env/invariant/begin}{\crefalias{theorem}{invariant}}
\crefname{theorem}{theorem}{theorems}
\crefname{lemma}{lemma}{lemmas}
\crefname{corollary}{corollary}{corollaries}
\crefname{invariant}{invariant}{invariants}

\newcommand{\ignore}[1]{}
\newcommand{\codestyle}[1]{{{\fontsize{9}{11}\ttfamily #1}}}

\newcommand{\faith}[1]{{\color{cyan} \textbf{Faith}: #1}}
\newcommand{\galy}[1]{{\color{red} \textbf{Galy}: #1}}

\title{Strong Linearizability without Compare\&Swap: \\ The Case of Bags}

\author{Faith Ellen \\ University of Toronto, Toronto, Canada \\ faith@cs.toronto.edu
\and
Gal Sela \\ EPFL, Lausanne, Switzerland \\ gal.sela@epfl.ch}
\date{}

\begin{document}

\maketitle

\begin{abstract}
Because strongly-linearizable objects provide stronger guarantees than linearizability,
they serve as
valuable building blocks for the design of concurrent data structures. Yet,  many objects that have linearizable implementations from 
base objects
weaker than compare\&swap objects
do not have strongly-linearizable implementations from 
the same 
base objects. We focus on one such object: 
the \emph{bag}, a multiset from which processes can take 
unspecified
elements.

We present the first lock-free, strongly-linearizable implementation of a bag from 
interfering objects (specifically, registers, and test\&set objects).
This may be surprising, since
there are provably no such implementations
of stacks or queues.

Since a bag can contain arbitrarily many elements, an unbounded amount of space must be used to implement it.
Hence, it makes sense to also consider
a bag with a bound on its capacity.
However, like stacks and queues, a bag with capacity $b$ shared by more than $2b$ processes
has no lock-free, strongly-linearizable implementation from interfering objects.
If we further restrict a bounded bag
so that only one process can insert into it, we are able to obtain 
a lock-free, strongly-linearizable implementation from 
$O(b+n)$ interfering objects, where $n$ is the number of processes.

Our goal is to understand the circumstances under which strongly-linearizable implementations of bags exist 
and, more generally, to understand the power of interfering objects.
\end{abstract}


\medskip

The conference version of this paper is available at \cite{ellen2025strong}.

\newpage

\section{Introduction}\label{sec:intro}

Concurrent data structures are often built using linearizable implementations of objects as if they were atomic objects. Although linearizability is the 
usual
correctness condition for concurrent algorithms, there are scenarios where composing linearizable implementations does not preserve desired properties. 
Namely, an algorithm that uses atomic objects may lose some of its properties when these atomic objects are replaced by linearizable implementations. 
Specifically, linearizability does not preserve hyperproperties (properties of sets of executions), for example, probability distributions of events for randomized programs and security properties such as noninterference \cite{golab2011linearizable,attiya2019putting}.
This can be rectified by using strongly-linearizable implementations \cite{golab2011linearizable}.
They guarantee
that the linearization of a prefix of a concurrent execution is a prefix of the linearization of the whole execution. This means that the linearization of operations in an execution prefix cannot depend on later events in the execution.

Attiya, Castaneda, and Enea \cite{attiya2024strong}
observed that
strongly-linearizable implementations of objects typically use
compare\&swap 
objects, which can be used to solve consensus among any number of processes.
They provided strongly-linearizable implementations of some objects,
such as a wait-free single-writer snapshot object and a wait-free max register from a fetch\&add object,
and a lock-free, readable, resettable test\&set object 
and a lock-free, readable, fetch\&in\-crement object 
from an infinite array of test\&set objects and registers.
Both fetch\&add and test\&set objects are less powerful than a compare\&swap object.

In this work, we aim to further explore the power of these well-studied building blocks in the context of strong linearizability. 
Specifically, we investigate strongly-linearizable implementations of a bag from these primitives.
Establishing the existence of such implementations deepens our understanding of the circumstances under which these primitives can be used to achieve strong linearizability.

A \textit{bag} 
is a multiset to which processes can insert elements and from which they can take unspecified elements. 
A concurrent bag is an abstraction of the interaction between producers and consumers
and is commonly used to distribute tasks among processes for load balancing and improved scalability.

Queues and stacks have strongly-linearizable implementations from 
compare\&swap 
objects \cite{hwang2022strongly,treiber1986systems,michael1998nonblocking}.
A bag has a straightforward implementation from a stack or a queue.
Therefore, a bag has a strongly-linearizable implementation from 
compare\&swap 
objects.
Although there is a wait-free, linearizable implementation of a stack
from registers, test\&set objects, and a readable fetch\&increment object~\cite{afek2007common2}
and there is a lock-free, linearizable implementation of a queue from the same set of objects~\cite{Li01},
neither of these implementations is strongly-linearizable.
In fact, Attiya, Casta\~neda, and Hendler~\cite{ACH18} proved that no lock-free, strongly-linearizable implementation of a stack or queue from such objects is possible.

\ignore{
Attiya, Casta\~neda, and Enea~\cite{attiya2024strong} claimed that the 
lock-free linearizable implementation of a queue
which appears in \Cref{fig: nonSL bag}, 
is a strongly-linearizable implementation of a bag.
We give a counterexample to this claim in \Cref{sec:non-SL bag}.
This was the starting point for our research.
}

The challenge when implementing a strongly-linearizable bag 
from these objects is
for
an operation that is  trying to take an element from the bag to detect when the bag is empty.
If elements can reside in multiple locations, it does not suffice for the operation to simply examine these locations one by one, concluding that the bag is empty if no elements are found. The problem is that new elements may have
been added to locations after they were examined.
Even if there was a configuration 
during the operation
in which the bag was empty, it might not be possible
to linearize the operation at this point while maintaining strong linearizability.
For example, after this point, in an alternative execution, a value could be
added to a location that the operation had not yet examined
and the operation could 
return that value, instead of EMPTY.

We address this challenge in \Cref{sec:LF SL bag}, modifying the lock-free, linearizable implementation of a queue by adding an additional readable fetch\&increment object. This object
is used by
operations which are inserting elements into the bag to inform operations which
are trying to take elements from the bag that the bag is no longer empty.
We prove that the resulting algorithm is a strongly-linearizable implementation of a bag.
Interestingly, although our algorithm is still a linearizable implementation of a queue, it is not a strongly-linearizable implementation of a queue.

A bag can grow arbitrarily large, 
so either the number of objects used to implement it or the size of the objects must be unbounded.
An alternative is a 
\emph{$b$-bounded bag}, which can simultaneously contain up to $b$ elements.
Unfortunately, 
a $b$-bounded bag
shared by more than $2b$ processes has
no lock-free, strongly-linearizable implementation from interfering objects.
However, we 
provide a lock-free, strongly-linearizable implementation of
a restricted version of a $b$-bounded bag, into which only one process, the \emph{producer}, can insert elements.
It uses
a bounded number of registers, each of bounded size, and readable, resettable test\&set objects.
Before presenting this implementation, we give two different implementations of a 1-bounded bag with a single producer, which introduce some of the ideas that appear in our $b$-bounded bag implementation.

In \Cref{sec:WF 1-bounded bag}, we present a wait-free, linearizable implementation of a 1-bounded bag.
Note that it provides a stronger progress guarantee.
To keep the space bounded, when the producer inserts 
an element into the bag, it may reuse
objects previously used for other elements. It has to do this carefully, to avoid reusing objects that other processes may be poised to access. This requires delicate coordination between the producer and the consumers.

In \Cref{sec:1-bounded bag}, we present a lock-free, strongly-linearizable implementation of a 1-bounded bag. 
It combines the mechanism for detecting an empty bag 
from our lock-free, strongly-linearizable implementation of a bag and 
the mechanism for reusing objects from our wait-free, linearizable implementation of a 1-bounded bag.


Our lock-free, strongly-linearizable implementation of a $b$-bounded bag, for any 
value of $b$,
is presented in \Cref{section: b-bounded}. 
To enable the producer to detect when the bag is full, we use an approach that
is similar to detecting when the bag is empty, but we have
to ensure that these two mechanisms do not interfere with one another.
In our implementation of a 1-bounded bag in \Cref{sec:1-bounded bag}, the producer determines whether the bag is empty or full by inspecting a single object.
This makes this implementation both simpler and more efficient than the special case of our implementation of a $b$-bounded bag when $b= 1$.

\ignore{
An ABA detecting register has a strongly-linearizable implementation from bounded multi-writer registers \cite{aghazadeh2015time}.
Thus, our algorithm implies that there exists a wait-free linearizable implementation for a 1-bounded bag with a single producer from resettable readable test\&set objects and bounded multi-writer registers.
}

\section{Preliminaries}\label{sec:preliminaries}

We consider an asynchronous system where processes communicate through operations
applied to shared objects.
Each type of object supports a fixed set of operations
that can be applied to it. 
For example, a \emph{register} supports \codestyle{read}(), which returns the value of the register,
and \codestyle{write}($x$), which changes the value of the register to $x$.

Each process can be modelled as a deterministic state machine.
A \emph{step} by a process specifies an operation 
that the process applies to a shared object and the response, if any, that the operation
returns. 
The process can then perform local computation, 
perhaps based on the response of the operation, to update its state.
A process  can crash, after which it takes no more steps.

A \emph{configuration} consists of the state of every process and the value of every shared object.
In an \emph{initial configuration}, each process is in its initial state and each object has its initial value.
An \emph{execution} is an alternating sequence of configurations and steps, starting with an initial configuration. If an execution is finite, it ends with a configuration.
The order in which processes take steps is determined by an adversarial scheduler.

The \emph{consensus number} of an object is the largest positive integer $n$
such that there is an algorithm solving consensus among $n$ processes
using only instances of this object and registers. 
A register 
has consensus number 1.
Another example of an object with consensus number 1 is 
an \emph{ABA-detecting register} \cite{aghazadeh2015time}.
It supports \codestyle{dWrite}($x$)
and \codestyle{dRead}().
When \codestyle{dWrite}($x$) is performed, $x$ 
is stored in the object.
When a process $p$ performs \codestyle{dRead}(), the value 
stored in 
the object is returned, together with an output bit, which is true if and only if process $p$ has previously performed \codestyle{dRead}() and some \codestyle{dWrite}() has been performed since its last \codestyle{dRead}().
We use a restricted version of an ABA-detecting register, where
\codestyle{dWrite}() does not have an input parameter
and  \codestyle{dRead}() simply returns the output bit.

A \emph{test\&set} object 
has initial value 0 and supports only one operation, \codestyle{t\&s}().
If the value of the object was 0, this operation changes its value to 1 and returns 0.
If the value of the object was 1, it just returns 1.
In the first case,  we say that the
\codestyle{t\&s} was \emph{successful} and, in the second case, we say that it was
\emph{unsuccessful}.
A \emph{resettable test\&set} object is like a test\&set object, except that it also supports
\codestyle{reset}(), which changes the value of the object to 0.
A \emph{fetch\&increment} object 
supports \codestyle{f\&i}(), 
which increases the value of the object by 1 and returns the previous value of the object.
A \emph{compare\&swap} object supports the operation \codestyle{c\&s}($u,v$), which checks whether
the object has value $u$ and, if so, the operation is \emph{successful} and it changes the value of the object to $v$.
It always returns the value of the object immediately before the operation was performed.
For any type of object that does not support \codestyle{read}, its readable version supports  \codestyle{read} in addition to its other operations.
Two operations, $op$ and $op'$, \emph{commute} 
if, whenever $op$ is applied to the object followed by $op'$,
the resulting value of the object is the same as when $op'$ is applied
followed by $op$.
Note that the responses to $op$ and $op'$ may be different when these operations are applied in the opposite order.
The operation $op'$ \emph{overwrites} the operation $op$
if, whenever $op$ is applied to the object followed by $op'$,
the resulting value of the object is the same as when only $op'$ is applied.
Note that the response to $op'$ may be different depending on whether $op$ was applied. 
An object is \emph{interfering} if every two of its operations either commute
or one of them overwrites the other.
All interfering objects have consensus number at most 2~\cite{herlihy1991wait}.
Registers are interfering.
The test\&set, resettable test\&set, and fetch\&increment objects and their readable versions are all
interfering and have consensus number 2.

Both a \emph{stack} and a \emph{queue} are also examples of 
objects with consensus number 2 \cite{herlihy1991wait}.
However, they are not interfering.
The value of these objects is an unbounded sequence of elements, which is initially empty.
A stack supports \codestyle{push}($x$) and \codestyle{pop}(),
whereas a queue supports \codestyle{enqueue}($x$) and \codestyle{dequeue}().
A \codestyle{pop}() or \codestyle{dequeue}() that returns a value is \emph{successful}.
If it returns EMPTY, it is \emph{unsuccessful}.


A \emph{bag} is 
another object that is not interfering. Its
value is a multiset of elements, that is initially empty.
The number of elements that can be in the multiset is unbounded.
The elements are taken from a set of values $V$ that does not include $\bot$.
A bag supports two operations, \codestyle{Insert}($x$) where $x \in V$ and \codestyle{Take}.
When \codestyle{Insert}($x$) is performed, the element $x$ is added to the multiset.
It does not return anything.
If the multiset is not empty, a \codestyle{Take} operation removes an arbitrary
element from the multiset and returns it. In this case, we say that the operation is 
\emph{successful}.
If the multiset is empty, then \codestyle{Take} returns EMPTY and
we say that it is \emph{unsuccessful.}
Note that this specification is nondeterministic.

A \emph{b-bounded bag} object is an object whose value is a multiset of at most $b$ elements.
If the multiset contains fewer than $b$ elements when \codestyle{Insert}($x$) is performed, the element $x$ is added to the multiset and OK is returned.
If the multiset already contains $b$ elements, the value of the bag does not change
and FULL is returned.
For simplicity, we assume no value repetitions in the bounded bags in our proofs.

An \emph{implementation} of an object $O$ shared by a set of processes $P$ from a collection of objects $C$
provides a representation of 
$O$ using objects in $C$ and algorithms for each process
in $P$ to perform each operation supported by $O$. 
In an execution, 
an operation on an implemented object
\emph{begins} when a process performs the first step of its algorithm for performing this operation
and is \emph{completed} at the step in which the process returns from the algorithm, with the response of the operation, if any.

A bag has a straightforward implementation from a stack or a queue.
The multiset is represented by the sequence, \codestyle{Insert}($x$)
is performed by applying \codestyle{push}($x$) or \codestyle{enqueue}($x$)
and \codestyle{Take}() is performed by applying \codestyle{pop}() or \codestyle{dequeue}(). Thus a bag has consensus number at most 2.

There is a simple algorithm for solving consensus among 2 processes using two single-writer registers and a bag initialized with one element:
To propose the value $x$, a process writes $x$ to its
single-writer register and then performs a \codestyle{Take}() operation on the bag. If the \codestyle{Take}() returns an element, the process returns $x$. Otherwise, it reads and returns the value that the other process wrote to its single-writer register.
It follows from Borowsky, Gafni, and Afek's \emph{Initial State Lemma} \cite[Lemma 3.2]{borowsky1994consensus}, that consensus among 2 processes can be solved
from registers and two initially empty bags.
Thus, a bag has consensus number exactly 2.

\begin{remove}
\galy{It therefore makes sense to implement a bag using interfering objects, an approach we explore in this paper.}
\faith{I don't see the implication explaining why looking at implementations from interfering objects makes sense.
Stacks and queues are also objects with consensus number 2.}\galy{And they do have linearizable implementations from interfering objects. In any case, I wanted to have some sentence that justifies why we elaborate so much on consensus numbers in the preliminaries section. Maybe we could say that implementing a bag from interfering objects might be possible as it is not ruled out by consensus numbers (which would be the case if they all had consensus number 1), and we explore this direction for strongly-linearizable implementations, which seem like a challenge as queues and stacks that have consensus number 2 have no lock-free, strongly-linearizable implementations from interfering objects.}
\faith{I don't think it makes sense to talk about this in the Preliminaries.
I think that we already motivate looking at implementing a bag from interfering objects
in the Introduction.}
\galy{The thing is I still feel like it's unclear why we detail so much about consensus numbers. Maybe it's really redundant and irrelevant to the paper.}
\end{remove}

\begin{remove}
Two operations are \emph{concurrent} if both begin before either one completes.
A sequence of operations together with their responses is \emph{legal} if it satisfies the sequential specification of the object.
A \emph{linearization} of an execution is an ordering
of every completed operation and a (possibly empty) subset of the operations that have begun, but have not terminated, such that:
\begin{itemize}
\item
every two of these operations that are not concurrent 
appear in the order in which they began and
\item
the sequence of these operations together with their responses is legal.
\faith{This isn't quite right. We also have to assign responses to the incomplete operations.}
\end{itemize}
Equivalently,  each completed operation and a (possibly empty) subset of the operations  that have begun but have not terminated can be assigned a distinct linearization point between its beginning and its completion in the execution such that the sequence of these operations in order of their linearization points
together with their responses is legal.
\end{remove}

For any execution, $\alpha$, consider an ordering of every completed operation and a (possibly empty) subset of the operations that have begun, but have not completed, such that, for every completed operation, $op$, each of these operations that begins after $op$ has completed is ordered after $op$. A way to obtain such an ordering is to first assign a linearization point in $\alpha$
to each of these operations within its \emph{operation interval}, i.e.,
no operation is linearized before it begins or after it has completed.
Then specify an ordering of the operations that are linearized at the same point.
Suppose there is a sequential execution in which the operations in the ordering are performed in order
such that every completed operation in $\alpha$
has the same response as it does in this sequential execution. Then we say that the ordering is a \emph{linearization} of 
$\alpha$.

An implementation of an object is \emph{linearizable}~\cite{HerlihyWing90} if each of its executions has a linearization.
An implementation is \emph{strongly-linearizable}~\cite{golab2011linearizable}
if there is a function $f$ that maps each execution $\alpha$ to a linearization of $\alpha$ such that, for every prefix $\alpha'$ of $\alpha$, $f(\alpha')$ is a prefix of $f(\alpha)$.



An implementation is \emph{wait-free} if, in every execution, every operation by a process that does not crash completes within a finite number of steps by that process.
An implementation is \emph{lock-free} if every infinite execution contains an infinite number of completed operations.

\section{Related Work}

Many universal constructions (for example, from compare\&swap objects and registers or from consensus objects and registers) provide strongly-linearizable,
wait-free implementations of any shared object~\cite{golab2011linearizable}.

Treiber \cite{treiber1986systems,michael1998nonblocking} 
gave a simple lock-free, strongly-linearizable implementation of a stack  from
a readable compare\&swap object
and registers. 
Each \codestyle{push} or successful \codestyle{pop} performs exactly one successful \codestyle{c\&s},
which is where the operation is linearized.
Each unsuccessful \codestyle{pop} performs exactly one read of a compare\&swap object
that returns NULL, which is where the operation is linearized.

Michael and Scott \cite{michael1996simple} gave a lock-free, linearizable implementation of a queue 
from readable compare\&swap objects and registers.
Assuming nodes are never reused, it is possible to make a small modification to their implementation to make it strongly-linearizable \cite{hwang2022strongly}.
Each \codestyle{enqueue} or successful \codestyle{dequeue}
performs exactly one successful \codestyle{c\&s} to an object other than $Tail$,  which is where it is linearized. 
Each unsuccessful \codestyle{dequeue} performs exactly one read of a compare\&swap object
that returns NULL, which is where the operation is linearized.
\begin{remove}
Assuming nodes are never reused, 
one of the reads their algorithm performs 
may be eliminated, resulting in a strongly-linearizable algorithm.
\end{remove}

Even though there are wait-free, linearizable implementations of registers, max-registers, 
and single-writer snapshot objects
from single-writer registers, 
Helmi, Higham, and Woelfel~\cite{HHW12} proved that there are no lock-free, strongly-linearizable implementations.
They also gave a wait-free, strongly-linearizable implementation of a bounded max-register 
from  registers.
Denysyuk and Woelfel \cite{DW2015} proved
there are no wait-free, strongly-linearizable implementations of 
max-registers and single-writer snapshot objects from registers, although they gave
lock-free strongly-linearizable implementations.
Later, Ovens and Woelfel~\cite{ovens2019strongly} gave lock-free, strongly-linearizable  implementations of an ABA-detecting register and a single-writer snapshot object from 
a bounded number of bounded size registers.

\ignore{
Recently, Attiya, Casta\~neda, and Enea~\cite{attiya2024strong} gave a
wait-free, strongly-linearizable implementation of a max-register and a single-writer snapshot object from a fetch\&add object.
They also gave\\
- a wait-free, strongly-linearizable implementation of a readable test\&set object from a test\&set object and a register,\\
- a wait-free, strongly-linearizable implementation of a readable, resettable test\&set object 
from a max-register and an infinite array of readable test\&set objects, and\\
- a lock-free, strongly-linearizable implementation of a readable fetch\&increment object
from an infinite array of readable test\&set objects.}

Afek, Gafni, and Morrison \cite{afek2007common2} gave a wait-free, linearizable implementation of a stack from interfering objects (specifically, a fetch\&add object, registers, and test\&set objects).
Attiya and Enea \cite{attiya2019putting} showed that this implementation is not strongly-linearizable.
More generally, Attiya, Casta\~neda, and Hendler~\cite{ACH18} 
proved that 
there is no lock-free, strongly-linearizable implementation of a stack or a queue shared by more than 2 processes from interfering objects.

Li~\cite{Li01} gave a lock-free, linearizable implementation of a queue from a fetch\&\-increment object, an infinite array of test\&set objects, and an infinite array of registers.
This implementation appears in \Cref{fig: nonSL bag}.
 He also gave a wait-free, linearizable implementation of a queue in which at most 2 processes
may perform \codestyle{dequeue}.
It uses a fetch\&increment object, an infinite array of registers, and each \codestyle{dequeue} dynamically allocates a constant number of registers and test\&set objects.
David~\cite{D04} gave a wait-free, linearizable implementation of a queue in which at most 1 process may \codestyle{enqueue}.
It uses registers, an infinite array of fetch\&increment objects, and a two-dimensional infinite array of swap objects.
It is unknown whether there is a wait-free, linearizable implementation of a queue from
interfering objects in which any number of processes can \codestyle{enqueue} and \codestyle{dequeue}.

\begin{figure}[h]
\begin{lstlisting}
@\underline{Shared variables:}@
    Items[1..]: an infinite array of registers, each initialized to $\bot$
    TS[1..]: an infinite array of test&set objects, each initialized to 0
    Max: a readable fetch&increment object, initialized to 1
@\underline{Insert(x)}:@ 
    max $\leftarrow$ Max.f&i()@\label{nonSL:Put:inc max}@
    Items[max].write(x)@\label{nonSL:Put:write item}@
@\underline{Take()}:@ 
    taken_old $\leftarrow$ 0@\label{nonSL:init taken old}@
    max_old $\leftarrow$ 0@\label{nonSL:init max old}@
    repeat@\label{nonSL:repeat}@
        taken_new $\leftarrow$ 0
        max_new $\leftarrow$ Max.read() - 1@\label{nonSL:take:read max}@
        for i $\leftarrow$ 1 to max_new do
            x $\leftarrow$ Items[i].read()@\label{nonSL:take:read item}@
            if x $\neq$ $\bot$ then
                if TS[i].t&s() = 0 then return x@\label{nonSL:take:t&s}@
                taken_new $\leftarrow$ taken_new + 1
        if (taken_new = taken_old) and (max_new = max_old) then return EMPTY@\label{nonSL:take:return EMPTY}@
        taken_old $\leftarrow$ taken_new@\label{nonSL:update taken old}@
        max_old $\leftarrow$ max_new@\label{nonSL:update max old}@
\end{lstlisting}
\caption{An implementation of a queue that is not a strongly-linearizable bag}\label{fig: nonSL bag}
\end{figure}

Attiya, Casta\~neda, and Enea \cite{attiya2024strong} claimed that Li's lock-free implementation of a queue was also a strongly-linearizable implementation of a bag. However, their claim is false.

Let $ins_1$ be a call of \codestyle{Insert(1)}, let $ins_2$ be a call of \codestyle{Insert(2)}, and let $tk$ be a call of
\codestyle{Take()}. Consider the following execution $\alpha$ of their algorithm: 
\begin{itemize}
    \item $ins_1$ performs its f\&i of \codestyle{Max} on \Cref{nonSL:Put:inc max}.
    \item $ins_2$ performs its f\&i of \codestyle{Max} on \Cref{nonSL:Put:inc max}.
    \item $tk$ performs \Cref{nonSL:init taken old}, \Cref{nonSL:init max old}, and does an entire iteration of the repeat loop starting on \Cref{nonSL:repeat},
    in which it sets 
    \textit{max\_new} to 2 on \Cref{nonSL:take:read max} and sets \textit{max\_old} to 2 on line \Cref{nonSL:update max old}.
    \item $tk$ starts a second iteration of the repeat loop,
    in which it sets \textit{max\_new} to 2 on \Cref{nonSL:take:read max} and reads $\bot$ from \codestyle{Items[1]} on \Cref{nonSL:take:read item}.
    However, it does not read \codestyle{Items[2]} on \Cref{nonSL:take:read item}.
    \item $ins_1$ writes 1 to \codestyle{Items[1]} on \Cref{nonSL:Put:write item} and returns.
\end{itemize}
Since $ins_1$ completes in $\alpha$, it must be linearized in $\alpha$.

Let $\alpha_1$ be the continuation of $\alpha$ in which $tk$ runs by itself to completion. It returns EMPTY as it encounters $\bot$ in every entry of \codestyle{Items} it reads.
Since $tk$ returns EMPTY in $\alpha\cdot\alpha_1$, it must be linearized before $ins_1$. Hence, $tk$ must be linearized in $\alpha$.

Let $\alpha_2$ be another continuation of $\alpha$, in which $ins_2$ writes 2 to \codestyle{Items[2]} on \Cref{nonSL:Put:write item} and returns, and then $tk$ runs by itself to completion. From the code, $tk$ reads 2 from \codestyle{Items[2]} and returns 2. 
This implies that $ins_2$ must be linearized before $tk$ in  $\alpha\cdot\alpha_2$ and, hence, in $\alpha$. Therefore $tk$ returns 2 in $\alpha \cdot \alpha_1$, which is a contradiction.

This shows that the algorithm in \Cref{fig: nonSL bag} is not a strongly-linearizable implementation of a bag for any choice of linearization points.
In particular, there is a problem with the linearization points
mentioned in \cite{attiya2024strongArxiv}:
\begin{itemize}
\item
Each \codestyle{Insert} operation is linearized when it performs \Cref{nonSL:Put:write item}.
\item
Each successful \codestyle{Take} operation is linearized when it performs a successful \codestyle{test\&set} on  \Cref{nonSL:take:t&s}.
\item 
Each unsuccessful \codestyle{Take} operation is linearized when it last reads \codestyle{Max} on \Cref{nonSL:take:read max}.
\end{itemize}
The issue is that, when an unsuccessful \codestyle{Take} operation is linearized, the bag might not be empty.
Let $ins$ be a call of \codestyle{Insert(1)} and let $tk_1$ and $tk_2$ be calls of
\codestyle{Take()}. Consider the following execution: 
\begin{itemize}
    \item $ins$ performs its f\&i of \codestyle{Max} on \Cref{nonSL:Put:inc max}.
    \item $tk_1$ performs \Cref{nonSL:init taken old}, \Cref{nonSL:init max old}, and does an entire iteration of the repeat loop starting on \Cref{nonSL:repeat}
    in which it sets 
    \textit{max\_new} to 1 on \Cref{nonSL:take:read max} and sets \textit{max\_old} to 1 on line \Cref{nonSL:update max old}.
    \item $ins$ writes 1 to \codestyle{Items[1]} on \Cref{nonSL:Put:write item} and returns.
    \item $tk_1$ starts a second iteration of the repeat loop,
    in which it sets \textit{max\_new} to 1 on \Cref{nonSL:take:read max}.
    \item $tk_2$ performs a \codestyle{Take} operation, returning the value 1 when it first performs \Cref{nonSL:take:t&s}.
    \item $tk_1$ completes its \codestyle{Take} operation, returning EMPTY when it returns
 on \Cref{nonSL:take:return EMPTY}.
\end{itemize}
Note that, when $tk_1$ last performs \Cref{nonSL:take:read max}, where it is linearized, 
the bag contains the element 1.


A $b$-bounded bag
shared by more than $2b$ processes has
no lock-free, strongly-linearizable implementation from interfering objects.
To see why, suppose $A$ is such an implementation for $n > 2b$ processes.
If a read operation is added to an interfering object, the result is still an interfering object. Hence, $A$ can be viewed as an implementation from readable interfering objects. 
Attiya, Castaneda, and Enea \cite{attiya2024strong} defined a \emph{$b$-ordering object}.
A $b$-bounded bag is an example of a $b$-ordering object.
Applying Lemma 5.2 from their paper shows 
that $b$-set agreement among $n$ processes can be solved using $A$. Since all interfering objects have consensus number at most 2, this implies that $b$-set agreement among $n>2b$ processes can be
solved using objects of consensus number at most 2.
Herlihy and Rajsbaum \cite{HR94} proved that this is impossible.

\section{A Lock-Free, Strongly-Linearizable Implementation of a Bag}\label{sec:LF SL bag}

The first algorithm we present is a strongly-linearizable implementation of an unbounded bag. The challenge in designing such an algorithm from 
interfering objects lies in identifying when the bag is empty.
To address this, we use a readable fetch\&increment object, \codestyle{Done}, which an \codestyle{Insert} operation increments as its last step, to inform \codestyle{Take} operations that the bag is not empty. 
A detailed description of the algorithm follows.

The implementation uses an infinite array of registers, $Items$, in which elements that have been inserted into the 
bag are stored, together with an infinite array of test\&set objects, $TS$.
The process that performs a successful \codestyle{t\&s}() on $TS[i]$ returns the element stored in $Items[i]$ and
removes it from the bag.

Each \codestyle{Insert}($x$) operation begins by performing \codestyle{f\&i}() on a readable fetch\&in\-crement object, \codestyle{Allocated}. 
This allocates the operation a new location within \codestyle{Items} in which it writes $x$.
Thus, the value of \codestyle{Allocated} is the index of the last allocated location within \codestyle{Items}.
Finally, the operation performs \codestyle{f\&i}() on a second 
readable fetch\&increment object, \codestyle{Done}, to inform \codestyle{Take}() operations about the insertion.



A \codestyle{Take}() operation
begins by reading
\codestyle{Done} and \codestyle{Allocated}.
Then it reads the allocated locations in \codestyle{Items}.
For each location, \codestyle{i}, if \codestyle{Items[i]} contains $x \neq \bot$,
the operation performs \codestyle{t\&s}() on the corresponding test\&set object, \codestyle{TS[i]}.
If this is successful, $x$ is returned.
After reading all the allocated locations in \codestyle{Items} without
performing a successful \codestyle{t\&s}(), the operation rereads \codestyle{Done}.
If its value has not changed since the operation last read \codestyle{Done}, EMPTY is returned.
\begin{remove}
until it finds one 
such that 
that contains a value $x \neq \bot$ and 
it performs a successful \codestyle{t\&s} on its corresponding test\&set object.
In this case, it returns $x$.
If $x$ is returned 
by a \codestyle{Take}() operation before the \codestyle{Insert}($x$) operation that inserted it increments \codestyle{Done}, we say that these operations are \emph{coupled}.
If the \codestyle{Take}() operation does not perform a successful
\codestyle{t\&s} operation,
it rereads \codestyle{Done} and, if its value has not changed since its previous read, the operation returns EMPTY.
\end{remove}
Otherwise, the operation repeats the entire sequence of steps.
Note that it begins again starting from the first location, in case an \codestyle{Insert} operation that was allocated an early location 
has recently written
to that location.
Pseudocode for our implementation of a bag
appears in \Cref{fig: SLbag2}.

\begin{figure}[b!]
\begin{lstlisting}
@\underline{Shared variables}:@
    Items[1..]: an infinite array of registers, each initialized to $\bot$
    TS[1..]: an infinite array of test&set objects, each initialized to 0
    Allocated: a readable fetch&increment object, initialized to 0
    Done: a readable fetch&increment object, initialized to 0
@\underline{Insert(x)}:@
    m $\leftarrow$ Allocated.f&i() + 1@\label{unbounded:Put:inc Alloc}@
    Items[m].write(x)@\label{unbounded:Put:write item}@
    Done.f&i()@\label{unbounded:Put:Done}@
@\underline{Take()}:@ 
    repeat
        d $\leftarrow$ Done.read()@\label{unbounded:take:read Done}@
        m $\leftarrow$ Allocated.read()@\label{unbounded:take:read Alloc}@
        for i $\leftarrow$ 1 to m do
            x $\leftarrow$ Items[i].read()@\label{unbounded:take:read item}@
            if x $\neq$ $\bot$ then
                if TS[i].t&s() = 0 then return x@\label{unbounded:take:t&s}@
        if d = Done.read() then return EMPTY@\label{unbounded:take:reread Done}@
\end{lstlisting}
\caption{A lock-free, strongly-linearizable bag}\label{fig: SLbag2}
\end{figure}

\begin{remove}
When a process performs an \codestyle{Insert}($x$) operation, it is allocated a new location within \codestyle{Items} in which it writes $x$. This location is obtained by performing f\&i on \codestyle{Allocated}. Thus, the value of \codestyle{Allocated} is the first unallocated location within \codestyle{Items}.
Finally, it performs \codestyle{Done.f\&i()}.

To begin a \codestyle{Take} operation, a process reads \codestyle{Done} and \codestyle{Allocated}.
Then it looks at each of the locations in \codestyle{Items} that have been allocated, until it finds one
that contains a value $x \neq \bot$ and whose corresponding test\&set object it is able to win.
In this case, it returns $x$.
If it does not find such a location, it rereads \codestyle{Done} and, if its value has not changed since its previous read,
it returns EMPTY.
Otherwise, it repeats this sequence of steps.
\end{remove}

\medskip

\noindent{\bf Strong-linearizability.} 
Consider any execution consisting of operations on this data structure.
We linearize the operations as follows:
\begin{itemize}
    \item A \codestyle{Take}() operation that performs a successful \codestyle{t\&s}() on \Cref{unbounded:take:t&s} is linearized at this step.
    It returns the element it last read on
    \Cref{unbounded:take:read item}.
    \item A \codestyle{Take}() operation that reads \codestyle{Done} on \Cref{unbounded:take:reread Done} and obtains the same value it obtained in its previous read of \codestyle{Done} (on \Cref{unbounded:take:read Done}) is linearized at this last read. It returns EMPTY.
    \item Consider an  \codestyle{Insert}($x$) operation, $ins,$ that was allocated location \codestyle{m} and wrote $x$ to \codestyle{Items[m]} on \Cref{unbounded:Put:write item}. 
    If some \codestyle{Take}() operation performs a successful \codestyle{t\&s}() on \codestyle{TS[m]} after $ins$ performed \Cref{unbounded:Put:write item}, but before $ins$ performs \codestyle{Done.f\&i}() on \Cref{unbounded:Put:Done},
    then $ins$ is linearized immediately before this \codestyle{Take}(). In this case, we say
    that these two operations are \emph{coupled}.
    Otherwise, $ins$ is linearized when it performs \codestyle{Done.f\&i}().
\end{itemize}
The linearization point of each operation occurs within its operation interval.
At any configuration, $C$, the bag contains the difference between the multiset of elements inserted by \codestyle{Insert} operations 
linearized before $C$ and the multiset of 
elements returned by
successful \codestyle{Take} operations linearized before $C$.

\begin{remove}

\item An \codestyle{Insert}($x$) operation that 
    performs \codestyle{Done.f\&i}()
    on \Cref{unbounded:Put:Done},
    but has not already been linearized as a coupled operation,
    is linearized at this step.
\end{remove}

Each \codestyle{Take}() operation is linearized immediately before it returns.
Thus, if it is linearized at some point in an execution, it is linearized at the same point
in every extension of that execution.
Each  \codestyle{Insert}($x$) operation that is coupled
with a \codestyle{Take}() operation is linearized immediately before that \codestyle{Take}() operation. 
This means that the inserted element is taken from the bag immediately after it is inserted.
Thus, if a coupled \codestyle{Insert}($x$) operation is linearized at some point in an execution, it is linearized at the same point
in every extension of that execution.
Each uncoupled \codestyle{Insert}($x$) operation is linearized 
when it increments \codestyle{Done}. 
This ensures that, if a \codestyle{Take}() operation sees 
that the value of \codestyle{Done} has not changed between
\Cref{unbounded:take:read Done} and \Cref{unbounded:take:reread Done},
then no uncoupled \codestyle{Insert}
has been linearized during this part of the execution.
An uncoupled \codestyle{Insert}($x$) operation is linearized immediately before it returns,
so it is linearized at the same point in every extension of the execution.
Hence, if the implementation is linearizable, it is strongly-linearizable.

To show that the ordering is a linearization,
it remains to prove that the values returned by the \codestyle{Take}() operations are consistent with the sequential specifications of a bag.
Let $tk$ be a \codestyle{Take}() operation that performs a successful \codestyle{t\&s}() of
\codestyle{TS}[$i$] on \Cref{unbounded:take:t&s} and let $x \neq \bot$ be the element it read from \codestyle{Items}[$i$] 
when it last performed \Cref{unbounded:take:read item}.
Since $tk$ read $x$ from \codestyle{Items}[$i$], there
was an \codestyle{Insert}($x$) operation that wrote $x$ into \codestyle{Items}[$i$] on
\Cref{unbounded:Put:write item}. Note that there was exactly one \codestyle{Insert} operation
that was allocated $i$ on \Cref{unbounded:Put:inc Alloc}, by the semantics of \codestyle{f\&i}().
This \codestyle{Insert}($x$) operation was linearized before $tk$, either because it
executed \Cref{unbounded:Put:Done} before $tk$ performed its successful \codestyle{t\&s}() on \codestyle{TS}[$i$]
or it was coupled with $tk$ and, hence, linearized immediately before $tk$.
This occurrence of element $x$ 
remains in the bag 
until the linearization point of $tk$ 
because the element in \codestyle{Items}[$i$] is only returned by the \codestyle{Take}() operation
that performs the successful \codestyle{t\&s}() on \codestyle{TS}[$i$].

Now let $tk$ be a \codestyle{Take} operation that reads \codestyle{Done} on \Cref{unbounded:take:reread Done} and obtains the same value it obtained in its previous read of \codestyle{Done} on \Cref{unbounded:take:read Done}. 
Let $C$ be the configuration immediately after $tk$ reads \codestyle{Done} on \Cref{unbounded:take:read Done} in its last iteration of the repeat loop and
let $C'$ be the configuration immediately before 
$tk$ reads \codestyle{Done} for the last time on \Cref{unbounded:take:reread Done}.
The value of \codestyle{Done} does not change between these two configurations.
Let $m$ be the value $tk$ read from \codestyle{Allocated} on \Cref{unbounded:take:read Alloc},
which is the largest location in \codestyle{Items} that
had been previously
allocated to an \codestyle{Insert} operation.

Any element, $x$,
that was in the bag in configuration $C$ was inserted into the bag by an uncoupled \codestyle{Insert}($x$) operation.
This operation was linearized when it performed \Cref{unbounded:Put:Done}.
Suppose that it wrote $x$ into \codestyle{Items}[$i$].
This happened prior to $C$.
Note that $i \leq m$,
since $m$ was obtained when $tk$ 
read \codestyle{Allocated} on \Cref{unbounded:take:read Alloc} after $C$
and the value of \codestyle{Allocated} only increases.
When $tk$ performed \Cref{unbounded:take:read item} between configurations $C$ and $C'$, 
it read $x$ from  \codestyle{Items}[$i$].
Since $tk$ later reaches \Cref{unbounded:take:reread Done},
its \codestyle{t\&s}() on \codestyle{TS}[$i$] was unsuccessful. 
Some other \codestyle{Take} operation previously performed a successful \codestyle{t\&s}() on \codestyle{TS}[$i$] and, when it did, 
this occurrence of element $x$
was taken from the bag.
Therefore, all elements that were in the bag in configuration $C$ were taken by 
\codestyle{Take} operations other than $tk$ before $C'$.

Now consider any element that was inserted into the bag between configurations $C$ and $C'$
by some \codestyle{Insert} operation.
Since the value of \codestyle{Done} did not change between $C$ and $C'$,
this \codestyle{Insert} operation did not perform \codestyle{Done.f\&i()} on \Cref{unbounded:Put:Done}.
Hence, it was coupled with a \codestyle{Take} operation and it was immediately
taken from the bag after it was inserted into the bag.
Therefore, all elements that were inserted into the bag between configurations $C$ and $C'$
were taken from the bag prior to configuration $C'$.
Hence, the bag is empty in configuration $C'$, 
immediately before the step at which
$tk$ is linearized.

\medskip

\noindent{\bf Lock-freedom.} The \codestyle{Insert} operation performs 3 steps, so it is wait-free. 
When a \codestyle{Take} operation finishes an iteration of the loop and is about to begin another iteration, 
the value of \codestyle{Done} it last read on \Cref{unbounded:take:reread Done} was different than what it last read on \Cref{unbounded:take:read Done}. This means that some \codestyle{Insert} operation performed \codestyle{Done.f\&i()} on \Cref{unbounded:Put:Done} between these two points of the execution. This step completes the \codestyle{Insert} operation.

\medskip

\subsection{A linearizable, but not strongly-linearizable, queue}
It is worth noting that the algorithm in \Cref{fig: SLbag2}, which implements a strongly-linearizable bag, also implements a linearizable queue.
The proof of this fact is very similar to the proof used by Li~\cite{Li01}.
However, it is 
not a strongly-linearizable queue, as we next prove.
%

Let $ins_1$ be a call of \codestyle{Insert($1$)},
let $ins_2$ be a call of \codestyle{Insert($2$)},
and let $tk$, $tk_1$, and $tk_2$ be calls of \codestyle{Take()}.
Consider the following execution $\alpha$ of the algorithm in \Cref{fig: SLbag2}: 
\begin{itemize}
    \item $ins_1$ performs its f\&i of \codestyle{Allocated} on \Cref{unbounded:Put:inc Alloc} and is allocated location 1.
    \item $ins_2$ performs its f\&i of \codestyle{Allocated} on \Cref{unbounded:Put:inc Alloc} and is allocated location 2.
    \item $tk$ reads 2 from \codestyle{Allocated} on \Cref{unbounded:take:read Alloc} and then reads $\bot$ from \codestyle{Items[1]} on \Cref{unbounded:take:read item}.
    \item $ins_1$ and $ins_2$ run to completion.
\end{itemize}
Since $ins_1$ and $ins_2$ complete in $\alpha$, they must be linearized in $\alpha$.

Let $\alpha_1$ be the continuation of $\alpha$ in which $tk$ runs by itself to completion. It reads $2$ from \codestyle{Items[2]} 
on \Cref{unbounded:take:read item} during its second iteration of the for loop and returns $2$. To satisfy the semantics of a queue, $ins_2$ must be linearized before $ins_1$ in $\alpha\cdot\alpha_1$ and, hence, in $\alpha$.

Let $\alpha_2$ be another continuation of $\alpha$, in which 
$tk_1$ runs by itself to completion and then $tk_2$ runs by itself to completion.
From the code, $tk_1$ reads 1 from \codestyle{Items[1]} and returns 1, and $tk_2$ reads 2 from \codestyle{Items[2]} and returns 2. 
To satisfy the semantics of a queue, $tk_1$ must be linearized before $tk_2$ in $\alpha\cdot\alpha_2$ and, hence, in $\alpha$. This is a contradiction.


\section{A Wait-Free, Linearizable Implementation of a 1-Bounded Bag with a Single Producer}\label{sec:WF 1-bounded bag}


\begin{figure}[b!]
\begin{lstlisting}
@\underline{Shared variables:}@
    Items[1..n+1]: an array of single-writer registers, @all@ initialized to $\bot$, which can only be written to by the producer 
    TS[1..n+1]: an array of readable, resettable test&set objects, @all@ initialized to 0, except @\text{for}@ TS[1], which is initialized to 1
    Allocated: a single-writer register, initialized to 1, which can only be written to by the producer
    Hazards[1..n]: an array of single-writer registers, @all@ initialized to $\bot$, where Hazards[i] can only be written to by $P_i$
@\underline{Persistent local variables of the producer:}@
    used: a register, initialized to $\emptyset$, which contains a subset of $\{1,\ldots,n+1\}$
    m: a register, initialized to 1, which is a local copy of Allocated
@\underline{Insert(x)} by the producer:@
    if TS[m].read() = 0 then return FULL@\label{WF-1-bounded:insert:check TS[m]}@
    Items[m].write($\bot$)@\label{WF-1-bounded:insert:clear item}@
    used $\leftarrow$ used $\cup$ {m}@\label{WF-1-bounded:insert:add to used}@
    hazardous $\leftarrow$ COLLECT(Hazards)@\label{WF-1-bounded:insert:collect hazards}@
    m $\leftarrow$ some index in $\{1, \ldots, n+1\}$ - hazardous@\label{WF-1-bounded:insert:choose index}@
    Allocated.write(m)@\label{WF-1-bounded:insert:write to Allocated}@
    for all i $\in$ used - hazardous do@\label{WF-1-bounded:insert:reset loop}@
        TS[i].reset()@\label{WF-1-bounded:insert:reset TS}@
    used $\leftarrow$ used $\cap$ hazardous@\label{WF-1-bounded:insert:narrow used down}@
    Items[m].write(x)@\label{WF-1-bounded:insert:write item}@
    return OK
@\underline{Take()} by consumer $P_i$ for $i\in \{1,\ldots,n\}$:@ 
    a $\leftarrow$ Allocated.read()@\label{WF-1-bounded:take:read Allocated}@
    Hazards[i].write(a)@\label{WF-1-bounded:take:add hazard}@
    x $\leftarrow$ Items[a].read()@\label{WF-1-bounded:take:read item}@
    if x $\neq$ $\bot$ then
        if TS[a].t&s() = 0 then @\label{WF-1-bounded:take:t&s}@
            Hazards[i].write($\bot$)@\label{WF-1-bounded:take:clear hazard1}@
            return x
    Hazards[i].write($\bot$)
    return EMPTY
\end{lstlisting}
\caption{A wait-free, linearizable 1-bounded bag with one producer and $n$ consumers}\label{fig: WF 1-bounded bag single producer}
\end{figure}

Our first bounded bag algorithm from interfering objects is a wait-free, linearizable implementation of a bag that can contain at most one element. It is shared by $n$ processes, $P_1,\ldots,P_{n}$, called \emph{consumers}, that can perform \codestyle{Take}() and a single process, 
called the \emph{producer}, that can perform \codestyle{Insert}($x$). 
To keep the space bounded, the producer reuses locations in \codestyle{Items} to store new elements.
The producer announces the 
most recent location it has 
allocated 
by writing it to the shared register \codestyle{Allocated}.
There is a resettable test\&set object, instead of a test\&set object, associated with each location in \codestyle{Items}.
An array, \codestyle{Hazards}, of 
hazard pointers~\cite{michael2004hazard} 
is used to prevent multiple consumers from returning the same element, and to ensure that each element is consumed before being overwritten by a new element.
Each consumer announces a location it is about to access
and the producer avoids reusing the announced locations.
Specifically, the producer
avoids resetting the corresponding test\&set objects and writing new elements to the corresponding 
locations in \codestyle{Items}.
The producer maintains two persistent local variables, \codestyle{m}, which is the last location in \codestyle{Items} it 
allocated,
and \codestyle{used}, which is a set of locations it has used and will
need to reset before they are reallocated. To eliminate the need for a special case to handle the first \codestyle{Insert} call (which is the only one not preceded by a previous insertion), the initial state of the data structure is as if location 1 had been allocated to the producer, it had inserted an element into the bag in this location, and then 
this element had been taken by some consumer. This is simulated by setting the initial values of \codestyle{m}, \codestyle{Allocated}, and \codestyle{TS[1]} to 1. 

The producer begins an \codestyle{Insert}($x$) operation by checking whether the test\&s\-et object,  \codestyle{TS[m]},
in the last allocated location,  \codestyle{m},
is 0. If so, it returns FULL.
Otherwise, it overwrites the element in \codestyle{Items[m]} with $\bot$
and adds the index \codestyle{m} to \codestyle{used}.
Afterwards, it collects the set of \emph{hazardous} locations stored in \codestyle{Hazards}.
It then allocates an arbitrary location from $\{1,\ldots,n+1\}$
that is not hazardous
and announces this location by writing it to \codestyle{Allocated}. 
Next, it resets the test\&set for each location in \codestyle{used} that is not hazardous.
Then it removes these locations from \codestyle{used}.
Finally,  it writes $x$ to 
the newly allocated location in \codestyle{Items} and
returns OK.

To perform a \codestyle{Take}() operation, a consumer reads \codestyle{Allocated}. It announces 
the location
it read in \codestyle{Hazards} and then reads the value in this location. 
If it is an element $x\neq\bot$,
the consumer performs \codestyle{t\&s}() on the resettable test\&set object for this location.
It then clears its announcement. 
If the \codestyle{t\&s}() was successful, it returns $x$. 
If either $x = \bot$ or the \codestyle{t\&s}() was unsuccessful, it returns EMPTY.

If \codestyle{Items}[$a$] 
contains an element $x \neq \bot$, but this element has already been taken from the bag,
it is essential that no consumer can perform a successful \codestyle{t\&s}() on 
the associated test\&set object, \codestyle{TS}[$a$], and take the element again.
To ensure this, a consumer, $P_i$ writes $a$ to \codestyle{Hazards}[$i$] 
before reading \codestyle{Items}[$a$].
This prevents the producer from resetting \codestyle{TS}[$a$] while $P_i$ is poised to access \codestyle{TS}[$a$].


After it has set \codestyle{Items[m]} to $\bot$, but
before resetting any test\&set objects, the producer will collect the locations that appear in \codestyle{Hazards} and refrain from resetting the test\&set objects
associated with the hazardous locations it collected.
Hence, it will not reset a test\&set object which any consumer is poised to access. 

Assuming that the 
universe of elements that can be inserted into the bag is bounded,
the registers used in the implementations store bounded values and thus, the amount of space used by the implementation is bounded. Pseudocode for our implementation appears in \Cref{fig: WF 1-bounded bag single producer}.
Because the code contains no unbounded loops,
the implementation is wait-free.


\subsection{Linearizability} 
Consider any execution consisting of operations on this
data structure.
We linearize the operations as follows:
\begin{itemize}
    \item
    A \codestyle{Take}() operation that performs a successful \codestyle{t\&s}() on \Cref{WF-1-bounded:take:t&s} is linearized at this step.
    It returns the element it last read on \Cref{WF-1-bounded:take:read item}.
    \item
    Consider a \codestyle{Take}() operation, $tk$, that returns EMPTY.
    If \codestyle{Items}[$a$] = $\bot$ or \codestyle{TS}[$a$] = 1 
    when $tk$ read $a$ from \codestyle{Allocated} on \Cref{WF-1-bounded:take:read Allocated},
    then $tk$ is linearized at this step.
    Otherwise, we can show that 
    some other
    \codestyle{Take}() operation performed a successful \codestyle{t\&s}() on \codestyle{TS}[$a$] after $tk$ performed \Cref{WF-1-bounded:take:read Allocated},
    but before $tk$ returned.
In this case,
$tk$ is linearized immediately after the first such  \codestyle{Take}() operation.
    \item
    An \codestyle{Insert}($x$) operation that reads 0 from \codestyle{TS[m]} (on \Cref{WF-1-bounded:insert:check TS[m]}) is linearized at this read. It returns FULL. 
    \item
    An \codestyle{Insert}($x$) operation that 
    writes $x$
    on \Cref{WF-1-bounded:insert:write item} is linearized at this write. It returns OK.
\end{itemize}

Keeping track of hazardous locations enables objects to be safely reused so that bounded space is used.
However, this makes the algorithm and its proof of correctness more intricate.
We now
prove that \Cref{fig: WF 1-bounded bag single producer}
is a linearizable implementation of a 1-bounded bag.
We begin by proving
a number of facts about the algorithm to facilitate our proof.




First, we show that every location of \codestyle{Items} contains $\bot$, except for  when the producer is in a certain part of its code.
When it is in that part of its code, only location \codestyle{Allocated} contains an element.

\begin{lemma}\label{WF-invariant: no 2 items concurrently}\label{WF-invariant: no item right before write}
Before the producer first performs \Cref{WF-1-bounded:insert:write item} and when the producer is strictly between \Cref{WF-1-bounded:insert:clear item} and \Cref{WF-1-bounded:insert:write item},
all entries of \codestyle{Items} are $\bot$. Otherwise, \codestyle{Items[Allocated]} $\neq \bot$ and all other entries of \codestyle{Items} are $\bot$.
\end{lemma}

\begin{proof}
    All components of \codestyle{Items} are initialized to $\bot$. 
    The first successful \codestyle{Insert}  writes an element to \codestyle{Items[Allocated]}
    on \Cref{WF-1-bounded:insert:write item}. Any subsequent successful \codestyle{Insert}  writes $\bot$  to \codestyle{Items[Allocated]} on \Cref{WF-1-bounded:insert:clear item}. Note that the producer does not change the value of \codestyle{Allocated} between \Cref{WF-1-bounded:insert:write item} and \Cref{WF-1-bounded:insert:clear item},
    so, immediately after \Cref{WF-1-bounded:insert:clear item}, all components of  \codestyle{Items} are $\bot$.
    Later, on \Cref{WF-1-bounded:insert:write item}, the producer
 writes an element to \codestyle{Items[Allocated]}.
\end{proof}

Next, we show that 
when a location in \codestyle{TS} can be
reset, the corresponding location in
\codestyle{Items}  contains  $\bot$.

\begin{lemma}\label{WF-invariant: bot while in used}
    If $k$ $\in$ \codestyle{used}, then \codestyle{Items[k]} = $\bot$.
\end{lemma}
\begin{proof}
    The producer adds the location \codestyle{m} = $k$ to \codestyle{used} on \Cref{WF-1-bounded:insert:add to used}, after it sets \codestyle{Items[m]} to $\bot$ on \Cref{WF-1-bounded:insert:clear item}. 
    Thus, \codestyle{Items[k]} = $\bot$ when $k$ is added to \codestyle{used}. 
    Suppose that an element is subsequently written to \codestyle{Items[k]}.
    This occurs on \Cref{WF-1-bounded:insert:write item}, with \codestyle{m} = $k$.
    The value of \codestyle{m} was chosen on \Cref{WF-1-bounded:insert:choose index} to be a value
    that is not in \codestyle{hazardous}.
    Thus, if $k$ was in \codestyle{used}, it was removed from \codestyle{used} on
    \Cref{WF-1-bounded:insert:narrow used down}, before an element was written to \codestyle{Items[m]}.
\end{proof}

When a \codestyle{Take} operation returns an element, we can identify an interval
immediately before it performed its successful \codestyle{t\&s} in which this test\&set object was 0
and the corresponding location in \codestyle{Items} contained the value it returned.

\begin{lemma}\label{WF-lemma using hazards}
Suppose $tk$ is a \codestyle{Take} operation that 
performs a successful \codestyle{t\&s}() of \codestyle{TS[$a$]} on \Cref{WF-1-bounded:take:t&s}.
Let $x$ be the element that $tk$ last read from \codestyle{Items[$a$]}
on \Cref{WF-1-bounded:take:read item}.
Then \codestyle{TS[$a$]} = 0 and \codestyle{Items[$a$]} = $x$ between
when the producer last wrote to \codestyle{Items[$a$]} prior to the read
and when $tk$ performed  its successful \codestyle{t\&s}().
\end{lemma}
\begin{proof}
Let $D_2$ be the configuration right before this step, let
$D_1$ be the configuration right before $tk$ reads \codestyle{Items[a]} 
on \Cref{WF-1-bounded:take:read item}, and let $D_0$ be the configuration right after the producer last wrote to \codestyle{Items[a]} on \Cref{WF-1-bounded:insert:write item} before $D_1$.
Let $P_i$ be the process that performed $tk$.
Prior to reading $x$ from \codestyle{Items[a]} right after $D_1$, process $P_i$ wrote $a$ to \codestyle{Hazards[i]} on \Cref{WF-1-bounded:take:add hazard}.
It next writes to \codestyle{Hazards[i]} on \Cref{WF-1-bounded:take:clear hazard1}, which occurs after $D_2$. So, between $D_1$ and $D_2$,
\codestyle{Hazards[i]} = $a$.
Since $tk$ reads $x$ from \codestyle{Items[a]} immediately after $D_1$
and the producer does not write to \codestyle{Items[a]} on \Cref{WF-1-bounded:insert:write item} between $D_0$ and $D_1$, 
\codestyle{Items[a]} = $x$ throughout this interval.
    
Right after $D_2$, a \codestyle{t\&s} of \codestyle{TS[a]} succeeds, so 
\codestyle{TS[a]} = 0 at $D_2$.
Assume, for sake of contradiction, that \codestyle{TS[a]} = 1 at some point between $D_0$ and $D_2$. Then, after this point, but before $tk$
performed its successful \codestyle{t\&s} on \codestyle{TS[a]}, the producer
must have reset \codestyle{TS[a]} on \Cref{WF-1-bounded:insert:reset TS}.
By \Cref{WF-invariant: no item right before write}, \codestyle{Items[a]} = $\bot$ when the reset was performed. Thus, between  $D_1$ and $D_2$, the producer must have written $\bot$ to \codestyle{Items[a]} on \Cref{WF-1-bounded:insert:clear item} and then performed the reset.
Between these two steps, the producer collects
\codestyle{Hazards} and reads $a$ from \codestyle{Hazards[i]}.
But, by \Cref{WF-1-bounded:insert:reset loop}, the producer does not reset \codestyle{TS[a]} on \Cref{WF-1-bounded:insert:reset TS}.
This is a contradiction. Thus \codestyle{TS[a]} = 0 between $D_0$ and $D_2$.

After writing $x$ to \codestyle{Items[a]} just before $D_0$,
    the producer next writes to \codestyle{Items[a]}
    on \Cref{WF-1-bounded:insert:clear item}.
    Between these two writes, the producer must read 1 from  \codestyle{TS[a]} on \Cref{WF-1-bounded:insert:check TS[m]}.
But \codestyle{TS[a]} = 0 between $D_0$ and $D_2$.
    Thus, \codestyle{Items[a]} = $x$ between $D_0$ and $D_2$. 
\end{proof}

It easily follows that,
when a successful \codestyle{t\&s} is performed on a location in \codestyle{TS}, the corresponding location in \codestyle{Items} does not contain  $\bot$.

\begin{corollary}
\label{WF-corollary to lemma using hazards}
When a \codestyle{Take} operation performs a successful \codestyle{t\&s} on \codestyle{TS[a]},
\codestyle{Items[a]} $\neq \bot$.
\end{corollary}

Now we show that when the producer writes an element into some location in \codestyle{Items}, 
that value is available to be taken by a consumer.

\begin{lemma}\label{WF-lemma: TS=0 on write to Items}
When the producer writes an element to \codestyle{Items[m]} on
\Cref{WF-1-bounded:insert:write item}, 
the value of \codestyle{TS[m]} is 0.
\end{lemma}
\begin{proof}
Let $e$ be a step of the execution in which the producer performs \Cref{WF-1-bounded:insert:write item}. Let $a$ denote the value of \codestyle{m} at this step.

Suppose that, prior to $e$, the producer has not performed \Cref{WF-1-bounded:insert:write item} with \codestyle{m} = $a$.
Since all locations of \codestyle{Items} initially contain $\bot$ and elements are only written
into \codestyle{Items} on \Cref{WF-1-bounded:insert:write item},
\codestyle{Items[a]} = $\bot$ prior to $e$.
Since \codestyle{t\&s} on \codestyle{TS[a]} is only performed on \Cref{WF-1-bounded:take:t&s} after
seeing that \codestyle{Items[a]} $\neq \bot$, no \codestyle{t\&s} on \codestyle{TS[a]} was performed prior to $e$.
If $a \neq 1$, then \codestyle{TS[a]} is initially 0, so it is 0 when $e$ is performed and we are done.
If $a = 1$, then $a$ was added to \codestyle{used} when the producer first performed
\Cref{WF-1-bounded:insert:add to used}.

Otherwise, consider the last step prior to $e$ in which the producer performed \Cref{WF-1-bounded:insert:write item} with \codestyle{m} = $a$.
When the producer next performs \Cref{WF-1-bounded:insert:add to used}, it adds $a$ to \codestyle{used}.

So in both cases $a$ was added to \codestyle{used} prior to $e$.
By \Cref{WF-invariant: bot while in used}, \codestyle{Items[a]} = $\bot$ while
$a$ was in \codestyle{used}. Immediately after step $e$, \codestyle{Items[a]} $\neq \bot$,
so $a$ is not in \codestyle{used}.
Prior to being removed from \codestyle{used} on \Cref{WF-1-bounded:insert:narrow used down},
$a$ was not in \codestyle{hazardous}, so \codestyle{TS[a]} was reset to 0 on \Cref{WF-1-bounded:insert:reset TS} by some step $e'$.

By \Cref{WF-invariant: no 2 items concurrently}, \codestyle{Items[a]} = $\bot$ at $e'$.
Between $e'$ and $e$, there was no write to \codestyle{Items[a]} on \Cref{WF-1-bounded:insert:write item}, so \codestyle{Items[a]} = $\bot$.
When a successful \codestyle{t\&s} on  \codestyle{TS[a]} is performed,
\Cref{WF-corollary to lemma using hazards}
implies that \codestyle{Items[a]} $\neq$ $\bot$.
Thus, no successful \codestyle{t\&s} on  \codestyle{TS[a]} was performed between $e'$ and $e$.
This implies that \codestyle{TS[a]} = 0 when $e$ was performed.
\end{proof}

We show that the linearization point of a \codestyle{Take} operation that returns EMPTY occurs within its execution interval.

\begin{lemma}\label{WF-lemma:unsuccessful take lin}
    Let $tk$ be a \codestyle{Take} operation that returns EMPTY. 
    Let $e$ be the step of the execution in which $tk$ reads \codestyle{Allocated} on \Cref{WF-1-bounded:take:read Allocated} into $a$. 
    Suppose that \codestyle{Items[a]} $\neq$ $\bot$ and \codestyle{TS[a]} = 0 at $e$. Then a successful \codestyle{t\&s} of \codestyle{TS[a]} occurs after $e$ and before $tk$ completes.
\end{lemma}
\begin{proof}
    If $tk$ reads an element from \codestyle{Items[a]} on \Cref{WF-1-bounded:take:read item} then, as it returns EMPTY, its \codestyle{t\&s} on \Cref{WF-1-bounded:take:t&s} must fail. This implies that another \codestyle{Take} operation has performed a successful \codestyle{t\&s} on \codestyle{TS[a]} after $e$ and before $tk$ failed its \codestyle{t\&s}, so we are done. 

    Otherwise, $tk$ reads $\bot$ from \codestyle{Items[a]}.
    But when $tk$ executes step $e$, \codestyle{Items[a]} $\neq$ $\bot$. By \Cref{WF-invariant: no 2 items concurrently}, at that point the producer must be strictly between \Cref{WF-1-bounded:insert:write item} and the following execution of \Cref{WF-1-bounded:insert:clear item}. 
    Let $e'$ be the step in which the producer entered this interval, namely, the last step of the execution in which it writes an element to \codestyle{Items} on \Cref{WF-1-bounded:insert:write item} before $e$. 
    As the producer does not reach \Cref{WF-1-bounded:insert:clear item} between $e'$ and $e$, it specifically does not reach \Cref{WF-1-bounded:insert:write to Allocated} (in which it modifies \codestyle{Allocated}) during that interval. Hence, it wrote to location $a$ at $e'$.
    By \Cref{WF-lemma: TS=0 on write to Items}, \codestyle{TS[a]} = 0 at $e'$. \codestyle{TS[a]} = 0 at $e$ by assumption. It must remain 0 between $e'$ and $e$, because if it had been set to 1, the producer would have had to reset it before $e$, but the producer does not reach \Cref{WF-1-bounded:insert:reset TS} during this interval.
    Since $tk$ reads $\bot$ from \codestyle{Items[a]} after $e$, it indicates that the producer previously wrote $\bot$ to \codestyle{Items[a]} on \Cref{WF-1-bounded:insert:clear item}. Before this write (and after $e'$) the producer read 1 from \codestyle{TS[a]} on \Cref{WF-1-bounded:insert:check TS[m]}. This implies that a successful \codestyle{t\&s} of \codestyle{TS[a]} was executed after $e'$. We have shown it could not occur between $e'$ and $e$, so it occurred between $e$ and the execution of \Cref{WF-1-bounded:take:read item} by $tk$.
\end{proof}

The linearization point of each operation occurs at one of its own steps, except for an unsuccessful \codestyle{Take} operation $tk$ in case \codestyle{Items[a]} $\neq$ $\bot$ and \codestyle{TS[a]} = 0 when it reads \codestyle{Allocated}. In this case, it is linearized after this read and before it completes, by \Cref{WF-lemma:unsuccessful take lin}.
Thus, in all cases, the linearization point of an operation is within its operation interval.

Fix an execution consisting of operations on the bag in \Cref{fig: WF 1-bounded bag single producer}.
To show that the ordering of operations that we defined is a linearization, it remains to prove that the values returned by the operations are consistent with the sequential specifications of a 1-bounded bag.
We do this by induction on the length of the execution.
Let $S$ be the ordering of a prefix, $\sigma$, of the execution, let $e$ be the next step of the execution, and $S'$ be the ordering
of $\sigma\cdot e$. 
Suppose $S$ is a linearization. 
Let $B$ be the set of items in the bag at the end of $S$. 
We will prove that $S'$ is a linearization.

It suffices to consider the steps at which operations are linearized, namely when
a \codestyle{Take} operation performs a successful \codestyle{t\&s} or reads \codestyle{Allocated} and when 
an \codestyle{Insert} operation reads 0 from \codestyle{TS} or 
writes to \codestyle{Items}.

\medskip

First, suppose that $e$ is a successful \codestyle{t\&s} of \codestyle{TS[a]} on \Cref{WF-1-bounded:take:t&s} by a \codestyle{Take} operation, $tk$.
Let $x \neq\bot$ be the value $tk$ read from \codestyle{Items[a]} when it performed \Cref{WF-1-bounded:take:read item}. 
If there is a pending \codestyle{Take()} operation, $tk'$, at the end of $\sigma$ that
read $a$ from \codestyle{Allocated} on
\Cref{WF-1-bounded:take:read Allocated} when \codestyle{Items[a]} $\neq$ $\bot$ and \codestyle{TS[a]} = 0,
then $S'$ is $S$ followed by $tk$ with return value $x$ and then $tk'$ with return value EMPTY. Otherwise, $S'$ is $S$ followed by $tk$ with return value $x$.
In both cases, to show that $S'$ is a linearization, we must show that $B = \{x\}$. 

Let $D_1$ be the configuration right before this read
and let $D_2$ be the configuration at the end of $\sigma$.
Let $D_0$ be the configuration right after the last write to \codestyle{Items[a]} on \Cref{WF-1-bounded:insert:write item} that was performed by the producer before $D_1$ and let $ins$ be the \codestyle{Insert} operation that performed this write.
By \Cref{WF-lemma using hazards}, 
\codestyle{TS[a]} = 0 and 
\codestyle{Items[a]} = $x$ between $D_0$ and $D_2$. Therefore, by \Cref{WF-invariant: no item right before write}, no write to any cell of \codestyle{Items} on \Cref{WF-1-bounded:insert:write item} is performed between $D_0$ and $D_2$, which means $ins$ performed the last write of an element to \codestyle{Items} in $\sigma$.

Since $ins$ performed \Cref{WF-1-bounded:insert:write item} right before $D_0$,
it was completed in $\sigma$.
Hence $ins$ appears in $S$. Immediately after it appears, the bag contains only $x$.
Between $D_0$ and $D_2$,  \codestyle{TS[a]} = 0, so no successful \codestyle{t\&s} is performed
on \codestyle{TS[a]}. Also, \codestyle{Items[a]} $\neq \bot$, so, by \Cref{WF-invariant: no 2 items concurrently}, \codestyle{Items[a']} $= \bot$ for all $a' \neq a$.
By \Cref{WF-corollary to lemma using hazards},
\codestyle{Items[a']} $\neq \bot$ when a successful \codestyle{t\&s} is performed on \codestyle{TS[a']}. Thus, no successful \codestyle{t\&s} is performed
on any location of \codestyle{TS} between $D_0$ and $D_2$.
Hence, no \codestyle{Take} operation that removes $x$ is linearized between $D_0$ and $D_2$. 
Since $ins$ is linearized right before $D_0$,
there are no successful \codestyle{Take} operations after $ins$ in $S$. Therefore, at the end of $S$, the bag still contains $x$.
Since $S$ is a linearization, it follows that $B = \{x\}$.

\medskip

Next, suppose that $e$ is a read of \codestyle{Allocated} on \Cref{WF-1-bounded:take:read Allocated} by a \codestyle{Take} operation $tk$ into $a$, when \codestyle{Items[a]} = $\bot$ or \codestyle{TS[a]} = 1. Then $S'$ is $S$ followed by $tk$ with return value EMPTY. 
To show that $S'$ is a linearization, we must show that $B$ is empty. 
$B$ is not empty if and only if there is a successful \codestyle{Insert} operation in $S$ after which there is no successful \codestyle{Take} operation.
This correlates to $\sigma$ in which the producer wrote to \codestyle{Items} on \Cref{WF-1-bounded:insert:write item} and there is no following successful \codestyle{t\&s} on \codestyle{TS}.
We will show that this is not the case in $S$.
Assume, for sake of contradiction, that $\sigma$ contains a write of an element to \codestyle{Items} by the producer on \Cref{WF-1-bounded:insert:write item}, but no successful \codestyle{t\&s} on \codestyle{TS} after the last such write. Let $e$ be the step in which the producer executes this last write. By \Cref{WF-lemma: TS=0 on write to Items}, \codestyle{TS[Allocated]} = 0 at $e$. Since no successful \codestyle{t\&s} follows $e$ in $\sigma$, the producer may not pass \Cref{WF-1-bounded:insert:check TS[m]}. In particular, it may not reach \Cref{WF-1-bounded:insert:write to Allocated} and thus \codestyle{Allocated} remains unchanged. It also does not reach \Cref{WF-1-bounded:insert:clear item}, hence (and since \codestyle{Allocated} remains unchanged), \codestyle{Items[Allocated]} remains non-$\bot$ until the end of $\sigma$.
The fact that \codestyle{Allocated} remains unchanged also implies that \codestyle{TS[Allocated]} remains 0 until the end of $\sigma$. 
But these contradict the fact that in the end of $\sigma$, \codestyle{Items[Allocated]} = $\bot$ or \codestyle{TS[Allocated]} = 1.
Thus, $B$ is empty.

\medskip

Now, suppose $e$ is a read of 0 from \codestyle{TS[m]} by an \codestyle{Insert(x)} operation $ins$ on \Cref{WF-1-bounded:insert:check TS[m]}. Then $S'$ is $S$ followed by $ins$ with return value FULL. 
To show that $S'$ is a linearization, we must show that $B$ is not empty.

Only successful insertions can modify \codestyle{m} and \codestyle{Allocated} or reset components of  \codestyle{TS}.
Initially, \codestyle{m} = \codestyle{Allocated} = 1, \codestyle{TS[1]} = 1,  and all other cells of \codestyle{TS} are initially  0. Hence, for $ins$ to read 0 from \codestyle{TS[m]} on \Cref{WF-1-bounded:insert:check TS[m]}, a previous successful \codestyle{Insert} operation must have been performed before $ins$, changing the value of  \codestyle{m} and \codestyle{Allocated} or resetting \codestyle{TS[1]}. 

Let $ins^*$ be the last successful \codestyle{Insert} operation the producer performed before $ins$. Since
$ins^*$ has completed, it is linearized before $ins$.
Since $ins^*$ is linearized when it performs \Cref{WF-1-bounded:insert:write item}, 
the producer does not perform \Cref{WF-1-bounded:insert:clear item} through \Cref{WF-1-bounded:insert:write item}
strictly between the linearization points of $ins^*$ and $ins$.
Note that \codestyle{m} is only changed on
\Cref{WF-1-bounded:insert:choose index}
and \codestyle{Allocated} is only changed on
\Cref{WF-1-bounded:insert:write to Allocated}, so 
their values do not change during this interval.
There are also no resets performed during this interval, since these only occur
on \Cref{WF-1-bounded:insert:reset TS}.
Moreover, by \Cref{WF-invariant: no 2 items concurrently}, throughout this interval, \codestyle{Items[Allocated]} $\neq \bot$
and all other locations of \codestyle{Items} contain $\bot$.
\Cref{WF-corollary to lemma using hazards} implies that the only successful \codestyle{t\&s} performed during this interval is on
\codestyle{TS[Allocated]}. However, if there is such a successful \codestyle{t\&s},
it is not reset before $ins$ reads \codestyle{TS[m]}, so it would read 1, rather than 0.
Hence, no successful \codestyle{Take} operation is linearized between $ins^*$ and $ins$, and so $B$ contains the item that $ins^*$ inserted.

\medskip

Finally, suppose $e$ is a write to \codestyle{Items} by an \codestyle{Insert(x)} operation $ins$ on \Cref{WF-1-bounded:insert:write item}. 
Then $S'$ is $S$ followed by $ins$ with return value OK. 
To show that $S'$ is a linearization, we must show that $B$ is empty.
This is clearly true if there are no \codestyle{Insert} operations in $S$.
Otherwise, let $ins^*$ be the last \codestyle{Insert} operation in $S$. 
Let $e^*$ be the step of $\sigma$ in which $ins^*$
wrote to \codestyle{Items[$a^*$]} on \Cref{WF-1-bounded:insert:write item}.
By \Cref{WF-lemma: TS=0 on write to Items}, \codestyle{TS[$a^*$]} = 0 when $e^*$ occurs.
Since $S$ is a linearization, right after $ins^*$, the bag contains one item. 
Only successful \codestyle{Insert} operations modify \codestyle{m} (on \Cref{WF-1-bounded:insert:choose index}), so \codestyle{m} = $a^*$ 
when $ins$ read 1 from \codestyle{TS[m]} on \Cref{WF-1-bounded:insert:check TS[m]}.
Thus, a successful \codestyle{t\&s} on \codestyle{TS[$a^*$]} must have been performed between $e^*$
and this step. The \codestyle{Take} operation that performed it is linearized at this 
\codestyle{t\&s}. It appears after $ins^*$ in $S$, so the bag is empty at the end of $S$.

\subsection{Not Strongly-Linearizable}\label{sec:WFnonSLProof}
We show that \Cref{fig: WF 1-bounded bag single producer}
is not a strongly-linearizable implementation of a 1-bounded bag with a single producer.

Let $ins_i$ be a call of \codestyle{Insert($i$)} and $tk_i$ a call of \codestyle{Take()}, for $1 \leq i \leq 3$. Consider the following execution $\alpha$ of the algorithm in \Cref{fig: WF 1-bounded bag single producer}: 
\begin{itemize}
    \item $ins_1$ starts running, and it is allocated location 1 on \Cref{WF-1-bounded:insert:choose index}. It then runs to completion.
    \item $tk_1$ runs and returns 1.
    \item $tk_2$ reads 1 from \codestyle{Allocated} on \Cref{WF-1-bounded:take:read Allocated}.
    \item $ins_2$ starts running. It writes $\bot$ to \codestyle{Items[1]} on \Cref{WF-1-bounded:insert:clear item}. It is allocated location 2 on \Cref{WF-1-bounded:insert:choose index}. It then runs to completion.
\end{itemize}

Let $\alpha_1$ be the continuation of $\alpha$ in which $tk_2$ runs by itself to completion. It returns EMPTY as it encounters $\bot$ in every entry of \codestyle{Items} it reads including \codestyle{Items[1]}.
Since $tk_2$ returns EMPTY in $\alpha\cdot\alpha_1$, it must be linearized before $ins_2$. Hence, $tk_2$ must be linearized in $\alpha$, with return value EMPTY. 

Let $\alpha_2$ be another continuation of $\alpha$, in which the following occur:
\begin{itemize}
    \item $tk_3$ runs and returns 2.
    \item $ins_3$ starts running. It is allocated location 1 on \Cref{WF-1-bounded:insert:choose index}. It continues to run, writes 3 to \codestyle{Items[1]} and returns.
    \item $tk_2$ runs to completion.
\end{itemize}
From the code, $tk_2$ reads 3 from \codestyle{Items[1]} and returns 3.
Hence, $tk_2$ must be linearized after $ins_3$ in $\alpha\cdot\alpha_2$. This contradicts strong-linearizability.

\section{A Lock-Free, Strongly-Linearizable Implementation of a 1-Bounded Bag with a Single Producer}\label{sec:1-bounded bag}

\begin{figure}[b!]
\begin{lstlisting}
@\underline{Shared variables:}@
    Items[1..n+1]: an array of single-writer registers, @all@ initialized to $\bot$, which can only be written to by the producer 
    TS[1..n+1]: an array of readable, resettable test&set objects, @all@ initialized to 0, except @\text{for}@ TS[1], which is initialized to 1
    Allocated: a single-writer register, initialized to 1, which can only be written to by the producer 
    Hazards[1..n]: an array of single-writer registers, @all@ initialized to $\bot$, where Hazards[i] can only be written to by $P_i$
    Done: an ABA-detecting register
@\underline{Persistent local variables of the producer:}@
    used: a register, initialized to $\emptyset$, which contains a subset of $\{1,\ldots,n+1\}$
    m: a register, initialized to 1, which is a local copy of Allocated
@\underline{Insert(x)} by the producer:@
    if TS[m].read() = 0 then return FULL@\label{1-bounded:insert:check TS[m]}@
    Items[m].write($\bot$)@\label{1-bounded:insert:clear item}@
    used $\leftarrow$ used $\cup$ {m}@\label{1-bounded:insert:add to used}@
    hazardous $\leftarrow$ COLLECT(Hazards)@\label{1-bounded:insert:collect hazards}@
    m $\leftarrow$ some index in $\{1, \ldots, n+1\}$ - hazardous@\label{1-bounded:insert:choose index}@
    Allocated.write(m)@\label{1-bounded:insert:write to Allocated}@
    for all i $\in$ used - hazardous do@\label{1-bounded:insert:reset loop}@
        TS[i].reset()@\label{1-bounded:insert:reset TS}@
    used $\leftarrow$ used $\cap$ hazardous@\label{1-bounded:insert:narrow used down}@
    Items[m].write(x)@\label{1-bounded:insert:write item}@
    Done.dWrite()@\label{1-bounded:insert:write to done}@
    return OK
@\underline{Take()} by $P_i$, for $i\in\{1,\ldots,n\}$:@ 
    Done.dRead()@\label{1-bounded:take:read done}@
    repeat
        a $\leftarrow$ Allocated.read()@\label{1-bounded:take:read Allocated}@
        Hazards[i].write(a)@\label{1-bounded:take:add hazard}@
        x $\leftarrow$ Items[a].read()@\label{1-bounded:take:read item}@
        if x $\neq$ $\bot$ then@\label{1-bounded:take:test item}@
            if TS[a].t&s() = 0 then @\label{1-bounded:take:t&s}@
                Hazards[i].write($\bot$)@\label{1-bounded:take:clear hazard1}@
                return x@\label{1-bounded:take:return item}@
        Hazards[i].write($\bot$)@\label{1-bounded:take:clear hazard2}@
        if Done.dRead() = false then return EMPTY@\label{1-bounded:take:reread done}@
\end{lstlisting}
\caption{A lock-free, strongly-linearizable 1-bounded bag with one producer and $n$ consumers}\label{fig: 1-bounded bag single producer}
\end{figure}

In this section, we present a lock-free, strongly-linearizable implementation of a bag that can contain at most one element. It is shared by $n$ processes, $P_1,\ldots,P_{n}$, called 
\emph{consumers}, that can perform \codestyle{Take}() and a single process, 
called the \emph{producer}, that can perform \codestyle{Insert}($x$).
It combines ideas from our lock-free, strongly-linearizable implementation of an unbounded bag in \Cref{sec:LF SL bag} and our 
wait-free, linearizable implementation of a 1-bounded bag in \Cref{sec:WF 1-bounded bag}.
To keep the space bounded, 
we use an ABA-detecting register for \codestyle{Done}, instead of a readable fetch\&increment object.
An ABA-detecting register 
has a lock-free, strongly-linearizable implementation from registers using bounded space~\cite{ovens2019strongly}.

An \codestyle{Insert}($x$) operation by the producer is the same as in the wait-free linearizable implementation
presented in the previous section, except the producer
writes to the ABA-detecting register \codestyle{Done} before it returns OK.

\begin{remove}
The producer begins an \codestyle{Insert}($x$) operation by checking whether the test\&set object,  \codestyle{TS[m]},
in the last allocated location,  \codestyle{m}, is 0 and returns FULL in this case. Otherwise, it clears the last item (setting \codestyle{Items[m]}  to $\bot$) and adds its index, \codestyle{m}, to \codestyle{used}.
Afterwards, it collects the set of \emph{hazardous} locations stored in \codestyle{Hazards}.
It then allocates an arbitrary location from $\{1,\ldots,n+1\}$
that is not hazardous
assigns it to \codestyle{m}, and announces it by writing this location to \codestyle{Allocated}. 
Next, it resets the test\&set for each location in \codestyle{used} that is not hazardous. Then it removes these locations from \codestyle{used}.
Finally,  it writes $x$ to 
the newly allocated location in \codestyle{Items},
writes to the ABA-detecting register \codestyle{Done}, and returns OK.
\end{remove}


To perform a \codestyle{Take}() operation, a consumer, $P_i$, reads \codestyle{Done} and \codestyle{Allo\-cated}. 
In \codestyle{Hazards[i]}, it announces 
the location, \codestyle{a}, that it read from \codestyle{Allocated}. 
Then $P_i$ reads \codestyle{Items[a]} and, if it contains an element
 $x\neq\bot$, then $P_i$
 performs  \codestyle{t\&s}() on
\codestyle{TS[a]}.
Next $P_i$ clears its announcement. If the \codestyle{t\&s}() was successful, $P_i$ returns $x$. 
If either $x = \bot$ or the \codestyle{t\&s}() was unsuccessful, $P_i$ rereads \codestyle{Done} and, if its value has not changed since $P_i$'s previous read,
$P_i$ returns EMPTY.
Otherwise, $P_i$ repeats the entire sequence of steps.
Pseudocode for our implementation appears in \Cref{fig: 1-bounded bag single producer}.

\ignore{
Our use of hazards is similar to hazard pointers \cite{michael2004hazard}. Intuitively, the hazard announcements by takers are used to prevent the producer from reclaiming locations that takers might be trying to access. More specifically, they are used to prevent the producer from resetting a test\&set object if some taker is poised to perform a \codestyle{t\&s} on it.
This ensures that multiple takers do not return the same item.
}



\subsection{Strong-Linearizability}
Consider any execution consisting of operations on this data structure. We linearize the operations as follows:
\begin{itemize}
\item
A \codestyle{Take}() operation that performs a successful \codestyle{t\&s}() on \Cref{1-bounded:take:t&s} is linearized when it performs this \codestyle{t\&s}(). It returns the value it read in its last execution of \Cref{1-bounded:take:read item}. 
\item
A \codestyle{Take}() operation that gets false from  \codestyle{Done.dRead()} on \Cref{1-bounded:take:reread done} 
is linearized at this step.
It returns EMPTY. 
\item
An \codestyle{Insert}($x$) operation that reads 0 from \codestyle{TS[m]} (on \Cref{1-bounded:insert:check TS[m]}) is linearized at this read. It returns FULL. 
\item
Consider an \codestyle{Insert}($x$) operation, $ins$, that 
allocated location \codestyle{a} to \codestyle{m} on \Cref{1-bounded:insert:choose index}
and wrote $x$ to \codestyle{Items[a]}
on \Cref{1-bounded:insert:write item}.
If some \codestyle{Take}() operation performs a successful \codestyle{t\&s}() on \codestyle{TS[a]} after $ins$ performed \Cref{1-bounded:insert:write item}, but before $ins$ performs \codestyle{Done.dRead()} 
on \Cref{1-bounded:insert:write to done}, then $ins$ is linearized  
immediately before the \codestyle{Take}() operation that performed the first such \codestyle{t\&s}(). 
\begin{remove}
If it has performed \Cref{1-bounded:insert:write item}, but not yet 
performed \codestyle{Done.dRead()} 
on \Cref{1-bounded:insert:write to done},
when some \codestyle{Take}() operation performs a successful \codestyle{t\&s}() on \codestyle{TS[a]},
$ins$ is linearized 
immediately before the \codestyle{Take}() operation that performed the first such \codestyle{t\&s}(). 
\end{remove}
In this case, $ins$ returns OK and we say that these \codestyle{Insert}($x$) and \codestyle{Take}() operations are \emph{coupled}.
\item
An \codestyle{Insert}($x$) operation that performs \codestyle{Done.dWrite()} on \Cref{1-bounded:insert:write to done}, but has not already been linearized as a coupled operation, is linearized at this step.
In this case, it returns OK and we
say that it is \emph{uncoupled}.
\end{itemize}


We now prove that the algorithm in \Cref{fig: 1-bounded bag single producer} is a strongly-linearizable implementation of a 1-bounded bag with a single producer.



The linearization point defined above for each operation, occurs at one of its own steps, except for coupled \codestyle{Insert} operations,
for which the linearization point is between its last two steps.
Thus, in all cases, the linearization point of an operation is within its operation interval.


Each \codestyle{Take}() operation is linearized immediately before it returns on
\Cref{1-bounded:take:reread done} or when it performs a successful \codestyle{t\&s} on \Cref{1-bounded:take:t&s}. Thus, if it is linearized at some point in an execution, it is linearized at the same point in every extension of that execution.
Each \codestyle{Insert}($x$) operation that is coupled with a \codestyle{Take}() operation is linearized immediately before that \codestyle{Take}() operation.
Thus, if a coupled \codestyle{Insert}($x$) operation is linearized at some point in an execution, it is linearized at the same point in every extension of that execution. 
Each uncoupled \codestyle{Insert}($x$) operation is linearized 
immediately before it returns, so it is linearized at the same point in every extension of the execution.
Hence, if the implementation is linearizable, it is strongly-linearizable.

We now prove a number of useful facts
about the implementation.
First, we show that every location of \codestyle{Items} contains $\bot$, except for  when the producer is in a certain part of its code.
When it is in that part of its code, only location \codestyle{Allocated} contains an element.

\begin{lemma}\label{invariant: no 2 items concurrently}\label{invariant: no item right before write}
Before the producer first performs \Cref{1-bounded:insert:write item} and when the producer is strictly between \Cref{1-bounded:insert:clear item} and \Cref{1-bounded:insert:write item},
all entries of \codestyle{Items} are $\bot$. Otherwise, \codestyle{Items[Allocated]} $\neq \bot$ and all other entries of \codestyle{Items} are $\bot$.
\end{lemma}

\begin{proof}
    All components of \codestyle{Items} are initialized to $\bot$. 
    The first successful \codestyle{Insert}  writes an element to \codestyle{Items[Allocated]}
    on \Cref{1-bounded:insert:write item}. Any subsequent successful \codestyle{Insert}  writes $\bot$  to \codestyle{Ite\-ms[Allocated]} on \Cref{1-bounded:insert:clear item}. Note that the producer does not change the value of \codestyle{Allocated} between \Cref{1-bounded:insert:write item} and \Cref{1-bounded:insert:clear item},
    so, immediately after \Cref{1-bounded:insert:clear item}, all components of  \codestyle{Items} are $\bot$.
    Later, on \Cref{1-bounded:insert:write item}, the producer
 writes an element to \codestyle{Items[Allocated]}.
\end{proof}

Next, we show that 
when a location in \codestyle{TS} can be
reset, the corresponding location in
\codestyle{Items}  contains  $\bot$.

\begin{lemma}\label{invariant: bot while in used}
    If $k$ $\in$ \codestyle{used}, then \codestyle{Items[k]} = $\bot$.
\end{lemma}
\begin{proof}
    The producer adds the location \codestyle{m} = $k$ to \codestyle{used} on \Cref{1-bounded:insert:add to used}, after it sets \codestyle{Items[m]} to $\bot$ on \Cref{1-bounded:insert:clear item}. 
    Thus, \codestyle{Items[k]} = $\bot$ when $k$ is added to \codestyle{used}. 
    Suppose that an element is subsequently written to \codestyle{Items[k]}.
    This occurs on \Cref{1-bounded:insert:write item}, with \codestyle{m} = $k$.
    The value of \codestyle{m} was chosen on \Cref{1-bounded:insert:choose index} to be a value
    that is not in \codestyle{hazardous}.
    Thus, if $k$ was in \codestyle{used}, it was removed from \codestyle{used} on
    \Cref{1-bounded:insert:narrow used down}, before an element was written to \codestyle{Items[m]}.
\end{proof}

When a \codestyle{Take} operation returns an element, we can identify an interval
immediately before it performed its successful \codestyle{t\&s} in which this test\&set object contained 0
and the corresponding location in \codestyle{Items} contained the value that \codestyle{Take} returned.

\begin{lemma}\label{lemma using hazards}
Suppose $tk$ is a \codestyle{Take} operation that 
performs a successful \codestyle{t\&s} of \codestyle{TS[$a$]} on \Cref{1-bounded:take:t&s}.
Let $x$ be the element that $tk$ last read from \codestyle{Items[$a$]}
on \Cref{1-bounded:take:read item}.
Then \codestyle{TS[$a$]} = 0 and \codestyle{Items[$a$]} = $x$ between
when the producer last wrote to \codestyle{Items[$a$]} prior to the read
and when $tk$ performed  its successful \codestyle{t\&s}.
\end{lemma}
\begin{proof}
Let $D_2$ be the configuration right before $tk$ performs the successful \codestyle{t\&s} of \codestyle{TS[a]}, let
$D_1$ be the configuration right before $tk$ reads \codestyle{Items[a]} 
on \Cref{1-bounded:take:read item}
for the last time, and 
let $D_0$ be the configuration right after the producer last wrote to \codestyle{Items[a]} on \Cref{1-bounded:insert:write item} before $D_1$.
Let $P_i$ be the process that performed $tk$.
 Prior to reading $x$ from \codestyle{Items[a]} right after $D_1$, process $P_i$ wrote $a$ to \codestyle{Hazards[i]} on \Cref{1-bounded:take:add hazard}.
     It next writes to \codestyle{Hazards[i]} on \Cref{1-bounded:take:clear hazard1}, which occurs after $D_2$. So, between $D_1$ and $D_2$,
     \codestyle{Hazards[i]} = $a$.
         Since $tk$ reads $x$ from \codestyle{Items[a]} immediately after $D_1$
    and the producer does not write to \codestyle{Items[a]} on \Cref{1-bounded:insert:write item} between $D_0$ and $D_1$, 
    \codestyle{Items[a]} = $x$ throughout this interval.
    
  Right after $D_2$, a \codestyle{t\&s} of \codestyle{TS[a]} succeeds, so 
     \codestyle{TS[a]} = 0 at $D_2$.
    Assume, for sake of contradiction, that \codestyle{TS[a]} = 1 at some point between $D_0$ and $D_2$. Then, after this point, but before $tk$
    performed its successful \codestyle{t\&s} on \codestyle{TS[a]}, the producer
    must have reset \codestyle{TS[a]} on \Cref{1-bounded:insert:reset TS}.
    By \Cref{invariant: no item right before write}, \codestyle{Items[a]} = $\bot$ when the reset was performed. Thus, between  $D_1$ and $D_2$, the producer must have written $\bot$ to \codestyle{Items[a]} on \Cref{1-bounded:insert:clear item} and then performed the reset.
    Between these two steps, the producer collects
    \codestyle{Hazards} and reads $a$ from \codestyle{Hazards[i]}.
But, by \Cref{1-bounded:insert:reset loop}, the producer does not reset \codestyle{TS[a]} on \Cref{1-bounded:insert:reset TS}.
This is a contradiction. Thus \codestyle{TS[a]} = 0 between $D_0$ and $D_2$.

After writing $x$ to \codestyle{Items[a]} just before $D_0$,
    the producer next writes to \codestyle{Items[a]}
    on \Cref{1-bounded:insert:clear item}.
    Between these two writes, the producer must read 1 from  \codestyle{TS[a]} on \Cref{1-bounded:insert:check TS[m]}.
But \codestyle{TS[a]} = 0 between $D_0$ and $D_2$.
    Thus, \codestyle{Items[a]} = $x$ between $D_0$ and $D_2$. 
\end{proof}

It easily follows that,
when a successful \codestyle{t\&s} is performed on a location in \codestyle{TS}, the corresponding location in \codestyle{Items} does not contain  $\bot$.

\begin{corollary}
\label{corollary to lemma using hazards}
When a \codestyle{Take} operation performs a successful \codestyle{t\&s} on \codestyle{TS[a]},
\codestyle{Items[a]} $\neq \bot$.
\end{corollary}

Now we show that when the producer writes an element into some location in \codestyle{Items}, 
that value is available to be taken by a taker.

\begin{lemma}\label{lemma: TS=0 on write to Items}
When the producer writes an element to \codestyle{Items[m]} on
\Cref{1-bounded:insert:write item}, 
the value of \codestyle{TS[m]} is 0.
\end{lemma}
\begin{proof}
Let $e$ be a step of the execution in which the producer performs \Cref{1-bounded:insert:write item}. Let $a$ denote the value of \codestyle{m} at this step.

Suppose that, prior to $e$, the producer has not performed \Cref{1-bounded:insert:write item} with \codestyle{m} = $a$.
Since all locations of \codestyle{Items} initially contain $\bot$ and elements are only written
into \codestyle{Items} on \Cref{1-bounded:insert:write item},
\codestyle{Items[a]} = $\bot$ prior to $e$.
Since \codestyle{t\&s} on \codestyle{TS[a]} is only performed on \Cref{1-bounded:take:t&s} after
seeing that \codestyle{Items[a]} $\neq \bot$, no \codestyle{t\&s} on \codestyle{TS[a]} was performed prior to $e$.
If $a \neq 1$, then \codestyle{TS[a]} is initially 0, so it is 0 when $e$ is performed and we are done.
If $a = 1$, then $a$ was added to \codestyle{used} when the producer first performed
\Cref{1-bounded:insert:add to used}.

Otherwise, consider the last step prior to $e$ in which the producer performed \Cref{1-bounded:insert:write item} with \codestyle{m} = $a$.
When the producer next performs \Cref{1-bounded:insert:add to used}, it adds $a$ to \codestyle{used}.

By \Cref{invariant: bot while in used}, \codestyle{Items[a]} = $\bot$ while
$a$ was in \codestyle{used}. Immediately after step $e$, \codestyle{Items[a]} $\neq \bot$,
so $a$ is not in \codestyle{used}.
Prior to being removed from \codestyle{used} on \Cref{1-bounded:insert:narrow used down},
$a$ was not in \codestyle{hazardous}, so \codestyle{TS[a]} was reset to 0 on \Cref{1-bounded:insert:reset TS} by some step $e'$.

By \Cref{invariant: no 2 items concurrently}, \codestyle{Items[a]} = $\bot$ at $e'$.
Between $e'$ and $e$, there was no write to \codestyle{Items[a]} on \Cref{1-bounded:insert:write item}, so \codestyle{Items[a]} = $\bot$.
When a successful \codestyle{t\&s} on  \codestyle{TS[a]} is performed,
\Cref{corollary to lemma using hazards}
implies that \codestyle{Items[a]} $\neq$ $\bot$.
Thus, no successful \codestyle{t\&s} on  \codestyle{TS[a]} was performed between $e'$ and $e$.
This implies that \codestyle{TS[a]} = 0 when $e$ was performed.
\end{proof}

Fix an execution consisting of operations on the bag in \Cref{fig: 1-bounded bag single producer}.
To show that the ordering of operations that we defined is a linearization, it remains to prove that the values returned by the operations are consistent with the sequential specifications of a 1-bounded bag.
We do this by induction on the length of the execution.
Let $S$ be the ordering of a prefix, $\sigma$, of the execution, let $e$ be the next step of the execution, and $S'$ be the ordering
of $\sigma\cdot e$. 
Suppose $S$ is a linearization. 
Let $B$ be the set of items in the bag at the end of $S$. 
We will prove that $S'$ is a linearization.

It suffices to consider the steps at which operations are linearized, namely when
a \codestyle{Take} operation performs a successful \codestyle{t\&s} or gets false from a \codestyle{dRead} of \codestyle{Done} and when 
an \codestyle{Insert} operation reads 0 from \codestyle{TS} or 
writes to 
\codestyle{Done}.

\medskip

First, suppose that $e$ is a successful \codestyle{t\&s} of \codestyle{TS[a]} on \Cref{1-bounded:take:t&s} by a \codestyle{Take} operation, $tk$.
Let $x \neq\bot$ be the value $tk$ read from \codestyle{Items[a]} when it last performed \Cref{1-bounded:take:read item} during $\sigma$. 
Let $D_1$ be the configuration right before this read
and let $D_2$ be the configuration at the end of $\sigma$.
Let $D_0$ be the configuration right after the last write to \codestyle{Items[a]} on \Cref{1-bounded:insert:write item} that was performed by the producer before $D_1$ and let $ins$ be the \codestyle{Insert} operation that performed this write.
By \Cref{lemma using hazards}, 
\codestyle{TS[a]} = 0 and 
\codestyle{Items[a]} = $x$ between $D_0$ and $D_2$. Therefore, by \Cref{invariant: no item right before write}, no write to any cell of \codestyle{Items} on \Cref{1-bounded:insert:write item} is performed between $D_0$ and $D_2$, which means $ins$ performed the last write of an element to \codestyle{Items} in $\sigma$.
        
If there is a pending \codestyle{Insert(x')} operation, $ins'$, at the end of $\sigma$ that
wrote to \codestyle{Items[m]} on
\Cref{1-bounded:insert:write item}
(but has not incremented \codestyle{Done} on \Cref{1-bounded:insert:write to done}),
then $ins' = ins$. 
In this case, $S'$ is $S$ followed by $ins$ with return value OK and then $tk$ with return value $x$.
To show that $S'$ is a linearization, we must show that $B = \emptyset$.
If there are no \codestyle{Insert} operations in $S$, then the bag is empty
at the end of $S$. Otherwise, let $ins^*$ be the last \codestyle{Insert} operation in $S$. 
Let $e^*$ be the step of $\sigma$ in which $ins^*$
wrote to \codestyle{Items[$a^*$]} on \Cref{1-bounded:insert:write item}.
By \Cref{lemma: TS=0 on write to Items}, \codestyle{TS[$a^*$]} = 0 when $e^*$ occurs.
Since $S$ is a linearization, right after $ins^*$, the bag contains one item. 
Only successful \codestyle{Insert} operations modify \codestyle{m} (on \Cref{1-bounded:insert:choose index}), so \codestyle{m} = $a^*$ 
when $ins$ read 1 from \codestyle{TS[m]} on \Cref{1-bounded:insert:check TS[m]}.
Thus, a successful \codestyle{t\&s} on \codestyle{TS[$a^*$]} must have been performed between $e^*$
and this step. The \codestyle{Take} operation that performed it is linearized at this 
\codestyle{t\&s}. It appears after $ins^*$ in $S$, so the bag is empty at the end of $S$.
        
If there is no \codestyle{Insert} operation that performed \Cref{1-bounded:insert:write item} in $\sigma$, but not \Cref{1-bounded:insert:write to done}, 
then $S'$ is $S$ followed by $tk$ with return value $x$.
To show that $S'$ is a linearization, we must show that $B = \{x\}$. 
Since $ins$ performed \Cref{1-bounded:insert:write item} right before $D_0$,
it also performed \Cref{1-bounded:insert:write to done} and, thus, was completed in $\sigma$.
Hence $ins$ appears in $S$. Immediately after it appears, the bag contains only $x$.
Between $D_0$ and $D_2$,  \codestyle{TS[a]} = 0, so no successful \codestyle{t\&s} is performed
on \codestyle{TS[a]}. Also, \codestyle{Items[a]} $\neq \bot$, so, by \Cref{invariant: no 2 items concurrently}, \codestyle{Items[a']} $= \bot$ for all $a' \neq a$.
By \Cref{corollary to lemma using hazards},
\codestyle{Items[a']} $\neq \bot$ when a successful \codestyle{t\&s} is performed on \codestyle{TS[a']}. Thus, no successful \codestyle{t\&s} is performed on any location of \codestyle{TS} between $D_0$ and $D_2$.
Hence, no \codestyle{Take} operation that removes $x$ is linearized between $D_0$ and $D_2$. 
Since $ins$ is linearized after $D_0$,
there are no successful \codestyle{Take} operations after $ins$ in $S$. Therefore, at the end of $S$, the bag still contains $x$.
Since $S$ is a linearization, it follows that $B = \{x\}$.

\medskip

Next, suppose that $e$ is an execution of \codestyle{Done.dRead}() on \Cref{1-bounded:take:reread done} by a \codestyle{Take} operation $tk$, which returns false. Then $S'$ is $S$ followed by $tk$ with return value EMPTY. 
To show that $S'$ is a linearization, we must show that $B$ is empty.

Let $C$ be the configuration at the end of $\sigma$,
let $C^*$ be the configuration immediately after $tk$ last read \codestyle{Done} on either \Cref{1-bounded:take:read done} or \Cref{1-bounded:take:reread done},
let $S^*$ be the linearization of the prefix of $\sigma$ ending with $C^*$, and
let $B^*$ be the set of items in the bag at the end of $S^*$.
Since $tk$ returns on \Cref{1-bounded:take:reread done}, 
the producer does not write to \codestyle{Done} 
between these two configurations.

Suppose there was a successful \codestyle{Insert} operation that is in $S$, but is not in $S^*$.
Then it was linearized between configurations $C^*$ and $C$.
Since the producer did not write to \codestyle{Done} on \Cref{1-bounded:insert:write to done} between $C^*$ and $C$,
this \codestyle{Insert} operation was coupled with a \codestyle{Take} operation. Hence it was immediately
taken from the bag after it was inserted into the bag.
Therefore, if $B^*$ is empty, then so is $B$.

So, suppose $B^*$  is not empty. Then it contains exactly one element, since $S$ is a linearization.
Items inserted by coupled \codestyle{Insert} operations are immediately taken from the bag by a \codestyle{Take} operation linearized at the same step, so they are not in the bag in the end of the linearization of any prefix of the execution.
Hence they are not in $B^*$.
Therefore, the item, $x$, in $B^*$ was inserted into the bag by an uncoupled \codestyle{Insert} operation, $ins$.
This operation wrote $x$ to \codestyle{Items} in some location, $a$, on \Cref{1-bounded:insert:write item}. When that occurred, \codestyle{TS[a]} = 0 by \Cref{lemma: TS=0 on write to Items}.
Since $ins$ is uncoupled, no taker performs a successful \codestyle{t\&s} on 
\codestyle{TS[a]} between this write and when $ins$ writes to \codestyle{Done} on \Cref{1-bounded:insert:write to done}.
Hence \codestyle{TS[a]} = 0 at the linearization point of $ins$.
No successful \codestyle{Insert} follows $ins$ in $S^*$, since $S^*$ is a linearization and the bag is 1-bounded.
No successful \codestyle{Take} operation follows $ins$ in $S^*$, since it would remove $x$ from the bag.
Thus, no successful \codestyle{t\&s} is performed between the linearization point of $ins$ and $C^*$ and, so, \codestyle{TS[a]} = 0 throughout this interval.


If $tk$ read a location $a' \neq a$ from \codestyle{Allocated} on \Cref{1-bounded:take:read Allocated} after $C^*$,
then the producer wrote $a'$ to \codestyle{Allocated} on \Cref{1-bounded:insert:write to Allocated} after $ins$ was completed, which implies that the producer passed \Cref{1-bounded:insert:check TS[m]}.
If $tk$ read $a$ from \codestyle{Allocated} on \Cref{1-bounded:take:read Allocated} after $C^*$ and then read $\bot$ from \codestyle{Items[a]} on \Cref{1-bounded:take:read item},
then the producer wrote $\bot$ to \codestyle{Items[a]} on \Cref{1-bounded:insert:clear item} after $ins$ was completed, which implies that the producer passed \Cref{1-bounded:insert:check TS[m]}.
So in both cases, the producer read 1 from \codestyle{TS[a]} on \Cref{1-bounded:insert:check TS[m]} after $ins$ was completed.
Since \codestyle{TS[a]} = 0 between the end of $ins$ and $C^*$, 
some \codestyle{Take} operation performed a successful \codestyle{t\&s} on \codestyle{TS[a]} between $C^*$ and $C$.
The first \codestyle{Take} to perform such a \codestyle{t\&s} follows $ins$ in $S$ and takes $x$ from the bag, so $x \not\in B$.

Otherwise, after $C^*$ $tk$ read $a$ from \codestyle{Allocated} and then read an element from \codestyle{Items[a]}.
In this case, $tk$ must have performed an unsuccessful \codestyle{t\&s} of \codestyle{TS[a]} on \Cref{1-bounded:take:t&s}, so
another \codestyle{Take} operation performed a successful \codestyle{t\&s} of \codestyle{TS[a]} after $C^*$ but before this step.
The first \codestyle{Take} to perform such a \codestyle{t\&s} follows $ins$ in $S$ and takes $x$ 
from the bag, so $x \not\in B$.

\medskip

Now, suppose $e$ is a read of 0 from \codestyle{TS[m]} by an \codestyle{Insert}($x$) operation $ins$ on \Cref{1-bounded:insert:check TS[m]}. Then $S'$ is $S$ followed by $ins$ with return value FULL. 
To show that $S'$ is a linearization, we must show that $B$ is not empty.

Only successful insertions can modify \codestyle{m} and \codestyle{Allocated} or reset components of  \codestyle{TS}.
Initially, \codestyle{m} = \codestyle{Allocated} = 1, \codestyle{TS[1]} = 1,  and all other cells of \codestyle{TS} are initially  0. Hence, for $ins$ to read 0 from \codestyle{TS[m]} on \Cref{1-bounded:insert:check TS[m]}, a previous successful \codestyle{Insert} operation must have been performed before $ins$, changing the value of  \codestyle{m} and \codestyle{Allocated} or resetting \codestyle{TS[1]}. 

Let $ins^*$ be the last successful \codestyle{Insert} operation the producer performed before $ins$. Since
$ins^*$ has completed, it is linearized before $ins$.
Since $ins^*$ is linearized after it performs \Cref{1-bounded:insert:write item}, 
the producer does not perform \Cref{1-bounded:insert:clear item} through \Cref{1-bounded:insert:write item}
between the linearization points of $ins^*$ and $ins$.
Note that \codestyle{m} is only changed on
\Cref{1-bounded:insert:choose index}
and \codestyle{Allocated} is only changed on
\Cref{1-bounded:insert:write to Allocated}, so 
their values do not change during this interval.
There are also no resets performed during this interval, since these only occur
on \Cref{1-bounded:insert:reset TS}.
Moreover, by \Cref{invariant: no 2 items concurrently}, throughout this interval, \codestyle{Items[Allocated]} $\neq \bot$
and all other locations of \codestyle{Items} contain $\bot$.
\Cref{corollary to lemma using hazards} implies that the only successful \codestyle{t\&s} performed during this interval is on
\codestyle{TS[Allocated]}. However, if there is such a successful \codestyle{t\&s},
it is not reset before $ins$ reads \codestyle{TS[m]}, so it would read 1, rather than 0.
Hence, no successful \codestyle{Take} operation is linearized between $ins^*$ and $ins$, and so $B$ contains the item that $ins^*$ inserted.

\medskip

Finally, suppose $e$ is a write to \codestyle{Done} by an uncoupled \codestyle{Insert}($x$) operation $ins$ on \Cref{1-bounded:insert:write to done}. 
Then $S'$ is $S$ followed by $ins$ with return value OK. 
To show that $S'$ is a linearization, we must show that $B$ is empty.
This is clearly true if there are no \codestyle{Insert} operations in $S$.
Otherwise, let $ins^*$ be the last \codestyle{Insert} operation in $S$. 
Let $e^*$ be the step of $\sigma$ in which $ins^*$
wrote to \codestyle{Items[$a^*$]} on \Cref{1-bounded:insert:write item}.
By \Cref{lemma: TS=0 on write to Items}, \codestyle{TS[$a^*$]} = 0 when $e^*$ occurs.
Since $S$ is a linearization, right after $ins^*$, the bag contains one item. 
Only successful \codestyle{Insert} operations modify \codestyle{m} (on \Cref{1-bounded:insert:choose index}), so \codestyle{m} = $a^*$ 
when $ins$ read 1 from \codestyle{TS[m]} on \Cref{1-bounded:insert:check TS[m]}.
Thus, a successful \codestyle{t\&s} on \codestyle{TS[$a^*$]} must have been performed between $e^*$
and this step. The \codestyle{Take} operation that performed it is linearized at this 
\codestyle{t\&s}. It appears after $ins^*$ in $S$, so the bag is empty at the end of $S$.

\subsection{Lock-Freedom} 
The \codestyle{Insert} operation has no unbounded loops, so it is wait-free. 
Suppose a \codestyle{Take} operation, $tk$,  finishes an iteration of the loop and is about to begin another iteration, 
because its last read of \codestyle{Done} on \Cref{1-bounded:take:reread done}
returned true. Then
some \codestyle{Insert} operation wrote to \codestyle{Done} on \Cref{1-bounded:insert:write to done} 
since when $tk$ previously  read \codestyle{Done} on \Cref{1-bounded:take:read done} or \Cref{1-bounded:take:reread done}.
This write step completes the \codestyle{Insert} operation.

\section{A Lock-Free Strongly-Linearizable Implementation of a $b$-Bounded Bag with a Single Producer}\label{section: b-bounded}

\begin{figure}[p]
\begin{lstlisting}
@\underline{Shared variables:}@
    Items[1..n+b]: an array of single-writer registers, each initialized to $\ \bot$, which can only be written to by the producer
    TS[1..n+b]: an array of readable, resettable test&set objects, each initialized to 0
    Allocated: a single-writer register, initialized to $\emptyset$, which contains a subset of $\{1,\ldots,n+b\}$ of size between 0 @\text{and}@ $b$ @\text{and}@ can only be written to by the producer 
    Hazards[1..n]: an array of single-writer registers, each initialized to $\ \bot$, where Hazards[i] can only be written to by $P_i$
    InsertDone, TakeDone: ABA-detecting registers
@\underline{Persistent local variables of the producer:}@
    used: a register, initialized to $\emptyset$, which contains a subset of $\{1,\ldots,n+b\}$
    alloc: a register, initialized to $\emptyset$, which is used to update Allocated
@\underline{Insert(x)} by the producer:@
    TakeDone.dRead()@\label{b-bounded:insert:read TakeDone}@
    repeat        
        for all m $\in$ alloc do
            if TS[m].read() = 1 then@\label{b-bounded:insert:check TS}@
                Items[m].write($\bot$)@\label{b-bounded:insert:clear item}@
                alloc $\leftarrow$ alloc - {m}@\label{b-bounded:insert:remove from alloc}@
                used $\leftarrow$ used $\cup$ {m}@\label{b-bounded:insert:add to used}@
        if |alloc| < $b$ then@\label{b-bounded:insert:alloc<b}@
            hazardous $\leftarrow$ COLLECT(Hazards)@\label{b-bounded:insert:collect hazards}@
            m $\leftarrow$ some index in $\{1, \ldots, n+b\}$ - alloc - hazardous@\label{b-bounded:insert:choose index}@
            alloc $\leftarrow$ alloc $\cup$ {m}@\label{b-bounded:insert:add to alloc}@
            Allocated.write(alloc)@\label{b-bounded:insert:add to Allocated}@
            for all i $\in$ used - hazardous do@\label{b-bounded:insert:reset loop}@
                TS[i].reset()@\label{b-bounded:insert:reset TS}@
            used $\leftarrow$ used $\cap$ hazardous@\label{b-bounded:insert:narrow used down}@
            Items[m].write(x)@\label{b-bounded:insert:write item}@
            InsertDone.dWrite()@\label{b-bounded:insert:write InsertDone}@
            return OK
        else if TakeDone.dRead() = false then@\label{b-bounded:insert:reread TakeDone}@
            return FULL
@\underline{Take()} by $P_i$, for $i\in\{1,\ldots,n\}$:@ 
    InsertDone.dRead()@\label{b-bounded:take:read InsertDone}@
    repeat
        allocated $\leftarrow$ Allocated.read()@\label{b-bounded:take:read Allocated}@
        for all a $\in$ allocated do
            Hazards[i].write(a)@\label{b-bounded:take:add hazard}@
            x $\leftarrow$ Items[a].read()@\label{b-bounded:take:read item}@
            if x $\neq$ $\bot$ then@\label{b-bounded:take:test item}@
                if TS[a].t&s() = 0 then@\label{b-bounded:take:t&s}@
                    Hazards[i].write($\bot$)@\label{b-bounded:take:clear hazard1}@
                    TakeDone.dWrite()@\label{b-bounded:take:successful writes TakeDone}@
                    return x
        Hazards[i].write($\bot$)        
        if InsertDone.dRead() = false then@\label{b-bounded:take:reread InsertDone}@
            TakeDone.dWrite()@\label{b-bounded:take:failing writes TakeDone}@
            return EMPTY
\end{lstlisting}
\caption{A lock-free, strongly-linearizable $b$-bounded bag with one producer and $n$ consumers}\label{fig: b-bounded bag single producer}
\end{figure}

Our final algorithm is a lock-free, strongly-linearizable implementation of a bag that can contain at most $b$ elements. It is shared by $n$ processes, $P_1,\ldots,P_{n}$, called 
\emph{consumers}, that can perform \codestyle{Take}() and a single process, 
called the \emph{producer}, that can perform \codestyle{Insert}($x$).
It extends our lock-free, strongly-linearizable implementation of a 1-bounded bag in \Cref{sec:1-bounded bag}.
In this algorithm, \codestyle{Allocated} is a shared register that stores a subset of at most $b$ locations, and
\codestyle{alloc} is a local variable the producer uses to update the contents of \codestyle{Allocated}.
The challenge in designing a strongly-linearizable bag algorithm from interfering objects was for an operation that is trying to take an element from the bag to detect when the bag is empty (and then return EMPTY). 
This 
is 
addressed using an ABA-detecting register, \codestyle{InsertDone}, to which
each \codestyle{Insert}($x$) operation writes 
after writing $x$ to \codestyle{Items}.
When designing a strongly-linearizable \emph{bounded} bag algorithm from interfering objects,
an operation that is trying to insert an element into the bag faces a symmetrical challenge, as it needs to detect when the bag is full (and then return FULL).
To address this, we use another ABA-detecting register, \codestyle{TakeDone}, to which each \codestyle{Take}() operation writes after performing a successful \codestyle{t\&s}().
Unsuccessful \codestyle{Take} operations also write to \codestyle{TakeDone} before returning, to help them be linearized before the unsuccessful \codestyle{Take}.

The producer begins an \codestyle{Insert}($x$) operation by reading \codestyle{Ta\-keDone}. 
Then it looks at each of the locations in \codestyle{TS} that have been allocated. For each location $m$ that contains 1, it clears the element in that location (setting \codestyle{Items[m]} to $\bot$), removes \codestyle{m} from \codestyle{alloc}, and adds \codestyle{m} to \codestyle{used}.
If \codestyle{alloc} contains less than $b$ locations, it adds $x$ to the bag as follows. It first collects the non-$\bot$ values from the \codestyle{Hazards} array into the set \codestyle{hazardous}. Afterwards, it allocates an arbitrary location, \codestyle{m}, from $\{1,\ldots,n+b\}$, excluding those in \codestyle{alloc} and in \codestyle{hazardous}, adds this location to \codestyle{alloc}, and announces it by
copying \codestyle{alloc} into
 \codestyle{Allocated}.
 Next, for each location in \codestyle{used} that is not in \codestyle{hazardous}, it resets the associated test\&set object and removes the location from \codestyle{used}.
Then it writes $x$ to the allocated object, \codestyle{Items[m]}. Finally, it 
writes to the ABA-detecting register \codestyle{InsertDone} and returns OK.
If \codestyle{alloc} contained $b$ locations, it rereads \codestyle{TakeDone} and, if its value has not changed since its previous read,
it returns FULL;
otherwise, it repeats the entire sequence of steps.

To perform a \codestyle{Take}() operation, a process reads \codestyle{InsertDone} and \codestyle{Allocated}. For each location $a$ in the set of locations obtained from \codestyle{Allocated}, it announces $a$ in \codestyle{Hazards} and then reads the 
element, \codestyle{x} from \codestyle{Items[a]}.
If $x\neq\bot$, it performs  \codestyle{t\&s}() on the associated test\&set object, \codestyle{TS[a]},
and, if
successful, it clears its announcement, writes to \codestyle{TakeDone}, 
and returns $x$. 
If the process completes the for loop, it
clears its announcement, 
 and rereads \codestyle{InsertDone}.
If the value of \codestyle{InsertDone}
has not changed since its previous read,
it writes to \codestyle{TakeDone} and returns EMPTY.
Otherwise, it repeats the entire sequence of steps.
Pseudocode for our implementation appears in \Cref{fig: b-bounded bag single producer}.


\subsection{Strong-Linearizability}
Consider any execution consisting of operations on this data structure.
We linearize the operations as follows:
\begin{itemize}
\item
An \codestyle{Insert} operation that obtains \codestyle{false} when it reads \codestyle{TakeDone} on \Cref{b-bounded:insert:reread TakeDone} is linearized at this read. It returns FULL. 
\item
Consider an \codestyle{Insert(x)} operation, $ins$, that allocated 
location $a$ 
to \codestyle{m} on \Cref{b-bounded:insert:choose index}
and wrote $x$ to 
\codestyle{Items[a]}
on \Cref{b-bounded:insert:write item}.
Suppose that, while it is poised to write to \codestyle{InsertDone} on \Cref{b-bounded:insert:write InsertDone},
some \codestyle{Take} operation, $tk$, performs a successful \codestyle{t\&s}() on \codestyle{TS[a]} and then some (possibly different) \codestyle{Take} operation writes to \codestyle{TakeDone}.
In this case, $ins$ is linearized immediately before $tk$,
$ins$ returns OK, and we say that $ins$ and $tk$ are \emph{coupled}.

We prove in \Cref{b-lemma: single t&s between Insert's writes to Items and InsertDone} below that at most one 
successful \codestyle{t\&s}() 
can be performed
on \codestyle{TS[a]} 
after $ins$ writes to
\codestyle{Items[a]}
and before it writes to \codestyle{InsertDone}. Hence, an \codestyle{Insert} operation is coupled with at most one \codestyle{Take} operation.
\item
Suppose that an \codestyle{Insert} operation, $ins$, writes to \codestyle{InsertDone} on \Cref{b-bounded:insert:write InsertDone}, but has not already been linearized as a coupled operation.
Then, $ins$ is linearized when it performs this write and returns OK, and we
say that $ins$ is \emph{uncoupled}.
\item
Consider a \codestyle{Take} operation, $tk$, that performs a successful \codestyle{t\&s}() on \codestyle{TS[a]} on \Cref{b-bounded:take:t&s}. 
It is linearized at the first among the following events to occur after the successful \codestyle{t\&s}() by $tk$:
a write to \codestyle{TakeDone} by any \codestyle{Take} operation on 
\Cref{b-bounded:take:successful writes TakeDone} or \Cref{b-bounded:take:failing writes TakeDone},
a read of 1 from \codestyle{TS[a]} on \Cref{b-bounded:insert:check TS},
and a write to \codestyle{InsertDone} on \Cref{b-bounded:insert:write InsertDone} by an uncoupled \codestyle{Insert} that wrote to \codestyle{Items[$a'$]} for some $a' \neq a$.
\begin{remove}
If no \codestyle{Insert} operation reads 1 from \codestyle{TS[a]} on \Cref{b-bounded:insert:check TS} and no uncoupled \codestyle{Insert} \galy{that wrote to \codestyle{Items[$a'$]} for some} $a' \neq a$ writes to \codestyle{InsertDone} on \Cref{b-bounded:insert:write InsertDone}
between when $tk$ performed its successful \codestyle{t\&s}()
and the next write to \codestyle{TakeDone} (by any \codestyle{Take} operation) on 
\Cref{b-bounded:take:successful writes TakeDone} or \Cref{b-bounded:take:failing writes TakeDone}, then
$tk$ is linearized at this write to \codestyle{TakeDone}.
Otherwise, $tk$ is linearized at the first step following the successful \codestyle{t\&s}() by $tk$ at which
a read of 1 from \codestyle{TS[a]} on \Cref{b-bounded:insert:check TS} occurs or an uncoupled \codestyle{Insert} \galy{that wrote to \codestyle{Items[$a'$]} for some} $a' \neq a$ writes to \codestyle{InsertDone} on \Cref{b-bounded:insert:write InsertDone}.
\end{remove}
Multiple successful \codestyle{Take} operations linearized at the same step are ordered arbitrarily and before any
other
operations linearized at this step, with one exception: a coupled \codestyle{Insert} is linearized right before its coupled \codestyle{Take} operation.
In all cases, $tk$ returns the value it last read on \Cref{b-bounded:take:read item}.
\begin{remove}
\item
If after it has performed \Cref{b-bounded:take:t&s}, before any \codestyle{Take} operation writes to \codestyle{TakeDone}, some \codestyle{Insert} operation reads 1 from \codestyle{TS[a]}, or some uncoupled \codestyle{Insert} operation that wrote to \codestyle{Items[b]} for some $b \neq a$ writes to \codestyle{InsertDone}, $tk$ is linearized at the first such event. 
In case it is linearized at a write to \codestyle{InsertDone} by an uncoupled \codestyle{Insert} operation and other operations are linearized at this write as well, successful \codestyle{Take} operations including $tk$ are linearized first, and if multiple successful \codestyle{Take} operations are linearized at this write to \codestyle{InsertDone}, the order among them (including $tk$) is arbitrary.
In any case, $tk$ returns the value it read in its last execution of \Cref{b-bounded:take:read item}.
\item
If after a \codestyle{Take} operation, $tk$, performs a successful \codestyle{t\&s}() on \codestyle{TS[a]}, some \codestyle{Take} operation writes to \codestyle{TakeDone}, and no \codestyle{Insert} operation reads 1 from \codestyle{TS[a]} and no uncoupled \codestyle{Insert} operation writes to \codestyle{InsertDone} in between, then $tk$ is linearized on the first such write to \codestyle{TakeDone}. 
If multiple successful \codestyle{Take} operations are linearized at the same write to \codestyle{TakeDone}, the order among them is arbitrary. 
$tk$ returns the value it read in its last execution of \Cref{b-bounded:take:read item}. 
\end{remove}
\item
Consider a \codestyle{Take} operation, $tk$, that obtains \codestyle{false} when it reads \codestyle{InsertDone} on \Cref{b-bounded:take:reread InsertDone}. 
If no uncoupled \codestyle{Insert} operation writes to \codestyle{InsertDone} while $tk$ is poised to write to \codestyle{TakeDone}, then $tk$ is linearized at its write to \codestyle{TakeDone}
(after any other operations linearized at this step).
Otherwise, $tk$ is linearized at the first such write to \codestyle{InsertDone}, before the \codestyle{Insert} and after any successful \codestyle{Take} operations linearized at this step. The ordering among unsuccessful \codestyle{Take} operations linearized at this write is arbitrary. In both cases, $tk$ returns EMPTY.
\begin{remove}
\item
Consider a \codestyle{Take} operation, $tk$, that obtains \codestyle{false} when it reads \codestyle{InsertDone} on \Cref{b-bounded:take:reread InsertDone}. If an 
uncoupled
\codestyle{Insert} operation writes to \codestyle{InsertDone} 
while $tk$ is poised to write to \codestyle{TakeDone}, 
$tk$ is linearized at the write to \codestyle{InsertDone} by the first such \codestyle{Insert}. It is linearized before the \codestyle{Insert}, and after successful \codestyle{Take} operations if any such operations are linearized at this step. If multiple failing \codestyle{Take} operations are linearized at the same write to \codestyle{InsertDone}, the order among them is arbitrary. 
$tk$ returns EMPTY.
\item
A \codestyle{Take} operation, $tk$, that obtains \codestyle{false} when it reads \codestyle{InsertDone} on \Cref{b-bounded:take:reread InsertDone} and then writes to \codestyle{TakeDone}, is linearized at its write to \codestyle{TakeDone} if no 
uncoupled 
\codestyle{Insert} operation writes to \codestyle{InsertDone} 
when $tk$ is poised to write to \codestyle{TakeDone}.
If other \codestyle{Take} operations are linearized at the same write, they are linearized before $tk$.
$tk$ returns EMPTY.
\end{remove}
\end{itemize}

To make the linearization points clearer, we 
also list the types of steps at which operations are linearized. For each, we specify the operations that can be linearized there. If 
multiple operations can be
linearized at the same step, we specify the ordering of these operations.

\begin{itemize}
    \item An \codestyle{Insert} reads 1 from \codestyle{TS[m]} on \Cref{b-bounded:insert:check TS}:
    \begin{itemize}
        \item A \textbf{successful \codestyle{Take}} that previously performed a successful \codestyle{t\&s} on \codestyle{TS[m]}.
    \end{itemize}
    \item A \textbf{successful uncoupled \codestyle{Insert}}, $ins$, which wrote to \codestyle{Items[m]}, writes to \codestyle{InsertDone} on \Cref{b-bounded:insert:write InsertDone}:
    \begin{itemize}
        \item \textbf{Successful \codestyle{Take}s} that previously performed a successful \codestyle{t\&s} on \codestyle{TS[m']} for $m' \neq m$, arbitrarily ordered.
        \item \textbf{Unsuccessful \codestyle{Take}s}, arbitrarily ordered.
        \item $ins$.
    \end{itemize}
    \item An \textbf{unsuccessful \codestyle{Insert}}, $ins$, obtains false from \codestyle{TakeDone.dRead}() on \Cref{b-bounded:insert:reread TakeDone}:
    \begin{itemize}
        \item $ins$. 
    \end{itemize}
    \item A \textbf{successful \codestyle{Take}} writes to \codestyle{TakeDone} on \Cref{b-bounded:take:successful writes TakeDone}:
    \begin{itemize}
        \item \textbf{Successful \codestyle{Take}s}, arbitrarily ordered.
        \begin{itemize}
            \item 
            \textbf{Successful coupled \codestyle{Insert}s}, each
            immediately before its coupled \codestyle{Take}.
        \end{itemize}
    \end{itemize}
    \item An \textbf{unsuccessful \codestyle{Take}} $tk$ writes to \codestyle{TakeDone} on \Cref{b-bounded:take:failing writes TakeDone}:
    \begin{itemize}
        \item \textbf{Successful \codestyle{Take}s}, arbitrarily ordered.
        \begin{itemize}
            \item 
            \textbf{Successful coupled \codestyle{Insert}s}, each
            immediately before its coupled \codestyle{Take}.
        \end{itemize}
        \item $tk$. 
    \end{itemize}
\end{itemize}

Next, we provide 
some intuition explaining our choice of 
linearization points, as well as 
why unsuccessful \codestyle{Take} operations write to \codestyle{TakeDone}
on \Cref{b-bounded:take:failing writes TakeDone} before returning EMPTY.

We pick the linearization point for an \codestyle{Insert} operation as we do
in \Cref{sec:1-bounded bag}:
a successful uncoupled Insert is linearized at its write to \codestyle{InsertDone},
a successful coupled \codestyle{Insert} operation is linearized right before its coupled \codestyle{Take} operation,
and an unsuccessful \codestyle{Insert} operation is linearized 
when it obtains
false from \codestyle{TakeDone.dRead}() on \Cref{b-bounded:insert:reread TakeDone}. 

We want to linearize a successful \codestyle{Take} operation at its
write to \codestyle{TakeDone}, 
which is intended to notify the producer that the bag is not full, similarly to how a
successful uncoupled \codestyle{Insert} operation notifies the consumers that the bag is not empty.
However, in the following four situations, 
we linearize a successful \codestyle{Take} operation earlier than its write to \codestyle{TakeDone}.

Suppose a successful \codestyle{Take} operation, $tk$, has performed a successful \codestyle{t\&s} operation on \Cref{b-bounded:take:t&s}, but has not yet been linearized when another successful \codestyle{Take} operation writes to \codestyle{TakeDone}. Then $tk$ is linearized at this write. Linearizing $tk$ at the earliest possible write to \codestyle{TakeDone} in this case is done to simplify the linearization.

An unsuccessful \codestyle{Take} operation, $tk$, cannot be linearized when it obtains false from \codestyle{InsertDone.\-dRead}(), because there might be successful \codestyle{Take} operations that remove elements from the bag, but are not yet linearized. These operations must be linearized before $tk$.
To ensure this, $tk$ writes to \codestyle{TakeDone} before returning EMPTY. This linearizes every  successful 
\codestyle{Take} operation that has performed a successful \codestyle{t\&s} operation on \Cref{b-bounded:take:t&s}, but 
has not been linearized before the write to \codestyle{TakeDone} by $tk$.
In this case, $tk$ is linearized at this write, after all such successful \codestyle{Take} operations.

A coupled \codestyle{Take} operation witnesses that an \codestyle{Insert} operation has happened and causes it to be linearized. Similarly, an \codestyle{Insert} operation, $ins$, that witnesses a successful \codestyle{Take} operation causes it to be linearized: If $ins$
reads 1 from \codestyle{TS[a]} on \Cref{b-bounded:insert:check TS}, but the  
\codestyle{Take} 
that last performed a successful \codestyle{TS[a].t\&s}() is not yet linearized,
then the \codestyle{Take} 
is linearized at this \codestyle{read}, which is before 
$ins$
removes $a$  from \codestyle{alloc} on \Cref{b-bounded:insert:remove from alloc}.
Then, when $ins$
sees that \codestyle{|alloc| < $b$} on \Cref{b-bounded:insert:alloc<b}, the bag is not full and 
it may insert a new element.


Suppose an unsuccessful \codestyle{Take} operation, $tk$, got false from \codestyle{InsertDone.dR\-ead}(), but has not yet written to \codestyle{TakeDone} when an uncoupled \codestyle{Insert} operation, $ins$, writes to \codestyle{InsertDone} on \Cref{b-bounded:insert:write InsertDone}.
Then $tk$ must be linearized
while the bag is still empty. Hence $tk$ is linearized before $ins$, at this \codestyle{dWrite}.
Each successful \codestyle{Take} operation that performed a successful \codestyle{t\&s} on \Cref{b-bounded:take:t&s},
but has not yet been linearized, is also linearized at this step, before $tk$, to
ensure
the bag is empty when $tk$ is linearized.
There is one exception:
If $ins$ wrote to \codestyle{Items[a]} on \Cref{b-bounded:insert:write item},
a \codestyle{Take} operation that 
performed a successful \codestyle{TS[a].t\&s}() at the same location after this write
should be linearized after $ins$.
A \codestyle{Take} operation that 
performed a successful \codestyle{TS[a].t\&s}() at the same location before this write
is guaranteed to be linearized before $ins$.
This is because it is linearized at or before the producer last read 1 from \codestyle{TS[a]} on 
\Cref{b-bounded:insert:check TS},
which occurs before the producer removes location $a$ from \codestyle{alloc} on \Cref{b-bounded:insert:remove from alloc},
which, in turn, must occur before $ins$ is allocated location $a$ on \Cref{b-bounded:insert:choose index}.

\medskip
We proceed to prove that the algorithm in \Cref{fig: b-bounded bag single producer} is a strongly-linearizable implementation of a $b$-bounded bag with a single producer.

\begin{remove}
\faith{This paragraph needs rewriting and perhaps needs to be moved.}
We later bring \Cref{b-lemma: single t&s between Insert's writes to Items and InsertDone} to show that the linearization point of a coupled \codestyle{Insert} operation is well defined,
and \Cref{b-lemma: takes on same location are linearized separately} to show that the linearization point of a successful \codestyle{Take} operation at the read of 1 from \codestyle{TS[a]} by an \codestyle{Insert} operation is well defined, in the sense that there is only one successful \codestyle{Take} linearized at such a read. 
\end{remove}


Fix an execution consisting of possibly concurrent operations in \Cref{fig: b-bounded bag single producer}.
Each operation $op$ in the execution is linearized either when a step by $op$ occurs, or when a step by another operation occurs between two steps of $op$. 
Thus, the linearization point defined above for every operation is within its interval.
Each operation is linearized at the same point in every extension of the execution, hence, if the implementation is linearizable, it is strongly-linearizable.

We now prove a number of useful facts about the implementation.

\codestyle{Items} is only modified by the producer. It uses its local set \codestyle{alloc} to keep track of the locations
in \codestyle{Items} that store elements.

\begin{lemma}\label{b-invariant: items < alloc}
If \codestyle{Items[$a$]} $\neq \bot$, then $a \in$ \codestyle{alloc}.
\end{lemma}

\begin{proof}
Before the producer writes an element $x$ to \codestyle{Items[m]} on \Cref{b-bounded:insert:write item},
it adds \codestyle{m} to \codestyle{alloc} on \Cref{b-bounded:insert:add to alloc}.
Immediately after the producer writes $\bot$ to
\codestyle{Items[m]} on \Cref{b-bounded:insert:clear item}, it removes \codestyle{m} from \codestyle{alloc} on \Cref{b-bounded:insert:remove from alloc}.
\end{proof}


When a location in \codestyle{TS} is 
reset on \Cref{b-bounded:insert:reset TS},
this location is in \codestyle{used}, which is a set
that is local to the producer.
We show that this implies that 
the corresponding location in
\codestyle{Items} contains $\bot$.

\begin{lemma}\label{b-invariant: bot while in used}
    If $a \in$ \codestyle{used}, then \codestyle{Items[$a$]} = $\bot$.
\end{lemma}
\begin{proof}
    The producer adds the location \codestyle{m} to \codestyle{used} on \Cref{b-bounded:insert:add to used}, 
    after  it sets \codestyle{Items[m]} to $\bot$ on \Cref{b-bounded:insert:clear item}. 
When the producer writes an element $x$ to \codestyle{Items[m]} on \Cref{b-bounded:insert:write item},
\codestyle{m} is not in \codestyle{used}.
This is because \codestyle{used} is a subset of \codestyle{hazardous}, from \Cref{b-bounded:insert:narrow used down},
and \codestyle{m} s not in \codestyle{hazardous}, from \Cref{b-bounded:insert:choose index}.
\end{proof}

When a \codestyle{Take} operation successfully returns an element, we can identify an interval
immediately before it performed its successful \codestyle{t\&s}() in which this test\&set object was 0
and the corresponding location in \codestyle{Items} contained the returned value.

\begin{lemma}
\label{b-lemma using hazards}
\label{b-corollary to lemma using hazards}
Suppose $tk$ is a \codestyle{Take} operation that 
performs a successful \codestyle{t\&s}() of \codestyle{TS[$a$]} on \Cref{b-bounded:take:t&s}.
Let $x$ be the element that $tk$ last read from \codestyle{Items[$a$]}
on \Cref{b-bounded:take:read item}.
Then \codestyle{TS[$a$]} = 0 and \codestyle{Items[$a$]} = $x$ between
when the producer last wrote to \codestyle{Items[$a$]} prior to the read
and when $tk$ performed  its successful \codestyle{t\&s}().
\end{lemma}

\begin{proof}
Let $i$ be the index of the process that performs $tk$.
Prior to reading $x$ from \codestyle{Items[$a$]},  $tk$ wrote $a$ to \codestyle{Hazards[$i$]} on \Cref{b-bounded:take:add hazard}.
It next writes to \codestyle{Hazards[$i$]} on \Cref{b-bounded:take:clear hazard1}.
So \codestyle{Hazards[$i$]} = $a$ between 
the step $s_1$
 when $tk$ last performed \Cref{b-bounded:take:read item}
and 
the step $s_2$
when it last performed  \Cref{b-bounded:take:t&s}.

Note that \codestyle{TS[$a$]} = 0 immediately before $tk$ performed its successful \codestyle{t\&s}() on \codestyle{TS[$a$]}.
Let $s_0$ be the step when the producer last wrote to \codestyle{Items[$a$]} (on \Cref{b-bounded:insert:write item}) prior to $s_1$.
So, \codestyle{Items[$a$]} = $x$ between $s_0$ and $s_1$.
Assume, for the sake of contradiction, that \codestyle{TS[$a$]} = 1  at some configuration between $s_0$ and $s_2$.
Then between this configuration and $s_2$, the producer
 must have reset \codestyle{TS[$a$]} on \Cref{b-bounded:insert:reset TS}.
From \Cref{b-bounded:insert:reset loop}, $a$ was in \codestyle{used} when the reset was performed, so, by \Cref{b-invariant: bot while in used}, \codestyle{Items[$a$]} = $\bot$ when the reset was performed. 
Thus, the producer must have written $\bot$ to \codestyle{Items[$a$]} on \Cref{b-bounded:insert:clear item} before it performed the reset.
Since \codestyle{Items[$a$]} $=x$ between $s_0$ and $s_1$,
the producer must have written $\bot$ to \codestyle{Items[$a$]}, collected \codestyle{Hazards} on \Cref{b-bounded:insert:collect hazards},
and performed the reset between  $s_1$ and $s_2$.
In particular,  the producer reads $a$ from \codestyle{Hazards[i]}.
But, by \Cref{b-bounded:insert:reset loop}, the producer does not reset \codestyle{TS[$a$]} on \Cref{b-bounded:insert:reset TS}.
This is a contradiction. Thus \codestyle{TS[$a$]} = 0 between $s_0$ and $s_2$.

By the test on \Cref{b-bounded:insert:check TS}, the producer does not write $\bot$ to \codestyle{Items[$a$]} on 
\Cref{b-bounded:insert:clear item} and $a$ is not removed from \codestyle{alloc} on \Cref{b-bounded:insert:remove from alloc}
between $s_0$ and $s_2$. Thus, the producer does not choose the index $a$ again on \Cref{b-bounded:insert:choose index}
and does not write to \codestyle{Items[$a$]} on \Cref{b-bounded:insert:write item} between these two steps.
Hence, \codestyle{Items[$a$]} = $x$ between
when the producer last wrote to \codestyle{Items[$a$]} at step $s_0$
and when $tk$ performed  its successful \codestyle{t\&s}() at step $s_2$.
\end{proof}

\begin{remove}
First, we show that there are up to $b$ non-$\bot$ values in \codestyle{Items}, and they are all found in locations that appear in \codestyle{alloc}.

\begin{lemma}\label{b-invariant: items < alloc}
The locations of non-$\bot$ values in \codestyle{Items} are a subset of the locations in \codestyle{alloc}.
\end{lemma}

\begin{proof}
\codestyle{alloc} is a local variable of the producer, who writes to it only on \Cref{b-bounded:insert:add to alloc,b-bounded:insert:remove from alloc}. 
A non-$\bot$ value is written to \codestyle{Items[m]} only on \Cref{b-bounded:insert:write item}, after $m$ is added to \codestyle{alloc} on \Cref{b-bounded:insert:add to alloc}.
$m$ is removed from \codestyle{alloc} on \Cref{b-bounded:insert:remove from alloc} only after \codestyle{Items[m]} is set to $\bot$ on \Cref{b-bounded:insert:clear item}.
\end{proof}

Next, we show that 
when a location in \codestyle{TS} can be
reset, the corresponding location in
\codestyle{Items} contains $\bot$.

\begin{lemma}\label{b-invariant: bot while in used}
    If $k$ $\in$ \codestyle{used}, then \codestyle{Items[k]} = $\bot$.
\end{lemma}
\begin{proof}
    The producer adds the location \codestyle{m} = $k$ to \codestyle{used} on \Cref{b-bounded:insert:add to used}, after it sets \codestyle{Items[m]} to $\bot$ on \Cref{b-bounded:insert:clear item}. 
    Thus, \codestyle{Items[k]} = $\bot$ when $k$ is added to \codestyle{used}. 
    Suppose that a non-$\bot$ value is subsequently written to \codestyle{Items[k]}.
    This occurs on \Cref{b-bounded:insert:write item}, with \codestyle{m} = $k$.
    The value of \codestyle{m} was chosen on \Cref{b-bounded:insert:choose index} to be a value
    that is not in \codestyle{hazardous}.
    Thus, if $k$ was in \codestyle{used}, it was removed from \codestyle{used} on
    \Cref{b-bounded:insert:narrow used down}, before a non-$\bot$ value was written to \codestyle{Items[m]}.
\end{proof}

When a \codestyle{Take} operation returns a non-$\bot$ value, we can identify an interval
immediately before it performed its successful \codestyle{t\&s}() in which this test\&set object was 0
and the corresponding location in \codestyle{Items} contained the returned value.

\begin{lemma}\label{b-lemma using hazards}
Suppose $tk$ is a \codestyle{Take} operation that 
performs a successful \codestyle{t\&s}() of \codestyle{TS[a]} on \Cref{b-bounded:take:t&s}.
Let $D_2$ be the configuration right before this step, let
$D_1$ be the configuration right before $tk$ reads \codestyle{Items[a]} 
on \Cref{b-bounded:take:read item}
for the last time, and let $x$ be the value that $tk$ read from \codestyle{Items[a]}.
Let $D_0$ be the configuration right after the producer last wrote to \codestyle{Items[a]} on \Cref{b-bounded:insert:write item} before $D_1$.
Then \codestyle{TS[a]} = 0 and \codestyle{Items[a]} = $x$ between $D_0$ and $D_2$.
\end{lemma}
\begin{proof}
Let $P_i$ be the process that performed $tk$.
 Prior to reading $x$ from \codestyle{Items[a]} right after $D_1$, process $P_i$ wrote $a$ to \codestyle{Hazards[i]} on \Cref{b-bounded:take:add hazard}.
     It next writes to \codestyle{Hazards[i]} on \Cref{b-bounded:take:clear hazard1}, which occurs after $D_2$. So, between $D_1$ and $D_2$,
     \codestyle{Hazards[i]} = $a$.
         Since $tk$ reads $x$ from \codestyle{Items[a]} immediately after $D_1$
    and the producer does not write to \codestyle{Items[a]} on \Cref{b-bounded:insert:write item} between $D_0$ and $D_1$, 
    \codestyle{Items[a]} = $x$ throughout this interval.
    
  Right after $D_2$, a \codestyle{t\&s}() of \codestyle{TS[a]} succeeds, so 
     \codestyle{TS[a]} = 0 at $D_2$.
    Assume, for sake of contradiction, that \codestyle{TS[a]} = 1 at some point between $D_0$ and $D_2$. Then, after this point, but before $tk$
    performed its successful \codestyle{t\&s}() on \codestyle{TS[a]}, the producer
    must have reset \codestyle{TS[a]} on \Cref{b-bounded:insert:reset TS}.
    According to \Cref{b-bounded:insert:reset loop}, $a$ was in \codestyle{used} when the reset was performed, so by \Cref{b-invariant: bot while in used}, \codestyle{Items[a]} = $\bot$ at that point. Thus, between  $D_1$ and $D_2$, the producer must have written $\bot$ to \codestyle{Items[a]} on \Cref{b-bounded:insert:clear item} and then performed the reset.
    Between these two steps, the producer collects
    \codestyle{Hazards} and reads $a$ from \codestyle{Hazards[i]}.
But, by \Cref{b-bounded:insert:reset loop}, the producer does not reset \codestyle{TS[a]} on \Cref{b-bounded:insert:reset TS}.
This is a contradiction. Thus \codestyle{TS[a]} = 0 between $D_0$ and $D_2$.

When writing $x$ to \codestyle{Items[a]} just before $D_0$, $a$ is in \codestyle{alloc} after it was added to \codestyle{alloc} on \Cref{b-bounded:insert:add to alloc}. To perform another write to \codestyle{Items[a]} on \Cref{b-bounded:insert:write item}, the producer needs to choose the index $a$ again on \Cref{b-bounded:insert:choose index}, which requires that $a$ $\notin$ \codestyle{alloc}. Before that, $a$ must be removed from \codestyle{alloc} on \Cref{b-bounded:insert:remove from alloc}. 
Hence, after writing $x$ to \codestyle{Items[a]} just before $D_0$,
    the producer next writes to \codestyle{Items[a]}
    on \Cref{b-bounded:insert:clear item}.
    Between these two writes, the producer must read 1 from  \codestyle{TS[a]} on \Cref{b-bounded:insert:check TS}.
But \codestyle{TS[a]} = 0 between $D_0$ and $D_2$.
    Thus, \codestyle{Items[a]} = $x$ between $D_0$ and $D_2$. 
\end{proof}

It easily follows that,
when a successful \codestyle{t\&s}() is performed on a location in \codestyle{TS}, the corresponding location in \codestyle{Items} does not contain  $\bot$.

\begin{corollary}
\label{b-corollary to lemma using hazards}
When a \codestyle{Take} operation performs a successful \codestyle{t\&s}() on \codestyle{TS[a]},
\codestyle{Items[a]} $\neq \bot$.
\end{corollary}
\end{remove}

Next, we show that when the producer writes an element 
into some location in \codestyle{Items}, 
that element is available to 
be taken by a taker.

\begin{lemma}\label{b-lemma: TS=0 on write to Items}
When the producer writes an element  to \codestyle{Items[m]} on
\Cref{b-bounded:insert:write item}, 
\codestyle{TS[m]} = 0.
\end{lemma}

\begin{remove}
\begin{proof}
Let $e$ be a step of the execution in which the producer performs \Cref{b-bounded:insert:write item}. Let $a$ denote the value of \codestyle{m} at this step.

Suppose that, prior to $e$, the producer has not performed \Cref{b-bounded:insert:write item} with \codestyle{m} = $a$.
Since all locations of \codestyle{Items} initially contain $\bot$ and non-$\bot$ values are only written
into \codestyle{Items} on \Cref{b-bounded:insert:write item},
\codestyle{Items[a]} = $\bot$ prior to $e$.
Since \codestyle{t\&s}() on \codestyle{TS[a]} is only performed on \Cref{b-bounded:take:t&s} after
seeing that \codestyle{Items[a]} $\neq \bot$, no \codestyle{t\&s}() on \codestyle{TS[a]} was performed prior to $e$.
Since \codestyle{TS[a]} is initially 0, it is 0 when $e$ is performed and we are done.

Otherwise, let $e^*$ be the last step prior to $e$ in which the producer performed \Cref{b-bounded:insert:write item} with \codestyle{m} = $a$.
This step is performed when $a$ is in \codestyle{alloc} after it was added to \codestyle{alloc} on \Cref{b-bounded:insert:add to alloc}. When the producer performs \Cref{b-bounded:insert:choose index} for the last time before $e$, it chooses the index $a$ which is not in \codestyle{alloc} at that point. This means that $a$ was earlier removed from \codestyle{alloc} after $e^*$. Indexes are removed from \codestyle{alloc} only on \Cref{b-bounded:insert:remove from alloc}. Thus, the producer executed \Cref{b-bounded:insert:remove from alloc} 
and then added $a$ to \codestyle{used} on \Cref{b-bounded:insert:add to used} 
between $e^*$ and $e$.

By \Cref{b-invariant: bot while in used}, \codestyle{Items[a]} = $\bot$ while
$a$ was in \codestyle{used}. Immediately after step $e$, \codestyle{Items[a]} $\neq \bot$,
so $a$ is not in \codestyle{used}.
Prior to being removed from \codestyle{used} on \Cref{b-bounded:insert:narrow used down},
$a$ was not in \codestyle{hazardous}, so \codestyle{TS[a]} was reset to 0 on \Cref{b-bounded:insert:reset TS} by some step $e'$.

By \Cref{b-invariant: bot while in used}, \codestyle{Items[a]} = $\bot$ at $e'$.
Between $e'$ and $e$, there was no write to \codestyle{Items[a]} on \Cref{b-bounded:insert:write item}, so \codestyle{Items[a]} = $\bot$.
When a successful \codestyle{t\&s}() on  \codestyle{TS[a]} is performed,
\Cref{b-corollary to lemma using hazards}
implies that \codestyle{Items[a]} $\neq$ $\bot$.
Thus, no successful \codestyle{t\&s}() on  \codestyle{TS[a]} was performed between $e'$ and $e$.
This implies that \codestyle{TS[a]} = 0 when $e$ was performed.
\end{proof}
\end{remove}

\begin{proof}
Consider a step $s$ of the execution in which the producer performs a write to \codestyle{TS[$a$]} on
 \Cref{b-bounded:insert:write item}.
 All locations  of \codestyle{Items} initially contain $\bot$ and elements are only written
into \codestyle{Items} on \Cref{b-bounded:insert:write item}.
Thus, if $s$ is the first step in which the producer performs a write to \codestyle{TS[$a$]} on
 \Cref{b-bounded:insert:write item}, then \codestyle{Items[a]} = $\bot$ prior to $s$.

Suppose that, prior to $e$, the producer has not performed \Cref{b-bounded:insert:write item} with \codestyle{m} = $a$.
All locations of \codestyle{Items} initially contain $\bot$, non-$\bot$ values are only written
into \codestyle{Items} on \Cref{b-bounded:insert:write item},
and  \codestyle{t\&s}() on \codestyle{TS[a]} is only performed on \Cref{b-bounded:take:t&s} after a taker
sees that \codestyle{Items[a]} $\neq \bot$, no \codestyle{t\&s}() on \codestyle{TS[a]} was performed prior to $s$.
Since \codestyle{TS[a]} is initially 0, it is 0 when $s$ is performed.

Otherwise, let $s'$ be the last step prior to $s$ at which the producer performed \Cref{b-bounded:insert:write item} with \codestyle{m} = $a$.
Then $a$ was added to \codestyle{alloc} when the producer last performed  \Cref{b-bounded:insert:add to alloc} prior to $s'$.
When the producer last performs \Cref{b-bounded:insert:choose index} prior to  $s$, the location $a$ was not in \codestyle{alloc}.
Thus, between these two steps, the producer removed $a$  from \codestyle{alloc}  on \Cref{b-bounded:insert:remove from alloc}
and added $a$ to \codestyle{used} on \Cref{b-bounded:insert:add to used}.

By \Cref{b-invariant: bot while in used}, \codestyle{Items[a]} = $\bot$ while
$a \in$  \codestyle{used}. Immediately after step $s$, \codestyle{Items[$a$]} $\neq \bot$,
so $a \not\in$ \codestyle{used}.
Prior to being removed from \codestyle{used} on \Cref{b-bounded:insert:narrow used down},
$a$ was not in \codestyle{hazardous},
 so \codestyle{TS[a]} was reset to 0 on \Cref{b-bounded:insert:reset TS}.
Since $a \in$  \codestyle{used} when this step occurred, \codestyle{Items[$a$]} = $\bot$.
Between this step and $s$,  there is no write to \codestyle{Items[$a$]} on \Cref{b-bounded:insert:write item},
by definition of $s'$, so \codestyle{Items[$a$]} = $\bot$.
When a successful \codestyle{t\&s}() on  \codestyle{TS[$a$]} is performed,
\Cref{b-corollary to lemma using hazards}
implies that \codestyle{Items[a]} $\neq$ $\bot$.
Thus, no successful \codestyle{t\&s}() on  \codestyle{TS[$a$]} was performed between these two steps.
Hence,  \codestyle{TS[$a$]} = 0 when $s$ was performed.
\end{proof}

The next lemma says that an \codestyle{Insert} operation is coupled with at most one \codestyle{Take} operation.
\begin{lemma}\label{b-lemma: single t&s between Insert's writes to Items and InsertDone}
Consider an \codestyle{Insert(x)} operation, $ins$, that allocated location $a$ to \codestyle{m} on \Cref{b-bounded:insert:choose index}.
At most one successful \codestyle{t\&s}() on \codestyle{TS[$a$]} by a \codestyle{Take} operation is performed while $ins$ is poised to write to \codestyle{InsertDone} on \Cref{b-bounded:insert:write InsertDone}.
\end{lemma}
\begin{proof}
By \Cref{b-lemma: TS=0 on write to Items}, when $ins$ writes to \codestyle{Items[a]} on \Cref{b-bounded:insert:write item}, \codestyle{TS[a]} = 0. Assume some \codestyle{Take} operation performs a successful \codestyle{t\&s}() on \codestyle{TS[a]} after $ins$ performs this write and before it performs its next step, which is a write to \codestyle{InsertDone} on \Cref{b-bounded:insert:write InsertDone}. For another \codestyle{Take} operation to perform a successful \codestyle{t\&s}() on \codestyle{TS[a]}, \codestyle{TS[a]} must be first reset to 0. It is reset only by the 
producer (on \Cref{b-bounded:insert:reset TS}).
Therefore, 
no other \codestyle{Take} operation can perform a successful \codestyle{t\&s}() on \codestyle{TS[a]} while $ins$ is poised to write to \codestyle{InsertDone}. 
\end{proof}

Between a successful \codestyle{t\&s}() and a reset of \codestyle{TS[a]},
the producer reads 1 from \codestyle{TS[a]}.
\begin{lemma}\label{b-lemma: producer reads 1 from TS between t&s and reset}
After a \codestyle{Take} operation performs a successful \codestyle{TS[a].t\&s}()
the producer must read 1 from \codestyle{TS[a]}
before it resets \codestyle{TS[a]}.
\end{lemma}
\begin{proof}
Let $ts$ be a successful \codestyle{TS[$a$].t\&s}() performed by a \codestyle{Take} operation
on \Cref{b-bounded:take:t&s}.  
By \Cref{b-corollary to lemma using hazards}, when $ts$ is performed, \codestyle{Items[$a$]} $\neq \bot$ and, hence, by \Cref{b-invariant: bot while in used}, $a \notin$ \codestyle{used}.
When a reset of \codestyle{TS[$a$]} is performed on \Cref{b-bounded:insert:reset TS},
the condition on \Cref{b-bounded:insert:reset loop} ensures that $a \in$ \codestyle{used}.
So $a$ is added to \codestyle{used} on \Cref{b-bounded:insert:add to used} between $ts$ and the next reset of \codestyle{TS[$a$]}.
Before $a$ is added to \codestyle{used}, the producer reads 1 from \codestyle{TS[a]} on \Cref{b-bounded:insert:check TS}.

Only the producer performs resets, but does not perform any resets between performing  
\Cref{b-bounded:insert:check TS} and next performing \Cref{b-bounded:insert:add to used}.
Since \codestyle{TS[a]} $= 0$ when $ts$ is performed, $ts$ does not occur between these two steps by the producer. Hence, the producer reads 1 from \codestyle{TS[a]} after $ts$ and before it
next resets \codestyle{TS[$a$]}.
\end{proof}

Two successful \codestyle{Take} operations that return elements from the same location 
in \codestyle{Items} are linearized at different steps.
\begin{lemma}\label{b-lemma: takes on same location are linearized separately}
Two successful \codestyle{Take} operations that both perform a successful \codestyle{t\&s}() on \codestyle{TS[$a$]} are not linearized at the same step.
\end{lemma}
\begin{proof}
Suppose that $tk_1$ is a \codestyle{Take} operation that performs a successful \codestyle{t\&s}(), $ts_1$, on \codestyle{TS[$a$]} before the \codestyle{Take} operation $tk_2$ performs a successful \codestyle{t\&s}(), $ts_2$, on \codestyle{TS[$a$]}.
A reset of \codestyle{TS[$a$]} occurs between $ts_1$ and $ts_2$. 
By \Cref{b-lemma: producer reads 1 from TS between t&s and reset}, 
the producer reads 1 from \codestyle{TS[a]} (on \Cref{b-bounded:insert:check TS})
between $ts_1$ and performing this reset. Note that $tk_1$ is linearized at or before this read.
However, $tk_2$ is linearized after $ts_2$.
Thus, they are not linearized at the same step.
\end{proof}

The following lemma says that, after an element is written to \codestyle{Items[$a$]},
the next \codestyle{Take} operation to perform a successful \codestyle{t\&s}()
on \codestyle{TS[$a$]} returns this element.

\begin{lemma}\label{b-lemma: item is taken by the first t&s after its write}
Consider 
an \codestyle{Insert}($x_1$) operation that writes $x_1$ to \codestyle{Items[$a$]} on \Cref{b-bounded:insert:write item}.
The \codestyle{Take} operation that next performs a successful \codestyle{TS[$a$].t\&s()} returns $x_1$.
\end{lemma}
\begin{proof}
Let $tk$ be the \codestyle{Take} operation that performs the first successful \codestyle{TS[a].t\&s()}, $ts$, following the write, $w_1$, of $x_1$ to \codestyle{Items[$a$]}.
Let $r$ be the last read of \codestyle{Items[a]} on \Cref{b-bounded:take:read item} 
performed by $tk$ prior to $ts$ and let
$x$ be the element that it read during $r$. 
Let $w$ be the last write to \codestyle{Items[a]} on \Cref{b-bounded:insert:write item} by the producer before $r$, which implies this is a write of $x$. 
Between $w$ and the producer's next write to \codestyle{Items[a]}, the producer executes \Cref{b-bounded:insert:choose index}, and at that point $a \notin$ \codestyle{alloc}, which implies by \Cref{b-invariant: items < alloc} that \codestyle{Items[a]} = $\bot$.
But by \Cref{b-lemma using hazards}, \codestyle{Items[a]} = $x$ between $w$ and $ts$. Therefore, the producer's next write to \codestyle{Items[a]} after $w$ must happen after $ts$. In other words, $w$ was the last write to \codestyle{Items[a]} before $ts$, namely, $w = w_1$. Hence, $x = x_1$. So $tk$ would return $x_1$.
\end{proof}

To show that the ordering of operations that we defined is a linearization, we prove that the values returned by the operations are consistent with the sequential specifications of a $b$-bounded bag.
We do this by induction on the length of the execution.

Let $S$ be the ordering of a prefix, $\sigma$, of the execution, let $e$ be the next step of the execution, and $S'$ be the ordering
of $\sigma\cdot e$. 
Suppose $S$ is a linearization. 
Let $B$ be the set of items in the bag at the end of $S$. 
Further suppose that for every \codestyle{Take} operation $tk$ that appears in $S$, returns an item $x$, and performs in $\sigma$ a successful \codestyle{t\&s}() on \codestyle{TS[a]} for some index $a$, 
an \codestyle{Insert}($x$) operation that wrote in $\sigma$ to \codestyle{Items[a]} appears in $S$ before $tk$ and no \codestyle{Take} operation returning $x$ appear in $S$ between this \codestyle{Insert} operation and $tk$. 
Due to our no repetitions assumption, this informally means that every \codestyle{Take} operation that performs in $\sigma$ a successful \codestyle{t\&s}() on some location in \codestyle{TS}, takes the item that was inserted to the bag by an \codestyle{Insert} operation that wrote to the same location in \codestyle{Items}.
Let $B$ be the set of items in the bag at the end of $S$. 
We will prove that $S'$ is a linearization, and that for every \codestyle{Take} operation $tk$ that appears in $S'$ but not in $S$, returns an item $x$ and performs in $\sigma$ a successful \codestyle{t\&s}() on \codestyle{TS[a]} for some index $a$, an \codestyle{Insert}($x$) operation that wrote in $\sigma$ to \codestyle{Items[a]} appears in $S$, such that no \codestyle{Take} operation returning $x$ appears in $S$ after it. The latter part is proven in \Cref{b-lemma: no take is linearized before take after last insert writing to the same location}.

It suffices to consider the steps at which operations are linearized, namely when a \codestyle{Take} operation writes to \codestyle{TakeDone} and when 
an \codestyle{Insert} operation reads \codestyle{TakeDone}, writes to \codestyle{InsertDone} or reads 1 from \codestyle{TS}.

In what follows we will use the following lemmas.
The first one states that between two successful \codestyle{Insert} operations that write to the same location, a single \codestyle{Take} operation that takes from that location is linearized and it returns the item inserted by the first \codestyle{Insert} among the two.
\begin{lemma}\label{b-lemma: take between 2 inserts to the same location}
Let $ins_1$ and $ins_2$ be two successful \codestyle{Insert} operations
that were allocated the same location $a$. Assume $ins_1$ was the first among the two to be executed by the producer. Then a \codestyle{Take} operation $tk$ that performs a successful \codestyle{TS[a].t\&s}()
and returns the item inserted by $ins_1$
is linearized 
after $ins_1$.
If $ins_2$ is linearized, then it is linearized after $tk$.
\end{lemma}
\begin{proof}
Let $ins_2'$ be the first successful \codestyle{Insert} operation to be executed after $ins_1$ among those allocated location $a$. 
Let $w_1$ be the write event to \codestyle{Items[a]} by $ins_1$ on \Cref{b-bounded:insert:write item}, and $w_2$ be the write event to \codestyle{Items[a]} by $ins_2'$.
Let $w_{InsertDone}$ be the write event to \codestyle{InsertDone} by $ins_1$ on \Cref{b-bounded:insert:write InsertDone}, which is its next step after writing to \codestyle{Items[a]}.
From the code, it is evident that $a \in$ \codestyle{alloc} when $w_{InsertDone}$ occurs.
Before $ins_2'$ writes to \codestyle{Items[a]}, it allocates location $a$ to \codestyle{m} on \Cref{b-bounded:insert:choose index}, and at that point $a \notin$ \codestyle{alloc}. Therefore, between $w_{InsertDone}$ and $w_2$, $m$ is removed from \codestyle{alloc}. This is done by the producer on \Cref{b-bounded:insert:remove from alloc}, after reading 1 from \codestyle{TS[a]} on \Cref{b-bounded:insert:check TS}. 
Let $r_{TS}$ be this read event. Then $r_{TS}$ occurs between $w_{InsertDone}$ and $w_2$.
By \Cref{b-lemma: TS=0 on write to Items}, \codestyle{TS[a]} = 0 at $w_1$. Thus, between $w_1$ and $r_{TS}$, a \codestyle{Take} operation performed a successful \codestyle{TS[a].t\&s}(). Let $tk$ be the \codestyle{Take} operation to perform the first such \codestyle{t\&s}(), and let $ts$ be this \codestyle{t\&s}() event. 
$tk$ returns the item inserted by $ins_1$ due to \Cref{b-lemma: item is taken by the first t&s after its write}. 

If $ins_2'$ is linearized, it is linearized after $w_2$. $tk$ is linearized at the latest at $r_{TS}$. Therefore, $tk$ is linearized 
before $ins_2'$.
To show that $tk$ is linearized after $ins_1$, we consider two cases. 
First consider the case in which when $ins_1$ is poised to write to \codestyle{InsertDone} on \Cref{b-bounded:insert:write InsertDone}, $ts$ occurs and then some \codestyle{Take} operation writes to \codestyle{TakeDone}. Then $ins_1$ is coupled with $tk$ and linearized before it. 
Otherwise, $ins_1$ is linearized at $w_{InsertDone}$, and $tk$ is linearized at one of the following events, all occurring after $w_{InsertDone}$: a read of 1 from \codestyle{TS[a]} by the producer on \Cref{b-bounded:insert:check TS}, a write to \codestyle{InsertDone} by an uncoupled \codestyle{Insert} operation that is allocated an index $\neq a$, or a write to \codestyle{TakeDone}.

\end{proof}

The following lemma clarifies the linearization order among a successful \codestyle{Take} operation and a successful \codestyle{Insert} operation that operate on the same location.
\begin{lemma}\label{b-lemma: take performing t&s before insert writes to the same location is linearized first}
Let $tk$ be a \codestyle{Take} operation that performs a successful \codestyle{TS[a].t\&s}(), in an event $ts$. Let $ins$ be an \codestyle{Insert} operation that writes to \codestyle{Items[a]} on \Cref{b-bounded:insert:write item}, in an event $w$ that occurs after $ts$. 
Then $tk$ is linearized before $ins$.
\end{lemma}
\begin{proof}
By \Cref{b-lemma: TS=0 on write to Items}, \codestyle{TS[a]} = 0 at $w$. Hence, between $ts$ and $w$, a reset of \codestyle{TS[a]} on \Cref{b-bounded:insert:reset TS} must take place. 
By \Cref{b-lemma: producer reads 1 from TS between t&s and reset}, between $ts$ and this reset, the producer reads 1 from \codestyle{TS[a]} on \Cref{b-bounded:insert:check TS}.
$tk$ is linearized at this read if it is not already linearized earlier. Therefore, $tk$ is linearized before $w$, and thus, before $ins$.
\end{proof}

The following lemma states that from the successful \codestyle{t\&s}() a \codestyle{Take} operation performs on \codestyle{TS[a]} until this operation is linearized, no \codestyle{Take} operation performs a successful \codestyle{TS[a].t\&s}().
\begin{lemma}\label{b-lemma: no other successful t&s from take's t&s until its linearization step}
Let $tk$ be a \codestyle{Take} operation that appears in $S'$ but not in $S$ and performs in $\sigma$ a successful \codestyle{t\&s}() on \codestyle{TS[a]} for some index $a$. 
Then no successful \codestyle{TS[a].t\&s}() occurs in $\sigma$ after $tk$'s successful \codestyle{TS[a].t\&s}().
\end{lemma}
\begin{proof}
Let $ts$ be $tk$'s successful \codestyle{TS[a].t\&s}().
Before another successful \codestyle{TS[a].t\&s}() occurs in $\sigma$ after $ts$, a reset of \codestyle{TS[a]} must be executed. 
By \Cref{b-lemma: producer reads 1 from TS between t&s and reset}, between $ts$ and this reset, the producer reads 1 from \codestyle{TS[a]} on \Cref{b-bounded:insert:check TS}.
But no such read is executed in $\sigma$ after $ts$, because if it were, $tk$ would be linearized at this step at the latest and would appear in $S$.
 is executed in $\sigma$ after $ts$.
We have that a successful \codestyle{TS[a].t\&s}() that follows $ts$ must occur after a reset, which in turn occurs after read of 1 from \codestyle{TS[a]}, which does not occur in $\sigma$ after $ts$. So such a successful \codestyle{t\&s}() does not occur in $\sigma$.
\end{proof}

The following lemma proves the second part of the induction statement. After this lemma we prove the first part, namely, that $S'$ is a linearization.
\begin{lemma}\label{b-lemma: no take is linearized before take after last insert writing to the same location}
Let $tk$ be a \codestyle{Take} operation that appears in $S'$ but not in $S$, returns an item $x$ and performs in $\sigma$ a successful \codestyle{t\&s}() on \codestyle{TS[a]} for some index $a$. Let $ins$ be the last \codestyle{Insert} operation to write to \codestyle{Items[a]} in $\sigma$ before $tk$ read it for the last time (which it did before executing its successful \codestyle{t\&s}()).
If $ins$ appears in $S$, then no \codestyle{Take} operation that returns $x$ appears after it in $S$.
\end{lemma}
\begin{proof}
Let $ts$ be $tk$'s successful \codestyle{TS[a].t\&s}().
Let $w$ be the write to \codestyle{Items[a]} on \Cref{b-bounded:insert:write item} by $ins$.
We first consider \codestyle{Take} operations in $S$ that perform in $\sigma$ a successful \codestyle{t\&s}() on \codestyle{TS[a]}, if such operations exist. 
No successful \codestyle{t\&s}() is performed on \codestyle{TS[a]} between $w$ and $ts$, as \codestyle{TS[a]} = 0 in this interval by \Cref{b-lemma using hazards}. No successful \codestyle{t\&s}() is performed on \codestyle{TS[a]} in $\sigma$ after $ts$ by \Cref{b-lemma: no other successful t&s from take's t&s until its linearization step}.
Any \codestyle{Take} operation performing a successful \codestyle{TS[a].t\&s}() before $w$, appears in $S$ before $ins$ by \Cref{b-lemma: take performing t&s before insert writes to the same location is linearized first}. 

Now, we consider \codestyle{Take} operations in $S$ returning $x$, which perform in $\sigma$ a successful \codestyle{t\&s}() on \codestyle{TS[a']} for some index $a' \neq a$. We will show that if any such operations exist, they must appear before $ins$ in $S$.
Assume for sake of contradiction that such an operation appears after $ins$ in $S$. Let $tk'$ be the first such operation to appear in $S$ after $ins$. Assume it performed a successful \codestyle{t\&s}() on \codestyle{TS[a']} for some location $a'$.
Due to the induction hypothesis, an \codestyle{Insert}($x$) operation that wrote in $\sigma$ to \codestyle{Items[a']} appears in $S$ before $tk'$, such that $tk'$ is the first \codestyle{Take} operation returning $x$ to appear after this \codestyle{Insert}($x$). This \codestyle{Insert} operation is not $ins$, because they write to different locations in \codestyle{Items}.
$tk'$ is also the first \codestyle{Take} operation returning $x$ to appear in $S$ after $ins$, due to the definition of $tk'$ in addition to what we earlier proved---that \codestyle{Take} operations that perform in $\sigma$ a successful \codestyle{t\&s}() on \codestyle{TS[a]} do not appear in $S$ after $ins$. Hence, right before $tk'$, $x$ appears at least twice in the bag. This is a contradiction to our no repetitions assumption.
\end{proof}

\medskip

First, suppose $e$ is an execution of \codestyle{TakeDone.dRead}() on \Cref{b-bounded:insert:reread TakeDone} by an \codestyle{Insert} operation $ins$, which returns false.
Then $S'$ is $S$ followed by $ins$ with return value FULL. 
To show that $S'$ is a linearization, we must show that $|B| = b$.

\begin{lemma}\label{b-lemma: at most b locations in alloc}
$|alloc| \leq b$.
\end{lemma}
\begin{proof}
\codestyle{alloc} is initialized to $\phi$. A location is added to \codestyle{alloc} only on \Cref{b-bounded:insert:add to alloc}, when $|alloc| < b$.
\end{proof}

Let $C$ be the configuration at the end of $\sigma$,
let $C^*$ be the configuration immediately after the last previous read of \codestyle{TakeDone} by $ins$ on either \Cref{b-bounded:insert:read TakeDone} or \Cref{b-bounded:insert:reread TakeDone},
let $S^*$ be the linearization of the prefix of $\sigma$ ending with $C^*$, and
let $B^*$ be the set of items in the bag at the end of $S^*$.

$|alloc| = b$ when $ins$ executes \Cref{b-bounded:insert:alloc<b} after $C^*$: $|alloc|$ cannot be smaller than $b$ at that step because $ins$ continues to the \codestyle{else if} clause, and it cannot be bigger than $b$ due to \Cref{b-lemma: at most b locations in alloc}.
$ins$ does not execute \Cref{b-bounded:insert:remove from alloc} between $C^*$ and $C$, because remove a location from \codestyle{alloc} would imply that $|alloc| > b$ before the removal, which is impossible by \Cref{b-lemma: at most b locations in alloc}.
So $ins$ does not modify $alloc$ between $C^*$ and $C$. Let $alloc^*$ be the value of $alloc$ in this interval. We have $|alloc^*| = b$.
This implies that $ins$ obtains 0 from the read of \codestyle{TS[m]} for all indices $m \in alloc$ in its executions of \Cref{b-bounded:insert:check TS} following $C^*$.

As $ins$ gets false from \codestyle{TakeDone.dRead()} at $e$, no write to \codestyle{TakeDone} is performed between $C^*$ and $C$. This implies that no successful \codestyle{Take} operations are linearized between these configuration, as these are linearized either at a write to \codestyle{TakeDone} or at steps by the producer which are not carried out in this interval (reading 1 from \codestyle{TS[m]} on \Cref{b-bounded:insert:check TS} or writing to \codestyle{InsertDone} on \Cref{b-bounded:insert:write InsertDone}).
So to prove that $|B| = b$, it is enough to show that $|B^*| = b$.

Let $m$ be a location in $alloc^*$. We will show that the item that is in \codestyle{Items[m]} at $C^*$ is in $B^*$. Since $|alloc^*| = b$ and due to our no repetitions assumption, this will imply $|B^*| = b$.
Let $ins'$ be the last \codestyle{Insert} operation to add $m$ to \codestyle{alloc} on \Cref{b-bounded:insert:add to alloc} before $C^*$. $ins'$ is successful so $ins' \neq ins$. Hence, $ins'$ was completed before $C^*$, and appears in $S^*$ as a successful insertion. 
We will show that a \codestyle{Take} operation that takes the item inserted by $ins'$ cannot appear in $S^*$ after $ins'$.

Let $w$ be the write to \codestyle{Items[m]} on \Cref{b-bounded:insert:write item} by $ins'$.
At $w$, $m \notin$ \codestyle{hazardous} according to \Cref{b-bounded:insert:choose index}, while \codestyle{used} $\subseteq$ \codestyle{hazardous} by \Cref{b-bounded:insert:narrow used down}.
Therefore, $m$ $\notin$ \codestyle{used} at $w$.
After $w$, $m$ is added to \codestyle{used} on \Cref{b-bounded:insert:add to used} only after it is removed from \codestyle{alloc} on \Cref{b-bounded:insert:remove from alloc}. The latter may happen only after $C$: it cannot happen until $C^*$ by definition of $ins'$, and it cannot happen between $C^*$ until $C$ since \codestyle{alloc} remains unchanged in this interval.
Thus, $m$ $\notin$ \codestyle{used} continues to hold from $w$ until $C$.
Therefore, no reset of \codestyle{TS[m]} is executed on \Cref{b-bounded:insert:reset TS} between $w$ and $C$. 
Recall that $ins'$ reads \codestyle{TS[m]} between $C^*$ and $C$ and obtains 0. As \codestyle{TS[m]} is not reset between $w$ and this read, a successful \codestyle{TS[m].t\&s}() cannot be executed within this interval, specifically not between $w$ and $C^*$.
A \codestyle{Take} operation that executes a successful \codestyle{TS[m].t\&s}() before $w$ is not linearized after $ins'$ by \Cref{b-lemma: take performing t&s before insert writes to the same location is linearized first}, and a \codestyle{Take} operation that executes a successful \codestyle{TS[m].t\&s}() after $C^*$ is not linearized at $S^*$. Therefore, no \codestyle{Take} operation is linearized after $ins'$ in $S^*$.

\medskip

Next, suppose $e$ is an execution of a read of 1 from \codestyle{TS[a]} by an \codestyle{Insert} operation $ins$ on \Cref{b-bounded:insert:check TS}. 
Further suppose a \codestyle{Take} operation, $tk$, performed a successful \codestyle{t\&s}() on \codestyle{TS[a]} on \Cref{b-bounded:take:t&s}, in a step $ts$. Assume between $ts$ and $e$, none of the following occurred: another read of 1 from \codestyle{TS[a]} by the producer, a write to \codestyle{TakeDone}, or a write to \codestyle{Items[a']} for $a' \neq a$ by an uncoupled \codestyle{Insert} operation. There could not be more than one such \codestyle{Take} operation by \Cref{b-lemma: takes on same location are linearized separately}.
Let the item $tk$ obtained from \codestyle{Items[a]} in its last execution of \Cref{b-bounded:take:read item} be $x$.
Then $S'$ is $S$ followed by $tk$ with return value $x$. 
To show that $S'$ is a linearization, we must show that $x \in B$.

Let $ins'$ be the last \codestyle{Insert} operation to write to \codestyle{Items[a]} on \Cref{b-bounded:insert:write item} before $tk$ reads it for the last time (which it does before $ts$), and let $w$ be this write. Then $w$ is a write of $x$ and $ins'$ in an \codestyle{Insert($x$)} operation. $ins' \neq ins$ since $ins'$ does not execute \Cref{b-bounded:insert:check TS} after $w$. Therefore, $ins'$ is completed in $\sigma$ and thus appears in $S$. Since $S$ is a linearization, in the prefix of $S$ that ends with $ins'$, the bag contains $x$. 
By \Cref{b-lemma: no take is linearized before take after last insert writing to the same location}, there is no \codestyle{Take} operation that returns $x$ in $S$ after $ins'$. This implies $x \in B$.

\medskip

Now, suppose $e$ is an execution of \codestyle{InsertDone.dWrite}() by an uncoupled \codestyle{Insert}($x$) operation $ins$ on \Cref{b-bounded:insert:write InsertDone}. Let $m$ be the index allocated to $ins$ when it executed \Cref{b-bounded:insert:choose index}. 
Let $tk_i$ for $1 \leq i \leq s$ denote every \codestyle{Take} operation that performed in $\sigma$ a successful \codestyle{t\&s}() on \codestyle{TS[$a_i$]} where $a_i \neq m$ such that since the \codestyle{t\&s}() until the end of $\sigma$, no write to \codestyle{TakeDone} was executed, no write to \codestyle{InsertDone} by an uncoupled \codestyle{Insert} operation was executed and no read of 1 from \codestyle{TS[$a_i$]} by an \codestyle{Insert} operation was executed.
Let the item $tk_i$ obtained from \codestyle{Items[$a_i$]} in its last execution of \Cref{b-bounded:take:read item} be $x_i$.
Let $f$ be the number of \codestyle{Take} operations that obtained false from \codestyle{InsertDone.dRead}() on \Cref{b-bounded:take:reread InsertDone} prior to $e$, and between the read and $e$ there was no write to \codestyle{InsertDone} and they have not yet executed \codestyle{TakeDone.dWrite}().
$S'$ is $S$ followed by $tk_1,\ldots,tk_s$, with return value $x_1,\ldots,x_s$ respectively, in an arbitrary order, then the above-mentioned $f$ failing \codestyle{Take} operations in an arbitrary order, all with return value EMPTY, and then $ins$ with return value OK. 
To show that $S'$ is a linearization, we must show that (i) $\{x_1,\ldots,x_s\} \subseteq B$, 
(ii) $B=\{x_1,\ldots,x_s\}$ if $f>0$,
and (iii) $|B|<b$ if $s=f=0$. Due to the no repetitions assumption, (ii) would imply $|B|=s$.

For each $1 \leq i \leq s$, let $ins_i$ be the \codestyle{Insert} operation that performed the last write to \codestyle{Items[$a_i$]} before $tk_i$ read it for the last time (which it did before executing its successful \codestyle{t\&s}()). Hence, $ins_i$ in an \codestyle{Insert}($x_i$) operation.
We begin with proving (i).
Let $i$ be an index such that $1 \leq i \leq s$.
$ins \neq ins_i$ since $ins$ was allocated the index $m$ while $ins_i$ was allocated the index $a_i \neq m$.
Since $e$ is a step performed by $ins$ and there is a single producer, $ins_i$ was completed in $\sigma$.
Hence $ins_i$ appears in $S$. Since $S$ is a linearization, in the prefix of $S$ that ends with $ins_i$, the bag contains $x_i$. 
By \Cref{b-lemma: no take is linearized before take after last insert writing to the same location}, there is no \codestyle{Take} operation after $ins_i$ in $S$ that takes $x_i$. This implies $x_i \in B$.

\smallskip

We proceed to prove (ii).
Let $tk$ be one of the $f$ \codestyle{Take} operations that obtained false from \codestyle{InsertDone.dRead}() on \Cref{b-bounded:take:reread InsertDone} prior to $e$ and have not yet executed \codestyle{TakeDone.dWrite}(). 
We will prove that no items other than \codestyle{$x_i$} for $1 \leq i \leq s$ are in $B$. 

Let $C$ be the configuration at the end of $\sigma$ (immediately before $e$),
let $\widehat C$ be the configuration immediately after $tk$ obtained false from \codestyle{InsertDone.dRead}() on \Cref{b-bounded:take:reread InsertDone},
let $C^*$ be the configuration immediately after the last previous read of \codestyle{InsertDone} by $tk$ on either \Cref{b-bounded:take:read InsertDone} or \Cref{b-bounded:take:reread InsertDone},
let $S^*$ be the linearization of the prefix of $\sigma$ ending with $C^*$, and
let $B^*$ be the set of items in the bag at the end of $S^*$.

Considering $tk$ obtained false from \codestyle{InsertDone.dRead}() right before $\widehat C$, 
the producer does not write to \codestyle{InsertDone} 
between $C^*$ and $\widehat C$.
No write to \codestyle{InsertDone} is performed by an uncoupled \codestyle{Insert} operation between $\widehat C$ and $C$ as $e$ is the first such write since $\widehat C$.

Suppose there was a successful \codestyle{Insert} operation that is in $S$, but not in $S^*$, which means it is linearized between configurations $C^*$ and $C$.
Since no uncoupled \codestyle{Insert} operation wrote to \codestyle{InsertDone} on \Cref{b-bounded:insert:write InsertDone} between $C^*$ and $C$,
this \codestyle{Insert} operation was coupled with a \codestyle{Take} operation. Hence it was immediately
taken from the bag after it was inserted into the bag.
Therefore, any item in $B$ was inserted to the bag in $S^*$ and was in $B^*$.

So, suppose there is $x^* \in B^*$ such that $x^* \neq x_i$ for all $1 \leq i \leq s$.
Items inserted by coupled \codestyle{Insert} operations are immediately taken from the bag by a \codestyle{Take} operation linearized at the same step, so they are not in the bag in the end of the linearization of any prefix of the execution.
Hence they are not in $B^*$.
Therefore, the item $x^*$ was inserted into the bag by an uncoupled \codestyle{Insert} operation, $ins^*$. (Each item in the bag may be uniquely associated with the \codestyle{Insert} operation that inserted it thanks to our no repetitions assumption.)
This operation wrote $x^*$ to \codestyle{Items} in some location, $a^*$, on \Cref{b-bounded:insert:write item}, and then wrote to \codestyle{InsertDone} on \Cref{b-bounded:insert:write InsertDone} before $C^*$. 
We will consider three distinct cases: (1) $a^* = a_i$ for some $1 \leq i \leq s$, (2) $a^* = m$, and (3) $a^* \neq a_i$ for all $1 \leq i \leq s$ and $a^* \neq m$. For each case, we will show that a \codestyle{Take} operation returning $x^*$ appears in $S$ after $ins^*$.

We start with case (1). $tk_i$ read $x_i \neq x^*$ from \codestyle{Items[$a^*$]} on \Cref{b-bounded:take:read item}. We showed in the proof of (i) that $ins_i$ appears in $S$.
If $ins_i$ appears in $S$ before $ins^*$, then by \Cref{b-lemma: take between 2 inserts to the same location}, a \codestyle{Take} operation that performed a successful \codestyle{TS[$a^*$].t\&s}() appears in $S$ after $ins_i$.
But this is impossible: in the end of the above proof of (i), we showed that no such \codestyle{Take} operations appear in $S$ after $ins_i$.
Therefore, $ins^*$ must appear before $ins_i$ in $S$. 
By \Cref{b-lemma: take between 2 inserts to the same location}, a \codestyle{Take} operation that takes $x^*$ appears in $S$ after $ins^*$, and we are done with case (1).

For case (2),
$ins^* \neq ins$ since the first appears in $S$ and the latter does not.
By \Cref{b-lemma: take between 2 inserts to the same location}, a \codestyle{Take} operation that takes $x^*$ appears in $S$ after $ins^*$.

We proceed to case (3).
When $ins^*$ wrote to \codestyle{Items[$a^*$]}, \codestyle{TS[$a^*$]} = 0 by \Cref{b-lemma: TS=0 on write to Items}. We consider several sub-cases.


(3a) If $a^*$ was removed from \codestyle{Allocated} on \Cref{b-bounded:insert:add to Allocated} before $tk$ read \codestyle{Allocated} on \Cref{b-bounded:take:read Allocated} after $C^*$,
then before this read the producer has removed $a^*$ from $alloc$ on \Cref{b-bounded:insert:remove from alloc}, which means the producer passed \Cref{b-bounded:insert:check TS} after $ins^*$ was completed.
(3b) If $tk$ read $a^*$ from \codestyle{Allocated} on \Cref{b-bounded:take:read Allocated} after $C^*$ and then read $\bot$ from \codestyle{Items[$a^*$]} on \Cref{b-bounded:take:read item},
then the producer wrote $\bot$ to \codestyle{Items[$a^*$]} on \Cref{b-bounded:insert:clear item} after $ins^*$ was completed, and before this write the producer passed \Cref{b-bounded:insert:check TS}.
So in both cases---(3a) and (3b), the producer read 1 from \codestyle{TS[$a^*$]} on \Cref{b-bounded:insert:check TS} after $ins^*$ was completed.
Since \codestyle{TS[$a^*$]} = 0 when $ins^*$ wrote to \codestyle{Items[$a^*$]}, 
some \codestyle{Take} operation performed a successful \codestyle{t\&s}() on \codestyle{TS[$a^*$]} after $ins^*$ wrote to \codestyle{Items[$a^*$]} and before $\widehat C$.
Let $tk^*$ be the first \codestyle{Take} to perform such a \codestyle{t\&s}().
It appears in $S$ because it is linearized at the latest at the read of 1 from \codestyle{TS[$a^*$]} by the producer. 

(3c) Otherwise, after $C^*$, $tk$ read $a^*$ from \codestyle{Allocated} and then read a non-$\bot$ value from \codestyle{Items[$a^*$]}.
In this case, $tk$ must have performed an unsuccessful \codestyle{t\&s}() of \codestyle{TS[$a^*$]} on \Cref{b-bounded:take:t&s}. As \codestyle{TS[$a^*$]} = 0 when $ins^*$ wrote to \codestyle{Items[$a^*$]}, another \codestyle{Take} operation performed a successful \codestyle{TS[$a^*$].t\&s}() after $ins^*$ wrote to \codestyle{Items[$a^*$]} and before the unsuccessful \codestyle{t\&s}() by $tk$.
Let $tk^*$ be the first \codestyle{Take} to perform such a \codestyle{t\&s}(). It appears in $S$ because otherwise it would be one of $tk_i$ for some $1 \leq i \leq s$, 
but it is not as it operates on location $a^*\neq a_i$ for all $1 \leq i \leq s$.

In all cases---(3a), (3b) and (3c), $tk^*$ takes $x^*$ from the bag by \Cref{b-lemma: item is taken by the first t&s after its write}. 
Additionally, $tk^*$ appears in $S$ after $ins^*$: $tk^*$ is linearized after $ins^*$ wrote to \codestyle{Items[$a^*$]} because $tk^*$ performs its successful \codestyle{TS[$a^*$].t\&s}() after that write. Since $ins^*$ is uncoupled, $tk^*$ cannot be linearized before the write to \codestyle{InsertDone} by $ins^*$, which is the linearization point of $ins^*$.

\smallskip
We proceed to prove (iii). Assume $s=f=0$.
We use the following lemma, which states that the items that are in the bag in the end of the linearization of any prefix of $\sigma$ are a subset of the non-$\bot$ items in \codestyle{Items} at the end of the prefix.

\begin{lemma}\label{b-lemma: B < Items}
For any prefix $\tau$ of $\sigma$, the items in the bag in the end of the linearization of $\tau$ is a subset of the non-$\bot$ items that are in \codestyle{Items} at the end of $\tau$.
\end{lemma}
\begin{proof}
We will show that when an item is added to the bag, it appears in some location in \codestyle{Items}; and that before it is removed from this location in \codestyle{Items} it is removed from the bag.

An item inserted to the bag by a coupled \codestyle{Insert} operation does not appear in the bag in the end of the linearization of any prefix of $\sigma$, since it is inserted and then taken at the same step.
Thus, an item $x$ becomes a part of the bag of the linearization in a step $w_{InsertDone}$ in which a successful uncoupled \codestyle{Insert}($x$) operation $ins'$ writes to \codestyle{InsertDone} and is linearized. 
Let $a$ be the location that $ins'$ is allocated.
Let $w$ be the write of $x$ to \codestyle{Items[a]} by $ins'$ on \Cref{b-bounded:insert:write item}, which happens before $w_{InsertDone}$.
\codestyle{Items[a]} still contains $x$ in $w_{InsertDone}$, because after $w$, the producer writes to \codestyle{Items[a]} only in following \codestyle{Insert} operations, after $ins'$ is completed.

We proceed to the second part, showing that if $x$ is removed from \codestyle{Items}, it is removed from the bag beforehand.
If \codestyle{Items[a]} = $x$ since $w_{InsertDone}$ until the end of $\sigma$, we are done.
Else, let $w_\bot$ be the step after which \codestyle{Items[a]} $\neq x$ for the first time since $w_{InsertDone}$.
When the producer writes $x$ to \codestyle{Items[a]} at $w$, $a$ is in \codestyle{alloc} after it was added to \codestyle{alloc} on \Cref{b-bounded:insert:add to alloc}. To perform another write to \codestyle{Items[a]} on \Cref{b-bounded:insert:write item}, the producer needs to choose the index $a$ again on \Cref{b-bounded:insert:choose index}, which requires that $a$ $\notin$ \codestyle{alloc}. Before that, $a$ must be removed from \codestyle{alloc} on \Cref{b-bounded:insert:remove from alloc}. 
Hence, after writing $x$ to \codestyle{Items[a]} at $w$,
the producer next writes to \codestyle{Items[a]}
on \Cref{b-bounded:insert:clear item}.
Namely, $w_\bot$ is a write of $\bot$ to \codestyle{Items[a]} by the producer.
We will show that a \codestyle{Take} operation takes $x$ from the bag of the linearization, after $ins'$ inserts it to the bag of the linearization, in a step that precedes $w_\bot$.
Let $r_{TS}$ be the producer's step preceding $w_\bot$, in which it executes \Cref{b-bounded:insert:check TS}. \codestyle{TS[a]} = 1 at $r_{TS}$.
By \Cref{b-lemma: TS=0 on write to Items}, \codestyle{TS[a]} = 0 at $w$. 
So a successful \codestyle{TS[a].t\&s}() is performed by a \codestyle{Take} operation between $w$ and $r_{TS}$. Let $tk$ be the \codestyle{Take} operation that performs the first such \codestyle{t\&s}, in a step $ts$.
$tk$ is linearized at the latest at $r_{TS}$, hence it is linearized before $w_\bot$.
It takes $x$ from the bag by \Cref{b-lemma: item is taken by the first t&s after its write}. 
It remains to show that $tk$ is linearized after $ins'$. 
$tk$ is linearized after $w$ because $tk$ is linearized after $ts$ which happens after $w$. Since $ins'$ is uncoupled, $tk$ cannot be linearized between $w$ and $w_{InsertDone}$. Therefore, $tk$ is linearized after $w_{InsertDone}$, which is the linearization point of $ins'$.

\end{proof}

Let $C$ be the configuration that appears in $\sigma$ right before the addition of $m$ to \codestyle{alloc} by $ins$ on \Cref{b-bounded:insert:add to alloc}. 
Let $B_C$ be the set of items in the bag at the end of the linearization of $\sigma$'s prefix that ends with $C$.
By \Cref{b-lemma: B < Items}, $|B_C|$ is less than or equal to the number of non-$\bot$ items that are in \codestyle{Items} at $C$. 
By \Cref{b-invariant: items < alloc}, the locations in \codestyle{Items} with non-$\bot$ values at $C$ are a subset of the locations in \codestyle{alloc} at the same configuration, and there are less than $b$ locations in \codestyle{alloc} according to the check on \Cref{b-bounded:insert:alloc<b}. Hence, $|B_C| < b$. 
No \codestyle{Insert} operations are linearized in $\sigma$ after $C$, since the producer is executing the uncoupled \codestyle{Insert} $ins$ in this interval, which is linearized only at $e$ that comes right after $\sigma$. Therefore, for any prefix of $\sigma$ that extends the prefix ending with $C$, the bag in the end of the linearization of this prefix does not have more items than $B_C$, and we are done.

\medskip

Finally, suppose $e$ is a write to \codestyle{TakeDone} by a \codestyle{Take}() operation $tk$ on \Cref{b-bounded:take:successful writes TakeDone} or \Cref{b-bounded:take:failing writes TakeDone}. 
Let $tk_i$ for $1 \leq i \leq s$ denote every \codestyle{Take} operation that performed in $\sigma$ a successful \codestyle{t\&s}() on \codestyle{TS[$a_i$]} such that no \codestyle{Insert} operation has written to \codestyle{Items[$a_i$]} before this \codestyle{t\&s}() but has not yet written to \codestyle{InsertDone} at $e$.
Let the item $tk_i$ obtained from \codestyle{Items[$a_i$]} in its last execution of \Cref{b-bounded:take:read item} be $x_i$.
Let $tk'$ be a \codestyle{Take} operation that performed in $\sigma$ a successful \codestyle{t\&s}() on \codestyle{TS[$a'$]} such that an \codestyle{Insert($x'$)} operation, $ins'$, has written to \codestyle{Items[$a'$]} before this \codestyle{t\&s}() and has not yet written to \codestyle{InsertDone} at $e$, if such $tk'$ and $ins'$ exist.
Then $S'$ is $S$ followed by $tk_1,\ldots,tk_s$ and $tk'$ if it exists, with return value $x_1,\ldots,x_s,x'$ respectively, in an arbitrary order, with $ins'$ placed right before $tk'$ if they exist, followed by $tk$ with return value EMPTY if $e$ is a write on \Cref{b-bounded:take:failing writes TakeDone} (otherwise, $tk$ is one of the $s$ successful \codestyle{Take} operations). 
To show that $S'$ is a linearization, we must show that
(i) $\{x_1,\ldots,x_s\} \subseteq B$, 
(ii) $B=\{x_1,\ldots,x_s\}$ if $e$ is a write on \Cref{b-bounded:take:failing writes TakeDone},
and (iii) $|B|<b$ if $tk'$ exists. Due to the no repetitions assumption, (ii) would imply $|B|=s$.

For each $1 \leq i \leq s$, let $ins_i$ be the \codestyle{Insert} operation that performed the last write to \codestyle{Items[$a_i$]} before $tk_i$ read it for the last time (which it did before executing its successful \codestyle{t\&s}()). Hence, $ins_i$ in an \codestyle{Insert}($x_i$) operation.
We begin with proving (i).
Let $i$ be an index such that $1 \leq i \leq s$.
$ins_i$ wrote to \codestyle{InsertDone} on \Cref{b-bounded:insert:write InsertDone} in $\sigma$ by definition of $tk_i$. Thus, $ins_i$ is linearized in $\sigma$, namely, it appears in $S$. Since $S$ is a linearization, in the prefix of $S$ that ends with $ins_i$, the bag contains $x_i$. 
By \Cref{b-lemma: no take is linearized before take after last insert writing to the same location}, there is no \codestyle{Take} operation after $ins_i$ in $S$ that takes $x_i$. This implies $x_i \in B$.

\smallskip

We proceed to prove (ii). We assume $e$ is a write on \Cref{b-bounded:take:failing writes TakeDone} in this case, hence, $tk$ is a failing \codestyle{Take} operation.
We will prove that no items other than \codestyle{$x_i$} for $1 \leq i \leq s$ are in $B$. 

Let $C$ be the configuration at the end of $\sigma$ (immediately before $e$),
let $\widehat C$ be the configuration immediately after $tk$ obtained false from \codestyle{InsertDone.dRead}() on \Cref{b-bounded:take:reread InsertDone},
let $C^*$ be the configuration immediately after the last previous read of \codestyle{InsertDone} by $tk$ on either \Cref{b-bounded:take:read InsertDone} or \Cref{b-bounded:take:reread InsertDone},
let $S^*$ be the linearization of the prefix of $\sigma$ ending with $C^*$, and
let $B^*$ be the set of items in the bag at the end of $S^*$.

Considering $tk$ obtained false from \codestyle{InsertDone.dRead}() right before $\widehat C$, 
the producer does not write to \codestyle{InsertDone} 
between $C^*$ and $\widehat C$. No write to \codestyle{InsertDone} is performed by an uncoupled \codestyle{Insert} operation between $\widehat C$ and $C$, because otherwise $tk$ would have been linearized at that write.

Suppose there was a successful \codestyle{Insert} operation that is in $S$, but not in $S^*$, which means it is linearized between configurations $C^*$ and $C$.
Since no uncoupled \codestyle{Insert} operation wrote to \codestyle{InsertDone} on \Cref{b-bounded:insert:write InsertDone} between $C^*$ and $C$,
this \codestyle{Insert} operation was coupled with a \codestyle{Take} operation. Hence it was immediately
taken from the bag after it was inserted into the bag.
Therefore, any item in $B$ was inserted to the bag in $S^*$ and was in $B^*$.

So, suppose there is $x^* \in B^*$ such that $x^* \neq x_i$ for all $1 \leq i \leq s$.
Items inserted by coupled \codestyle{Insert} operations are immediately taken from the bag by a \codestyle{Take} operation linearized at the same step, so they are not in the bag in the end of the linearization of any prefix of the execution.
Hence they are not in $B^*$.
Therefore, the item $x^*$ was inserted into the bag by an uncoupled \codestyle{Insert} operation, $ins^*$. (Each item in the bag may be uniquely associated with the \codestyle{Insert} operation that inserted it thanks to our no repetitions assumption.)
This operation wrote $x^*$ to \codestyle{Items} in some location, $a^*$, on \Cref{b-bounded:insert:write item}, and then wrote to \codestyle{InsertDone} on \Cref{b-bounded:insert:write InsertDone} before $C^*$. 
We will consider each of the following cases separately: (1) $a^* = a_i$ for some $1 \leq i \leq s$, (2) $a^* = a'$ (relevant only when $ins'$ exists), and (3) $a^* \neq a_i$ for all $1 \leq i \leq s$ and $a^* \neq a'$. For each case, we will show that a \codestyle{Take} operation returning $x^*$ appears in $S$ after $ins^*$.

We start with case (1). $tk_i$ read $x_i \neq x^*$ from \codestyle{Items[$a^*$]} on \Cref{b-bounded:take:read item}. We showed in the proof of (i) that $ins_i$ appears in $S$.
If $ins_i$ appears in $S$ before $ins^*$, then by \Cref{b-lemma: take between 2 inserts to the same location}, a \codestyle{Take} operation that performed a successful \codestyle{TS[$a^*$].t\&s}() appears in $S$ after $ins_i$.
But this is impossible: in the end of the above proof of (i), we showed that no such \codestyle{Take} operations appear in $S$ after $ins_i$.
Therefore, $ins^*$ must appear before $ins_i$ in $S$. 
By \Cref{b-lemma: take between 2 inserts to the same location}, a \codestyle{Take} operation that takes $x^*$ appears in $S$ after $ins^*$, and we are done with case (1).

For case (2),
$ins^* \neq ins'$ since the first appears in $S$ and the latter does not.
By \Cref{b-lemma: take between 2 inserts to the same location}, a \codestyle{Take} operation that takes $x^*$ appears in $S$ after $ins^*$.

We proceed to case (3).
When $ins^*$ wrote to \codestyle{Items[$a^*$]}, \codestyle{TS[$a^*$]} = 0 by \Cref{b-lemma: TS=0 on write to Items}. We consider several sub-cases.


(3a) If $a^*$ was removed from \codestyle{Allocated} on \Cref{b-bounded:insert:add to Allocated} before $tk$ read \codestyle{Allocated} on \Cref{b-bounded:take:read Allocated} after $C^*$,
then before this read the producer has removed $a^*$ from $alloc$ on \Cref{b-bounded:insert:remove from alloc}, which means the producer passed \Cref{b-bounded:insert:check TS} after $ins^*$ was completed.
(3b) If $tk$ read $a^*$ from \codestyle{Allocated} on \Cref{b-bounded:take:read Allocated} after $C^*$ and then read $\bot$ from \codestyle{Items[$a^*$]} on \Cref{b-bounded:take:read item},
then the producer wrote $\bot$ to \codestyle{Items[$a^*$]} on \Cref{b-bounded:insert:clear item} after $ins^*$ was completed, and before this write the producer passed \Cref{b-bounded:insert:check TS}.
So in both cases---(3a) and (3b), the producer read 1 from \codestyle{TS[$a^*$]} on \Cref{b-bounded:insert:check TS} after $ins^*$ was completed.
Since \codestyle{TS[$a^*$]} = 0 when $ins^*$ wrote to \codestyle{Items[$a^*$]}, 
some \codestyle{Take} operation performed a successful \codestyle{t\&s}() on \codestyle{TS[$a^*$]} after $ins^*$ wrote to \codestyle{Items[$a^*$]} and before $\widehat C$.
Let $tk^*$ be the first \codestyle{Take} to perform such a \codestyle{t\&s}().
It appears in $S$ because it is linearized at the latest at the read of 1 from \codestyle{TS[$a^*$]} by the producer. 

(3c) Otherwise, after $C^*$, $tk$ read $a^*$ from \codestyle{Allocated} and then read a non-$\bot$ value from \codestyle{Items[$a^*$]}.
In this case, $tk$ must have performed an unsuccessful \codestyle{t\&s}() of \codestyle{TS[$a^*$]} on \Cref{b-bounded:take:t&s}. As \codestyle{TS[$a^*$]} = 0 when $ins^*$ wrote to \codestyle{Items[$a^*$]}, another \codestyle{Take} operation performed a successful \codestyle{TS[$a^*$].t\&s}() after $ins^*$ wrote to \codestyle{Items[$a^*$]} and before the unsuccessful \codestyle{t\&s}() by $tk$.
Let $tk^*$ be the first \codestyle{Take} to perform such a \codestyle{t\&s}(). It appears in $S$ because otherwise it would be one of $tk_i$ for some $1 \leq i \leq s$, 
but it is not as it operates on location $a^*\neq a_i$ for all $1 \leq i \leq s$.

In all cases---(3a), (3b) and (3c), $tk^*$ returns $x^*$ \Cref{b-lemma: item is taken by the first t&s after its write}. 
Additionally, $tk^*$ appears in $S$ after $ins^*$: $tk^*$ is linearized after $ins^*$ wrote to \codestyle{Items[$a^*$]} because $tk^*$ performs its successful \codestyle{TS[$a^*$].t\&s}() after that write. Since $ins^*$ is uncoupled, $tk^*$ cannot be linearized before the write to \codestyle{InsertDone} by $ins^*$, which is the linearization point of $ins^*$.

\smallskip
We proceed to prove (iii). Assume $tk'$ exists.
Let $C$ be the configuration that appears in $\sigma$ right before the addition of $a'$ to \codestyle{alloc} by $ins'$ on \Cref{b-bounded:insert:add to alloc}. 
Let $B_C$ be the set of items in the bag at the end of the linearization of $\sigma$'s prefix that ends with $C$.
By \Cref{b-lemma: B < Items}, $|B_C|$ is less than or equal to the number of non-$\bot$ items that are in \codestyle{Items} at $C$. 
By \Cref{b-invariant: items < alloc}, the locations in \codestyle{Items} with non-$\bot$ values at $C$ are a subset of the locations in \codestyle{alloc} at the same configuration, and there are less than $b$ locations in \codestyle{alloc} according to the check on \Cref{b-bounded:insert:alloc<b}. Hence, $|B_C| < b$. 
No \codestyle{Insert} operations are linearized in $\sigma$ after $C$, since the producer is executing the coupled \codestyle{Insert} $ins'$ in this interval, which is linearized only at $e$ that comes right after $\sigma$. Therefore, for any prefix of $\sigma$ that extends the prefix ending with $C$, the bag in the end of the linearization of this prefix does not have more items than $B_C$, and we are done.


\subsection{Lock-Freedom.} 
When a \codestyle{Take} operation finishes an iteration of the \codestyle{repeat} loop and is about to begin another iteration, 
some \codestyle{Insert} operation wrote to \codestyle{InsertDone} on \Cref{b-bounded:insert:write InsertDone} between the \codestyle{Take}'s penultimate read of \codestyle{InsertDone} on \Cref{b-bounded:take:read InsertDone} or \Cref{b-bounded:take:reread InsertDone} and its last read of \codestyle{InsertDone} on \Cref{b-bounded:take:reread InsertDone}.
This write step completes the \codestyle{Insert} operation.
When an \codestyle{Insert} operation finishes an iteration of the \codestyle{repeat} loop and is about to begin another iteration, 
some \codestyle{Take} operation wrote to \codestyle{TakeDone} on \Cref{b-bounded:take:successful writes TakeDone} or \Cref{b-bounded:take:failing writes TakeDone} between the \codestyle{Insert}'s penultimate read of \codestyle{TakeDone} on \Cref{b-bounded:insert:read TakeDone} or \Cref{b-bounded:insert:reread TakeDone} and its last read of \codestyle{TakeDone} on \Cref{b-bounded:insert:reread TakeDone}.
This write step completes the \codestyle{Take} operation.

\section{Discussion}\label{sec:discussion}

This paper explores strongly-linearizable implementations of bags 
from
interfering objects.
We presented the first lock-free, strongly-linearizable implementation of a bag from interfering objects,
which is, interestingly, also a linearizable implementation of a queue, but not a strongly-linearizable implementation.
We also presented several implementations of bounded bags with a single producer from interfering objects:
a wait-free, linearizable implementation of a 1-bounded bag, a lock-free, strongly-linearizable implementation of a 1-bounded bag, and a lock-free, strongly-linearizable implementation of a
$b$-bounded bag for any value of $b$.

A direct extension of this work is investigating
whether it is possible to extend our bounded bag implementations to support two producers
or multiple producers but only one or two consumers.
Other open questions are whether
there are 
wait-free, strongly-linearizable implementations of ABA-detecting registers and bags
from interfering objects.
It would also be interesting to construct a lock-free, strongly-linearizable implementation of a bag
from interfering objects,
such that, at every point in the execution, the number of interfering objects used 
is bounded above by a function of the number of processes and the number of elements the bag contains.
More broadly, exploring strongly-linearizable implementations of objects beyond bags using interfering objects 
is
a compelling direction for future work.

\ignore{The number of objects used by our implementation  of a lock-free, strongly-linearizable bag can easily be modified so that it is
proportional to the number of \codestyle{Insert} operations performed during the operation.
Use a linked list, where each node contains a register and a test\&set object.
An \codestyle{Insert} operations appends a new node to this list (so it is now lock-free, rather than wait-free).
We no longer need the fetch\&increment object, \codestyle{Max}, although we still keep \codestyle{Done}.  
An iteration of \codestyle{Take} traverses the linked list until it receives a NULL pointer, instead of examining \codestyle{m} 
locations. 
The linearization points don't change, since they were at test\&set operations and accesses to \codestyle{Done}.
To make the number of object proportional to the number of elements in the bag (perhaps plus the number of active processes),
processes performing \codestyle{Take} can also remove nodes from the linked list once the test\&set object
has value 1. We might need reference counting to ensure that there is no process accessing a node before it is reused. 
}


\section*{Acknowledgments}
This work was supported in part by the Natural Science and Engineering Research Council of Canada (NSERC) grant RGPIN-2020-04178.

\bibliographystyle{plainurl}
\bibliography{refs}
\end{document}